\documentclass[a4paper,reqno,12pt]{amsart}

\allowdisplaybreaks[1]

\usepackage{geometry}
\usepackage{hyperref}
\usepackage{courier}
\usepackage{times}
\usepackage{amsmath,amssymb}
\usepackage{enumitem}
\usepackage[boxed,vlined,longend]{algorithm2e}
\usepackage[foot]{amsaddr} 

\usepackage{color}
\usepackage[normalem]{ulem}

\theoremstyle{plain}
\newtheorem{proposition}{\bf Proposition}
\newtheorem{theorem}{\bf Theorem}
\newtheorem{lemma}{\bf Lemma}

\theoremstyle{definition}
\newtheorem{definition}{\bf Definition}

\newcommand{\cdc}{\texttt{CDC}}
\newcommand{\FDRM}{\texttt{FDRM}}
\newcommand{\MRD}{\texttt{MRD}}
\newcommand{\cC}{\mathcal{C}}
\newcommand{\cE}{\mathcal{E}}
\newcommand{\cF}{\mathcal{F}}
\newcommand{\cG}{\mathcal{G}}
\newcommand{\cL}{\mathcal{L}}
\newcommand{\cM}{\mathcal{M}}
\newcommand{\cS}{\mathcal{S}}
\newcommand{\cU}{\mathcal{U}}
\newcommand{\cV}{\mathcal{V}}
\newcommand{\F}{\mathbb{F}}
\newcommand{\ds}{d_S}
\newcommand{\dham}{d_H}
\newcommand{\dr}{d_R}
\newcommand{\rb}{\color{red}{\bullet}}
\newcommand{\bb}{\color{blue}{\bullet}}

\title{Lifted codes and the multilevel construction for constant dimension codes}
\author{Sascha Kurz}\address{\textnormal{Department of Mathematics, University of Bayreuth, Germany, sascha.kurz@uni-bayreuth.de}}

\date{}

\begin{document}

\maketitle


{\small
\noindent \textbf{Abstract:}
 Constant dimension codes are e.g.\ used for error correction and detection in random linear network coding, so that constructions 
 for these codes have achieved wide attention. Here, we improve over $150$ lower bounds by describing better constructions 
 for subspace distance $4$.

\medskip

\noindent\textbf{Keywords:} constant dimension codes, multilevel construction, Echelon--Ferrers construction, linkage, network coding
}

\bigskip
\section{Introduction}
\label{sec_intro}
Let $q$ be a prime power and $\F_q$ be the finite field with $q$ elements. For two integers $0\le k\le n$ we denote by 
$\cG_q(n,k)$ the set of all $k$-dimensional subspaces in $\F_q^n$. The so-called subspace distance $\ds(U,W):=\dim(U)+\dim(W)-2\dim(U\cap W)=2k-2\dim(U\cap W)$ 
defines a metric on $\cG_q(n,k)$. A subset $\cC\subseteq\cG_q(n,k)$ is called a \emph{constant dimension code} (\cdc), its elements are also called codewords, 
and $\ds(\cC)=\min\!\left\{\ds(U,W)\,:\, U,W\in \cC,U\neq W\right\}$ is the corresponding \emph{minimum (subspace) distance}. We also call $\cC$ an $(n,M,d,k)_q$ {\cdc} if $\cC$ 
has cardinality $M$ and $\ds(\cC)\ge d$. The maximum possible cardinality of an $(n,M,d,k)_q$ {\cdc} is denoted by $A_q(n,d;k)$. Constant dimensions codes 
are e.g.\ applied in random linear network coding, see e.g.~\cite{koetter-kschischang08}, and the determination of bounds for $A_q(n,d;k)$ is one of the main 
problems. Here we improve more than $150$ of the previously best known constructions for {\cdc}s. An online table for 
bounds for $A_q(n,d;k)$ can be found at \url{subspacecodes.uni-bayreuth.de}, see the corresponding technical manual \cite{TableSubspacecodes}. 

The remaining part of this paper is structured as follows. In Section~\ref{sec_preliminaries} we introduce the necessary preliminaries and describe constructions 
for {\cdc}s from the literature. Our theoretical and algorithmical results are the topic of Section~\ref{sec_results}. The resulting numerical improvements for 
lower bounds for $A_q(n,d;k)$ are listed in Appendix~\ref{sec_improved_lower_bounds}. Extensive computational data about the details of the underlying constructions  
are given in Appendix~\ref{sec_skeleton_codes}.

\section{Preliminaries}
\label{sec_preliminaries}

Given a {\cdc} $\cC$ we first consider the question how to represent its codewords, i.e., $k$-dimensional subspaces $U\in\cG_q(n,k)$. Starting from a 
generator matrix whose $k$ rows form a basis of $U$ the application of the Gaussian elimination algorithm gives a unique generator matrix in \emph{reduced 
row echelon form} denoted by $E(U)$. In the other direction we write $U=\langle E(U)\rangle$. By $v(U)\in\F_2^n$ we denote the characteristic vector of 
the pivot columns in $E(U)$, which is also called \emph{identifying vector}. The \emph{Ferrers tableaux} $T(U)$ of $U$ arises from $E(U)$ by removing the zeroes 
from each row of $E(U)$ left to the pivots and afterwards removing all pivot columns. If we then replace all remaining entries by dots we obtain the 
\emph{Ferrers diagram} $\cF(U)$ of $U$ which only depends on the identifying vector $v(U)$. As an example we consider
$$
  U=\left\langle\begin{pmatrix}
  1&0&1&1&0&1&0&1&0&1\\
  1&0&0&1&1&1&1&1&1&1\\
  0&0&0&1&1&0&0&0&1&0\\ 
  0&0&0&0&0&0&1&1&0&1
  \end{pmatrix}\right\rangle\in\cG_2(10,4),
$$  
where we have 
$$
  E(U)=
  \begin{pmatrix}
  1&0&0&0&0&1&0&0&0&0\\
  0&0&1&0&1&0&0&1&1&1\\
  0&0&0&1&1&0&0&0&1&0\\
  0&0&0&0&0&0&1&1&0&1
  \end{pmatrix},
$$
$v(U)=1	0	1	1	0	0	1	0	0	0\in\F_2^{10}$, and 
$$
  \cF(U)=
  \begin{array}{llllll}
    \bullet & \bullet & \bullet & \bullet & \bullet & \bullet \\
            & \bullet & \bullet & \bullet & \bullet & \bullet \\
            & \bullet & \bullet & \bullet & \bullet & \bullet \\
            &         &         & \bullet & \bullet & \bullet 
  \end{array}.
$$

The \emph{Hamming distance} $\dham(u,w)=\#\left\{1\le i \le n\,:\, u_i\neq w_i\right\}$, for 
$u,w\in\F_q^n$, can be used to lower bound the subspace distance between two codewords 
$U,W\in\cG_q(n,k)$:
\begin{lemma}(\cite[Lemma 2]{etzion2009error})\\[-4mm]
  \label{lemma_dist_subspace_hamming}
  
  \noindent 
  For $U,W\in\cG_q(n,k)$ we have $\ds(U,W)\ge \dham(v(U),v(W))$.
\end{lemma}  

If the identifying vectors of two codewords coincide, then we can utilize the \emph{rank distance} 
$\dr(A,B):=\operatorname{rank}(A-B)$ for matrices $A,B\in\F_q^{m\times l}$:
\begin{lemma}(\cite[Corollary 3]{silberstein2011large})\\[-4mm]
  \label{lemma_dist_subspace_rank}

  \noindent 
  For $U,W\in\cG_q(n,k)$ with $v(U)=v(W)$ we have $\ds(U,W)=2\dr(E(U),E(W))$.
\end{lemma}

Since $\dr$ is a metric, we call a subset $C\subseteq\F_q^{m\times l}$ of matrices 
a \emph{rank-metric code}. If $C$ is a linear subspace of $\F_q^{m\times l}$ we call the code 
\emph{linear}. Given a Ferrers diagram $\cF$ with $m$ dots in the rightmost column and $l$ dots 
in the top row, we call a rank-metric code $C_{\cF}$ a \emph{Ferrers diagram rank-metric} (\FDRM) 
code if for any codeword $M\in \F_q^{m\times l}$ of $C_{\cF}$ all entries not in $\cF$ are zero. 
By $\dr(C_{\cF})$ we denote the minimum rank distance, i.e., the minimum of the rank distance between 
pairs of different codewords. 

\begin{definition}\label{definition_lifted}
  (\cite{silberstein2015error})\\
  Let $\cF$ be a Ferrers diagram and $C_{\cF}\subseteq \F_q^{k\times(n-k)}$ be an {\FDRM} code. The 
  corresponding \emph{lifted {\FDRM} code} $\cC_{\cF}$ is given by
  $$
    \cC_{\cF}=\left\{U\in\cG_q(n,k)\,:\, \cF(U)=\cF, T(U)\in C_{\cF}\right\}.
  $$ 
\end{definition}
We remark that the bijection between the codewords of the {\cdc} $\cC_{\cF}$ and the {\FDRM} code $C_{\cF}$ 
generalizes the construction of lifted \emph{maximum rank distance} (\MRD) codes. An {\MRD} code corresponds 
to the case of a Ferrers diagram $\cF$ with $k$ dots in each column and $n-k$ dots in each row (more details below).

Directly from Lemma~\ref{lemma_dist_subspace_rank} and Definition~\ref{definition_lifted} we can conclude:
\begin{lemma}(\cite[Lemma 4]{etzion2009error})\\\label{lemma_FDRM_CDC_equivalence} 
  Let $C_{\cF}\subseteq \F_q^{k\times(n-k}$ be an {\FDRM} code with minimum rank distance $\delta$, then the lifted 
  {\FDRM} code $\cC_{\cF}\subseteq \cG_q(n,k)$ is an $(n,\#C_{\cF},2\delta,k)_q$ {\cdc}.
\end{lemma}

Let $v(\cF)$ be the identifying vector of a given Ferrers diagram $\cF$. In general, we denote by $A_q(n,d;k;v)$ the 
maximum cardinality $M$ of an $(n,M,d,k)_q$ {\cdc} where all codewords have $v\in\F_2^n$ as identifying vector. We 
also speak of an $(n,M,d,k,v)_q$ {\cdc}. With this (and Lemma~\ref{lemma_FDRM_CDC_equivalence}) the upper bound for 
the cardinality of $C_{\cF}$ from \cite[Theorem 1]{etzion2009error}\footnote{In the cited paper the upper bound was 
stated for linear {\FDRM} codes only. However, the statement is also true without this assumption, as observed by 
e.g.\ the same authors later.} can be rewritten to:
\begin{theorem}
  \label{thm_upper_bound_ef}
  $$
    A_q(n,d;k;v(\cF))\le q^{\min\!\left\{\nu_i\,:\,0\le i\le d/2-1\right\}},
  $$
  where $\nu_i$ is the number of dots in $\cF$, which are neither contained in the first $i$ rows nor contained in the 
  last $\tfrac{d}{2}-1-i$ columns. 
\end{theorem}  
If we choose a minimum subspace distance of $d=6$, then we obtain $$A_2(10,6;4;1011001000)\le 2^8$$ due to
$$
  \begin{array}{llllll}
    \bb & \bb & \bb & \bb & \rb & \rb \\
        & \bb & \bb & \bb & \rb & \rb \\
        & \bb & \bb & \bb & \rb & \rb \\
        &     &     & \bb & \rb & \rb 
  \end{array}
  \quad  
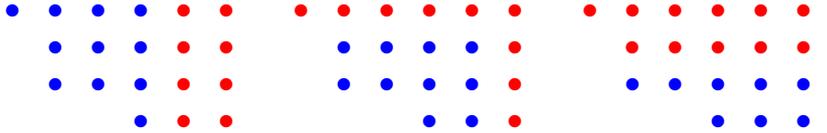
  \begin{array}{llllll}
    \rb & \rb & \rb & \rb & \rb & \rb \\
        & \bb & \bb & \bb & \bb & \rb \\
        & \bb & \bb & \bb & \bb & \rb \\
        &     &     & \bb & \bb & \rb 
  \end{array}
  \quad  
  \begin{array}{llllll}
    \rb & \rb & \rb & \rb & \rb & \rb \\
        & \rb & \rb & \rb & \rb & \rb \\
        & \bb & \bb & \bb & \bb & \bb \\
        &     &     & \bb & \bb & \bb 
  \end{array}
  .
$$

The Hamming weight of a vector $v\in\F_2^n$ is the Hamming distance $\dham(v,\mathbf{0})$ of $v$ to the zero vector.
\begin{theorem}(\cite[Theorem 3]{etzion2009error})\\ \label{thm_EF}
  If $\cS\subseteq\F_2^n$ is a set of binary vectors with Hamming weight $k$ that has minimum Hamming distance $d$ 
  and $\cC_v\subseteq\cG_q(n,k)$ is an $(n,\star,d,k,v)_q$ {\cdc} for each $v\in \cS$, then 
  $\cC=\cup_{v\in\cS} \cC_v$ is an $(n,\star,d,k)_q$ {\cdc} with cardinality $\sum_{v\in \cS} \#\cC_v$.
\end{theorem}

Choosing the $\cC_v$ as lifted {\FDRM} codes, the underlying construction is called \emph{multilevel construction} in \cite{etzion2009error} 
and \emph{Echelon-Ferrers construction} in some other papers. The binary code $\cS\subseteq\F_2^n$ is called \emph{skeleton code}. Using our 
notation, a given skeleton code $\cS$ with minimum Hamming distance $\tfrac{d}{2}$ gives
$$
  A_q(n,d;k)\ge \sum_{v\in \cS} A_q(n,d;k;v).
$$
The upper bound of Theorem~\ref{thm_upper_bound_ef} is attained in many cases including $d\le 4$ and rectangular Ferrers diagrams. 
For other cases we refer e.g..\ to \cite{antrobus2019maximal,liu2019constructions} and the references mentioned therein. Indeed, for 
finite fields no strict improvement of Theorem~\ref{thm_upper_bound_ef} is known so that it is conjectured that the upper bound can 
always be attained. If $2\le 2k\le n$ and $\cF$ is the rectangular Ferrers diagrams with $k$ dots in each column and $n-k$ dots in 
each row, then a rank-metric code $C_{\cF}\subseteq\F_q^{k\times(n-k)}$ attaining the maximum possible cardinality 
$q^{(n-k)(k-d/2+1)}$ for a given minimum subspace distance $d\le 2k$ is called \emph{maximum rank distance} (\MRD) code. Even linear 
{\MRD} codes exist for all parameters, so that lifting gives the well-known lower bound
$$
  A_q(n,d;k)\ge q^{(n-k)(k-d/2+1)}
$$
(assuming $2k\le n$), which is at least half the optimal value for $d\ge 4$, see e.g.\ \cite[Proposition 8]{heinlein2017asymptotic}.

Instead of starting with an {\FDRM} code $C_{\cF}$ and lift it to a {\cdc} $\cC_{\cF}$ one can also start from an $(m,N,d,k)_q$ {\cdc} 
$\cC$ and an {\MRD} code $\cM\subseteq \F_q^{k\times(n-m)}$ with minimum rank distance $d/2$. With this we can construct a {\cdc}
$$
  \cC'=\left\{\langle E(U)|M\rangle\,:\,U\in\cC, M\in\cM\right\}\subseteq \cG_q(n,k)
$$
with $\ds(\cC')=d$ and $\#\cC'=\#\cC\cdot\#\cM$, where $A|B$ denotes the concatenation of two matrices $A$ and $B$ with the same number of rows. This 
lifting variant was called \emph{Construction D} in \cite[Theorem 37]{silberstein2015error}, cf.\ \cite[Theorem 5.1]{gluesing9cyclic}. By construction, 
the identifying vectors of the codewords of $\cC'$ contain their $k$ ones in the first $m$ positions. More generally, we denote by 
$$
  {n_1 \choose k_1}\dots {n_l \choose k_l} 
$$  
the set of binary vectors which contain exactly $k_i$ ones in positions $1+\sum_{j=1}^{i-1} n_j$ to $\sum_{j=1}^{i} n_j$ for all $1\le i\le l$. 
With this, we write $A_q\!\left(n,d;k;{n_1 \choose k_1}\dots {n_l \choose k_l}\right)$ for the maximum cardinality of an $(n,\star,d,k)_q$ {\cdc} 
whose codewords all have identifying vectors in this set and state:  

\begin{theorem}
  \label{thm_construction_d}
  For each $0\le \Delta<n$ we have
  $$A_q(n,d;k)\ge A_q\!\left(n,d;k;{{n-\Delta}\choose k},{\Delta\choose 0}\right)\ge q^{\Delta(k-d/2+1)}A_q(n-\Delta,d;k).$$
\end{theorem}

The special structure of the identifying vectors can be used to add further codewords. In our notation the \emph{linkage construction} from 
\cite[Theorem 2.3]{gluesing10construction}, \cite[Corollary 39]{silberstein2015error} can be written as
$$
  A_q(n,d;k)\ge A_q\!\left(n,d;k;{{n-\Delta}\choose k},{\Delta\choose 0}\right)+A_q\!\left(n,d;k;{{n-\Delta}\choose 0},{\Delta\choose k}\right),
$$
which was improved to
\begin{eqnarray*}
  A_q(n,d;k)&\ge& A_q\!\left(n,d;k;{{n-\Delta}\choose k},{\Delta\choose 0}\right)\\ 
  &&+A_q\!\left(n,d;k;{{n-\Delta-k+d/2}\choose 0},{\Delta+k+d/2\choose k}\right)
\end{eqnarray*}
in \cite[Theorem 18, Corollary 4]{heinlein2017asymptotic}, taking Theorem~\ref{thm_construction_d} and $A_q(n,d;k;{n-m \choose 0},{m\choose k})=A_q(m,d;k)$ 
into account. By using the notation ${{n-n'}\choose {\le k-k'}},{{n'}\choose {\ge k'}}$ for the set of vectors in $\F_2^n$ with at most $k-k'$ ones in the 
first $n-n'$ positions and at least $k'$ ones in the last $n'$ positions we can denote by $A_q(n,d;k;{{n-n'}\choose {\le k-k'}},{{n'}\choose {\ge k'}})$ the 
maximum cardinality of an $(n,\star,d,k)_q$ {\cdc} whose codewords have identifying vectors in this set, so that Lemma~\ref{lemma_dist_subspace_hamming} gives: 
\begin{lemma}
  \label{lemma_ef_comb_special}
  For each $0\le \Delta<n$ we have
  \begin{eqnarray*}
    A_q(n,d;k)&\ge& A_q\!\left(n,d;k;{{n-\Delta}\choose k},{\Delta\choose 0}\right)\\
    &&             +A_q\!\left(n,d;k;{{n-\Delta}\choose {\le k-d/2}},{\Delta\choose {\ge d/2}}\right).  
  \end{eqnarray*}   
\end{lemma}
We remark that Lemma~\ref{lemma_ef_comb_special} was implicitly contained in the proofs of many papers improving lower bounds for $A_q(n,d;k)$, see e.g.\ 
\cite{cossidente2019combining,xu2018new} and the references cited therein. In \cite{kurz2019note} the quantity $A_q\!\left(n,d;k;{{n-\Delta}\choose {\le k-d/2}},{\Delta\choose {\ge d/2}}\right)$   
 was introduced as $B_q(n,\Delta,d;k)$. For the special case $\Delta=k$ a lower bound for $A_q\!\left(n,d;k;{{n-\Delta}\choose {\le k-d/2}},{\Delta\choose {\ge d/2}}\right)$ 
 was constructed in \cite{xu2018new} via
 $$
   \left\{\langle M|I_k\rangle\,:\, M\in \cM,\operatorname{rank}(M)\le k-d/2\right\},
 $$
 where $I_k$ denotes the $k\times k$ unit matrix and $\mathcal{M}\subseteq \F_q^{k\times(n-k)}$ is a rank metric code with $\dr(\cM)\ge d/2$. By replacing $I_k$ by $E(U)$ for 
 all codewords of a $(\Delta,\star,d,k)_q$ {\cdc} we obtain yet another variant of the lifting idea. One of the most general versions can be found in \cite[Lemma 4.1]{cossidente2019combining}.

Of course we can also utilize the multilevel/Echelon-Ferrers construction from Theorem~\ref{thm_EF} to obtain lower bounds for 
$A_q\!\left(n,d;k;{{n-\Delta}\choose {\le k-d/2}},{\Delta\choose {\ge d/2}}\right)$ by restricting the skeleton code $\cS$ to subsets of the set  
${{n-\Delta}\choose {\le k-d/2}},{\Delta\choose {\ge d/2}}$ of identifying vectors. More generally, we can define the Hamming distance 
$\dham(\cS,\cS')$ between two sets $\cS,\cS'\subseteq \F_2^n$ as 
$$
  \dham(\cS,\cS')=\min\left\{\dham(v,v')\,:\,v\in\cS, v'\in\cS'\right\}
$$
and slightly generalize Theorem~\ref{thm_EF} to:
\begin{theorem} 
  \label{thm_EF_generalized}
  Let $\cS_i\subseteq\F_2^n$ and $\cC_i$ be $(n,\star,d,k,\cS_i)_q$ {\cdc s} for all $1\le i\le n$ such that $\dham(\cS_i,\cS_j)\ge d$ for all $1\le i<j\le l$. Then 
  $\cC=\cup_{1\le i\le l} \cC_i$ is an $(n,\star,d,k)_q$ {\cdc} with cardinality $\#\cC=\sum_{i=1}^l \#\cC_i$.
\end{theorem}
If we choose $n=n_1+n_2$, $\cS_1=\left({n_1 \choose k},{n_2\choose 0}\right)\subseteq \F_2^n$, $\cS_2=\left({n_1 \choose 0},{n_2\choose k}\right)\subseteq \F_2^n$, and 
$\cS_3,\dots,\cS_l\subseteq {n\choose k} \cap \left({n_1\choose \ge d/2},{n_2\choose \ge d/2}\right)\subseteq \F_2^n$ of cardinality $\#\cS_j=1$ (for $3\le j\le l$) 
with $\dham(\cS_3,\dots,\cS_l)\ge d$, then the conditions of Theorem~\ref{thm_EF_generalized} are satisfied. This is \cite[Theorem 3.1]{li2019construction}.\footnote{In 
\cite{he2020note} it was claimed that \cite[Theorem 3.1]{li2019construction} is incorrect. However, the stated {\lq\lq}counterexample{\rq\rq} is flawed since the example for 
$C_2$ is not of the form specified in \cite[Theorem 3.1]{li2019construction} since there are e.g.\ non-zero entries in the first four columns of $a_1$.}  
More examples where constant dimension codes with different sets of identifying vectors are combined can be found in \cite{heinlein2017coset}.

\section{Results}
\label{sec_results}

We want to apply Theorem~\ref{thm_EF_generalized} in order to obtain improved lower bounds for $A_q(n,d;k)$. In order to avoid to explicitly 
deal with the existence and construction of {\FDRM} codes we restrict ourselves to subspace distance $d=4$, where the upper bound of 
Theorem~\ref{thm_upper_bound_ef} can be always attained. As an introductory example we state:
\begin{proposition}
  \label{prop_13_4_5}
  We have $A_2(13,4;5)\ge 4\,796\,825\,069$ and 
  $A_q(13,4;5) \ge q^{32}+q^{28}+q^{26}+8q^{24}+q^{23}+3q^{22}+q^{21}+4q^{20}+4q^{19}
  +5q^{18}+q^{17}+9q^{16}+8q^{15}+9q^{14}+6q^{13}+7q^{12}+5q^{11}+q^{10}+5q^9
  +3q^8+q^7+3q^6+4q^5+3q^4+q^3+3q^2$. 
\end{proposition}
\begin{proof}
  We apply Theorem~\ref{thm_EF} with skeleton codes $\mathcal{S}_{13,4,5}^1$ and $\mathcal{S}_{13,4,5}^2$, see Appendix~\ref{sec_skeleton_codes}.
\end{proof}
For $q\in\{2,3\}$ the previously best lower bounds were $A_2(13,4;5)\ge 4\,796\,417\,559$ and $A_3(13,4;5)\ge 1\,880\,918\,023\,783\,990$ \cite{he2019hierarchical}, 
while our parametric lower bound yields $A_3(13,4;5)\ge 1\,880\,918\,252\,176\,932$. The only algorithmical challenge is to find good skeleton codes. Note 
that $\mathcal{S}_{13,4,5}^1$ gives $A_q(13,4;5)\ge q^{32}+q^{28}+q^{26}+8q^{24}+q^{23}+2q^{22}+3q^{21}+5q^{20}+3q^{19}+3q^{18}+3q^{17}+8q^{16}+8q^{15}+9q^{14}
+5q^{13}+8q^{12}+9q^{11}+q^{10}+7q^9+2q^8+2q^7+2q^6+q^5+3q^4+2q^3+3q^2+q^0$, i.e., the coefficient for $q^{22}$ is only $2$ instead of $3$, which pays off for $q=2$ 
since the coefficient for $q^{21}$ is $3$ instead of $1$ and the coefficient for $q^{20}$ is $5$ instead of $4$. 

If the maximum sizes of the {\FDRM} codes are known, as it is the case for subspace distance $d=4$, for given parameters $n$, $d$, $k$, and $q$ the problem 
of determining the best lower bound for $A_q(n,d;k)$ based on Theorem~\ref{thm_EF} can be easily formulated as maximum weighted clique problem. To this end we denote 
by $\cG_{n,d,k,q}=\left(\cV_{n,d,k,q},\cE_{n,d,k,q}\right)$ the graph consisting of vertices corresponding to the binary vectors in $\F_2^n$ with Hamming weight $k$. 
W.l.o.g.\ we label the elements of $\cV_{n,d,k,q}$ from $1$ to ${n\choose k}$ such that $w(1)\ge w(2)\ge \dots$, where $w(v)$ denotes the weight of $v$ that is given 
by the maximum possible cardinality of the corresponding {\FDRM} code. For two different vertices $v,v'\in \cV_{n,d,k,q}$ the edge $\{v,v'\}$ is in $\cE_{n,d,k,q}$ iff 
the Hamming distance between $v$ and $v'$ is at least $d$. With this, the feasible skeleton codes are in bijection to the cliques of $\cG_{n,d,k,q}$ and we are searching 
for the maximum weight cliques. Looping over all cliques or all inclusion maximal cliques of $\cG_{n,d,k,q}$ becomes computationally intractable even for moderate 
sized parameters due to the quickly increasing number of vertices. Given an upper bound $ub$ for the maximum clique size in $\cG_{n,d,k,q}$, see e.g.\ 
\cite{brouwer1990new} and \url{https://www.win.tue.nl/~aeb/codes/Andw.html}, we can compute an upper bound on the weight $w(\cS'):=\sum_{v\in\cS'}w(v)$ of every 
clique that contains a subclique:

\begin{lemma}
  \label{lemma_ub_clique}
  Let $\cS,\cS'$ be  cliques in $\cG_{n,d,k,q}$ with $\cS\subseteq\cS'$ and $\max\left\{w(v)\,:\, v\in \cS\right\}<\min\left\{w(v)\,:\,w\in\cS'\backslash \cS\right\}$, where 
  $w(i)\ge w(j)\ge 0$ for all $i\le j$. Then, we have $w(\cS')\le \Omega$, where $\Omega$ is the value computed by Algorithm~\ref{alg:clique_ub} applied to $\cG_{n,d,k,q}$, $w$, and $\cS$.   
\end{lemma}   
\begin{proof}
  By $cand$ we denote the set of vertices $m+1\le v\le \#\cV$ with $\Big\{\left\{x,v\right\}\,:\,x\in \cS\Big\}\subseteq \cE$, where $m$, $\cV$, and $\cE$ are as 
  in Algorithm~\ref{alg:clique_ub}. Since $\cS\subseteq\cS'$ and $\cS'$ is a clique, we have $\Big\{\left\{x,v\right\}\,:\,x\in \cS\Big\}\subseteq \cE$ for all 
  vertices $v\in\cS'\backslash\cS$. Due to our assumption $\max\left\{w(v)\,:\, v\in \cS\right\}<\min\left\{w(v)\,:\,w\in\cS'\backslash \cS\right\}$ we also 
  have $m+1\le v\le \#\cV$ for all $v\in\cS'\backslash\cS$. Thus, we have $\cS'\subseteq\cS\cup cand$. If $\#\cS+\# cand\le ub$, then $\hat{\cS}=\cS\cup cand$ and 
  $w(\cS')\le w(\cS\cup cand)=w(\hat{\cS})=\Omega$. Otherwise we have $\#\cS'\le ub=\#\hat{\cS}$, set $s=\# \left(\cS'\backslash\cS\right)$, , $\hat{s}=
  \# \left(\hat{\cS}\backslash\cS\right)$, and write $\cS'=\cS\cup\left\{a_1,\dots,a_s\right\}$, $\hat{\cS}=\cS\cup\left\{b_1,\dots,b_{\hat{s}}\right\}$, where we 
  assume that the sequences $a_i$ and $b_i$ are increasing. Due to our assumption on the weight function $w$ and the choice of $\hat{\cS}\backslash\cS$ as those vertices 
  with the smallest elements in $cand$ we have $w(a_i)\le w(b_i)$ for all $1\le i\le s$, so that
  $$
    w(\cS')=w(\cS)+\sum_{i=1}^s w(a_i)\le w(\cS)+\sum_{i=1}^s w(b_i)\le w(\cS)+\sum_{i=1}^{\hat{s}} w(b_i)=w(\hat{\cS})=\Omega.   
  $$  
\end{proof}

\SetKwFunction{dive}{Dive}
\SetKwFunction{ub}{UB}
\SetKwFunction{isb}{IsStrictlyBetter}
\SetKwFunction{newrecord}{NewRecord}
\begin{algorithm}[htp]
  \KwIn{graph $\cG=(\cV,\cE)$ with weight function $w\colon \cV\to\mathbb{N}$ such that $\cV\subseteq\mathbb{N}$ and $w(i)\ge w(j)$ if $i\le j$,  
  a clique $\cS$ in $\cG$, and an upper bound $ub$ on the maximum clique size in $\cG$}
  \KwOut{An upper bound $\Omega$ for the weight of every clique extension of $\cS$}   
  $\Omega\longleftarrow w(\cS)$\;
  $\hat{\cS}\longleftarrow \cS$\;
  $m\longleftarrow 0$\;
  \If{$\cS\neq \emptyset$}
  {
    $m\longleftarrow \max\{v\,:\,v\in\cS\}$\;    
  }
  \For{$v$ from $m+1$ to $\# \cV$}
  {
    \If{$\Big\{\left\{x,v\right\}\,:\,x\in \cS\Big\}\subseteq \cE$ \textbf{and} $\#\hat{\cS}<ub$}
    { 
      $\Omega\longleftarrow \Omega+w(v)$\;
      $\hat{\cS}\longleftarrow \hat{\cS}\cup\left\{v\right\}$\;
    }      
  }
  \Return{$\Omega$}\;
  \caption{\texttt{UB}: upper bound for the weight of an extended clique}\label{alg:clique_ub}
\end{algorithm}
                                                                                                                      
We can use Algorithm~\ref{alg:clique_ub} to determine a maximum weight clique of $\cG_{n,d,k,q}$ without explicitly traversing all cliques:
\begin{proposition}
  \label{prop_max_weight_clique_exhaustive}
  Algorithm~\ref{alg:maximum_weight_clique} determines a maximum weight clique $\cU$ in $\cG_{n,d,k,q}$.   
\end{proposition}
\begin{proof}
  Let $\cU'$ be a maximum weight clique in $\cG_{n,d,k,q}$ and $\cU'_i$ be the subset of the smallest $i$ elements of $\cU$ for $1\le i\le \#\cU'$. 
  Now let $m$ be the largest index such that \texttt{Dive} is called with $\cS=\cU'_m$. Since \texttt{Dive} is initially called with $\cS=\emptyset=\cU'_0$, 
  $0\le m\le \#\cU'$ is well defined. If $m=\#\cU'$ then either $\cU$ is set to $\cU'$ or we already have $w(\cU)\ge w(\cU')$. Note that every replacement of 
  $\cU$ strictly increases the value of $w(\cU)$. In the remaining cases we assume $m<\#\cU$ and set $v'=\cU'_{m+1}\backslash\cU'_m$. By construction we have $l\le v'$, 
  so that the for loop of Algorithm~\ref{alg:dive} attains $v=v'$. Since $\cU'_{m+1}$ is a clique we have $\Big\{\left\{x,v\right\}\,:\,x\in \cS\Big\}\subseteq \cE$ 
  and since $\texttt{Dive}$ is not called with $\cU'_{m+1}=\cS\cup\{v\}$ we have \texttt{UB}$(\cG,\cS\cup\{v\},ub)\le w(\cU)$, so that Lemma~\ref{lemma_ub_clique} 
  gives $w(\cU)\ge w(\cU')$. Since every replacement of $\cU$ strictly increases the value of $w(\cU)$ the proposed statement follows.  
\end{proof}   

\begin{algorithm}[htp]
  \KwIn{graph $\cG=(\cV,\cE)$ with weight function $w\colon \cV\to\mathbb{N}$ such that $\cV\subseteq\mathbb{N}$ and $w(i)\ge w(j)$ if $i\le j$,  
  a clique $\cS$ in $\cG$, and an upper bound $ub$ on the maximum clique size in $\cG$}
  \KwOut{a weight maximum clique $\cU$ with respect to $w$}   
  \tcp{global data structures:}  
  $\cU \longleftarrow \emptyset$\;
  \tcp{local data structures:}
  $\cS \longleftarrow \emptyset$\;
  \dive($\cG$, $w$, $\cS$, $ub$)\; 
  \Return{$\cU$}\;
  \caption{Framework for the maximum weight clique algorithm}\label{alg:maximum_weight_clique}
\end{algorithm}

\begin{algorithm}[htp]
  \KwIn{clique $\cS\subseteq \cV$ and the input data from Algorithm~\ref{alg:maximum_weight_clique}}
  \KwOut{-}
  \If{$w(\cS)>w(\cU)$}
  {
    $\cU \longleftarrow \cS$\;
  }
  Let $1\le l\le \# V$ be the smallest index such that the elements in $\cS$ have strictly smaller indices; return if no such index exists\;
  \For{$v$ from $l$ to $\# \cV$} 
  {
    \If{$\Big\{\left\{x,v\right\}\,:\,x\in \cS\Big\}\subseteq \cE$}
    { 
      \If{\ub($\cG$, $w$, $\cS\cup\{v\}$, $ub$)$>w(\cU)$}
      {
        \dive($\cG$, $w$, $\cS\cup\{v\}$, $ub$)\;
      }
    }  
  } 
  \Return{}\;
  \caption{Subroutine \texttt{Dive}}\label{alg:dive}
\end{algorithm}

As an application of Proposition~\ref{prop_max_weight_clique_exhaustive} we remark that for $(n,d,k)=(13,4,5)$ and $q\in\{2,\dots,9\}$ the lower bound 
stated in Proposition~\ref{prop_13_4_5} is indeed the optimal multilevel/Echelon-Ferrers construction. So, based on Lemma~\ref{lemma_ub_clique}, an exhaustive 
search is indeed possible, while in \cite{he2019hierarchical} only a heuristic was used. However, we also have to use heuristics for larger parameters and 
replace Algorithm~\ref{alg:maximum_weight_clique} by Algorithm~\ref{alg:maximum_weight_clique_heuristic}. (The enlargement of $\cL'$ has to be implemented 
in a modified version of \texttt{Dive} to be technically correct.)

\begin{algorithm}[htp]
  \KwIn{graph $\cG=(\cV,\cE)$ with weight function $w\colon \cV\to\mathbb{N}$ such that $\cV\subseteq\mathbb{N}$ and $w(i)\ge w(j)$ if $i\le j$,  
  a clique $\cS$ in $\cG$, a list $\cL$  of cliques in $\cG$, and parameters $\Delta_1$, $\Delta_2$, $ub$}
  \KwOut{a weight maximum clique $\cU$ with respect to $w$ and a list $\cL'$ of cliques}   
  \tcp{global data structures:}  
  $\cU \longleftarrow \emptyset$\;
  $\cL' \longleftarrow \emptyset$\;
  \For{each $\cS\in\cL$}
  {
    \dive($\cG$, $w$, $\cS$, $ub$)\;
    if \dive is called with $\#\cS\ge ub-\Delta_2$ then add the smallest $ub-\Delta_2-\Delta_1$ elements of $\cS$ as a clique to $\cL'$\; 
  }
  \Return{$\cU$}\;
  \caption{Framework a heuristic for the maximum weight clique algorithm}\label{alg:maximum_weight_clique_heuristic}
\end{algorithm}
 
We can iteratively apply Algorithm~\ref{alg:maximum_weight_clique_heuristic} in order to heuristically find a clique of large weight in $\cG_{n,d,k,q}$. Starting 
from $\cL=\emptyset$ we can use the determined list $\cL'$ in the next round increasing $ub$ incrementally, say by $10$ in each iteration, until no further improvement 
is found. The advantage of that approach is that using the parameters $\Delta_1$ and $\Delta_2$ we can control the extend to which we want to perform an exhaustive search. The 
iterative process also partially prevents from stalling in local neighborhoods of partial cliques that are very diverse to a global optimum. As an example for a result 
obtained by the application of Algorithm~\ref{alg:maximum_weight_clique_heuristic} we state:  

\begin{proposition}
  \label{prop_17_4_6}
  $A_q(17,4;6) \ge q^{55}+q^{51}+q^{49}+8q^{47}+3q^{45}+3q^{44}+5q^{43}+q^{42}+5q^{41}+9q^{40}
  +22q^{39}+7q^{38}+11q^{37}+13q^{36}+19q^{35}+3q^{34}+17q^{33}+15q^{32}+69q^{31}+20q^{30}
  +49q^{29}+22q^{28}+33q^{27}+15q^{26}+23q^{25}+20q^{24}+38q^{23}+17q^{22}+29q^{21}+24q^{20}
  +40q^{19}+19q^{18}+20q^{17}+15q^{16}+28q^{15}+15q^{14}+13q^{13}+8q^{12}+7q^{11}+5q^{10}
  +3q^9+10q^8+q^7+2q^6+q^5+2q^4+q^3+q^2+q^1+q^0$.
\end{proposition}
\begin{proof}
  We apply Theorem~\ref{thm_EF} with skeleton codes $\mathcal{S}_{17,4,6}$, see Appendix~\ref{sec_skeleton_codes}.
\end{proof}

In Table~\ref{table_improved_ef_parameters} we list the cases where Algorithm~\ref{alg:maximum_weight_clique} or Algorithm~\ref{alg:maximum_weight_clique_heuristic} yields 
a skeleton code such that Theorem~\ref{thm_EF} gives a strictly better constructive lower bound for $A_q(n,4;k)$. In all cases we obtain improvements for all $q\in\{2,\dots,9\}$,   
see Table~\ref{table_improved_ef_numerical} in Appendix~\ref{sec_improved_lower_bounds}. The utilized skeleton codes are listed in Appendix~\ref{sec_skeleton_codes}. Some of them 
consist of over 1000~identifying vectors. Whenever we state two skeleton codes, for given parameters $n$ and $k$, the first code gives the lower bound for $A_2(n,4;k)$ and the second 
code gives the lower bound for general field sizes $q$, which is larger for $q\ge 3$, see Proposition~\ref{prop_13_4_5} for an example. Whenever we only list one skeleton code, it 
yields the utilized lower bound for all field sizes $q\ge 2$, see Proposition~\ref{prop_17_4_6} for an example. We remark that for $d=4$ all explicitly stated numerical 
results of \cite{he2019hierarchical,liu2019parallel} have been strictly improved. (The other results of \cite{he2019hierarchical} depend on the unproven assumption, nevertheless 
the author states otherwise, that the upper bound of Theorem~\ref{thm_upper_bound_ef} is attained for $d=6$.)  

\begin{table}[htp]
  \begin{center}
    \begin{tabular}{cl}
      \hline
      $k$ & $n$ \\ 
      \hline
      5 & 13, 14 \\
      6 & 14, 15, 16, 17 \\
      7 & 16, 17,  18, 19 \\
      8 & 18, 19 \\ 
      9 & 19 \\ 
      \hline
    \end{tabular}
    \caption{Parameters where improved codes for $A_q(n,4;k)$ have been found using Algorithm~\ref{alg:maximum_weight_clique} or Algorithm~\ref{alg:maximum_weight_clique_heuristic}.} 
    \label{table_improved_ef_parameters}       
  \end{center}
\end{table}

\bigskip

Another approach, to obtain improved lower bounds, is based on Theorem~\ref{thm_construction_d} and Lemma~\ref{lemma_ef_comb_special}: 
\begin{proposition}
  \label{prop_11_4_3}
  $A_q(11,4;3) \ge q^8\cdot A_q(7,4;3)+q^4 + q^3 + 2q^2 + q + q^0$.
\end{proposition}
\begin{proof}
  We apply Theorem~\ref{thm_construction_d} with $\Delta=4$ to deduce $A_q\!\left(11,4;3;{7\choose 3},{4\choose 0}\right)\ge q^{8}\cdot A_q(7,4;3)$ and 
  Theorem~\ref{thm_EF} with the skeleton code $\cS_{11,7,4,3}$, see Appendix~\ref{sec_skeleton_codes}, to deduce\\ 
  $A_q\!\left(11,4;3;{7\choose \le 1},{4\choose \ge 2}\right)\ge
     q^4 + q^3 + 2q^2 + q + q^0$. 
  With this, the stated lower bound follows from Lemma~\ref{lemma_ef_comb_special}.
\end{proof}
For $q\ge 3$ Proposition~\ref{prop_11_4_3} yields a strictly larger lower bound than the previous record from \cite{heinlein2019generalized}. 
Using $\Delta=3$ and skeleton code $$\cS_{10,7,4,3}=\left\{ 0000100110, 0000010101, 0000001011  \right\}\,\hat{=}\,\{38, 21, 11\},$$ 
where the integers give the binary identifying vectors reading them in base $2$ representation, gives $A_q(10,4;3) \ge q^6\cdot A_q(7,4;3)+q^2+q+1$, 
which equals the lower bound from \cite{heinlein2019generalized}. For the naming of the skeleton codes $\cS_{n,m,d,k}$ for $A_q(n,d;k)$ we use 
$m=n-\Delta$ as an abbreviation.   

\begin{table}[htp]
  \begin{center}
    \begin{tabular}{cl}
      \hline
      $k$ & $n$\\ 
      \hline
      3 & 
          11(7), 12(7), 13(7), 14(11,7), 15(11,7), 16(13,7), \\ 
        & 17(13,7), 18(13,7), 19(13,7) \\ 
      4 & 13(8), 14(8), 15(8), 17(12), 18(12), 19(12)\\ 
      \hline
    \end{tabular}
    \caption{Parameters where improved codes for $A_q(n,4;k)$ have been found using Theorem~\ref{thm_construction_d} and Lemma~\ref{lemma_ef_comb_special}.} 
    \label{table_improved_ef_linkage}       
  \end{center}
\end{table}

In Table~\ref{table_improved_ef_linkage} we list the parameters $n$ and $k$ where the approach based on Theorem~\ref{thm_construction_d} and Lemma~\ref{lemma_ef_comb_special} 
yields strict improvements. The corresponding value of $m$ is stated in brackets, where the last is for general field sizes $q$ and the last but one for $q=2$. In 
some cases, when $n$ is rather small, better lower bounds for $A_2(n,4;k)$ have been found using integer linear programming and prescribed automorphisms, see e.g.\ 
\cite{kohnert2008construction} for an 	introductory paper and \cite{TableSubspacecodes} for the precise reference per specific instance. The numerical improvements are 
listed in Table~\ref{table_improved_ef_linkage_numerical} in Appendix~\ref{sec_improved_lower_bounds} and the corresponding skeleton codes in Appendix~\ref{sec_skeleton_codes}. 
We remark that for $k=4$ and $n\in\{12,16\}$ the obtained codes are inferior compared to those from \cite{cossidente2019combining}. We remark that for $d=4$ all numerical 
results of \cite{li2019construction} have been strictly improved. 

In order to find the skeleton codes Algorithm~\ref{alg:maximum_weight_clique} and Algorithm~\ref{alg:maximum_weight_clique_heuristic} are applied by modifying the input graphs 
$\cG_{n,m,d,k,q}$, where $m=n-\Delta$, accordingly. Since Lemma~\ref{lemma_ub_clique} and Algorithm~\ref{alg:clique_ub} significantly speed up the solution process by cutting 
parts of the search tree as early as possible, we also want to use {\lq\lq}good{\rq\rq} upper bounds on the maximum clique size in $\cG_{n,m,d,k,q}$. While of cause 
the maximum clique size in $\cG_{n,d,k,q}$ is an upper bound this can usually be improved. 

\begin{proposition}
  \label{prop_ILP_ub}
  Let $1\le k\le m\le n$ and $2\le d/2\le k$ be integers. Then, the maximum cardinality of a set $\cS\subseteq\F_2^n$ of binary vectors with Hamming weight $k$ that 
  have at most $k-d/2$ ones in their first $m$ coordinates with minimum Hamming distance $d$ is upper bounded by $\sum_{j\in J} c_j$, 
  where $J=\left\{j\in\mathbb{N}\,:\, j\ge k-n+m,j\le k-d/2\right\}$ and the $c_j$ are integers satisfying the constraints
  $$
    \sum_{j\in J} {j\choose i}\cdot {{k-j}\choose{k-d/2+1-i}}c_j\le {m\choose i}{{n-m}\choose{k-d/2+1-i}}  
  $$
  for all $0\le i\le k-d/2+1$. Moreover, we have $c_0\le A_1(n-m,d;k)$, which denotes the maximum cardinality of a subset of binary vectors in $\F_2^{n-m}$ with Hamming 
  weight $k$ and minimum Hamming distance $d$. 
\end{proposition}
\begin{proof}
  Given $\cS$ we let $c_j$ be the number of elements in $\cS$ that have exactly $j$ ones in their first $m$ coordinates. It can be easily checked that $j\in J$ implies 
  $0\le j\le m$, $0\le k-j\le n-m$, and $j\le k-d/2$, i.e., the counts $c_j$ are well defined. Since $\cS$ has minimum Hamming distance $d$ every subset $I\subseteq\F_2^n$ 
  of cardinality $k-d/2+1$ is contained in the support of at most one element in $\cS$. For a given integer $0\le i\le k-d/2+1$ we consider all those sets $I$ containing 
  exactly $i$ elements in $\{1,\dots,m\}$. There are exactly ${m\choose i}{{n-m}\choose{k-d/2+1-i}}$ of those sets and every element of $\cS$ with $j$ ones in the first $m$ 
  coordinates has a support that contains exactly ${j\choose i}\cdot {{k-j}\choose{k-d/2+1-i}}$ of those sets. Thus, the stated inequalities are valid. The additional 
  upper bound for $c_0$ follows directly from our choice of the count $c_0$.  
\end{proof}
Given fixed parameters the integer linear program of Proposition~\ref{prop_ILP_ub} can be solved easily. As an example we will show that also parametric upper bounds can be 
concluded. To this end, let $k=4$, $d=4$, and $n-m\ge 4$, so that $J=\{0,1,2\}$. For $i\in\{0,1,2\}$ the constraints of Proposition~\ref{prop_ILP_ub} read
\begin{eqnarray*}
  4c_0+c_1 &\le& {m\choose 0}{{n-m}\choose 3}={{n-m}\choose 3},\\
  3c_1+2c_2&\le& {m\choose 1}{{n-m}\choose 2}=\frac{m(n-m)(n-m-1)}{2},\\
  2c_2&\le& {m\choose 2}{{n-m}\choose 1}=\frac{m(m-1)(n-m)}{2},
\end{eqnarray*}
noting that that the constraint for $i=3=k-d/2+1$ is trivially satisfied. Since $c_1$ and $c_2$ are non-negative integers, dividing the second constraint by two gives
$$
  c_1+c_2\le \left\lfloor\frac{m(n-m)(n-m-1)}{4}\right\rfloor, 
$$ 
so that $c_0\le A_1(n-m,4;4)$ gives
$$
  c_0+c_1+c_2\le \left\lfloor\frac{m(n-m)(n-m-1)}{4}\right\rfloor +A_1(n-m,4;4).
$$
For $(n,m)=(13,8)$ we obtain an upper bound of $40+1=41$, which indeed is attained by $\cS_{13,8,4,4}$, see Appendix~\ref{sec_skeleton_codes}, while $A_1(13,4;4)=65$.  
Similarly, for $(n,m)=(14,8)$ we obtain an upper bound of $60+3=63$, which indeed is attained by $\cS_{14,8,4,4}$, while $A_1(14,4;4)=91$ is much larger.  

We remark that Proposition~\ref{prop_ILP_ub} mimics \cite[Lemma 4.1]{kurz2019note}, which proves an upper bound for 
$A_q\!\left(n,d;k;{{n-\Delta}\choose {\le k-d/2}},{\Delta\choose {\ge d/2}}\right)$. Possibly the quantity $A_1(n,m,d;k)$ upper bounded in Proposition~\ref{prop_ILP_ub},  
generalizing $A_1(n,d;k)$, is interesting on its own and similar techniques as for the (partial) determination of $A_1(n,d;k)$ can be applied. To this end, we remark 
that the upper bound of Proposition~\ref{prop_ILP_ub} can surely be improved in most cases but was sufficiently good for our purpose.    


\newpage

\appendix

\section{Improved lower bounds}
\label{sec_improved_lower_bounds}

In Table~\ref{table_improved_ef_numerical} we compare the numerical improvements for $A_q(n,4;k)$ based on Theorem~\ref{thm_EF}. The previously 
best known lower bounds are mostly from \cite{he2019hierarchical}. However, the lower bounds for $A_q(19,4;7)$ from \cite{he2019hierarchical} are flawed, 
i.e., some magnitudes too small.  

\tiny

\begin{table}[htp!]
  \begin{center}
    \begin{tabular}{|l|l|l|}
      \hline
      $A_q(n,4;k)$ & New & Old \\ 
      \hline
      \hline
        $A_2(13,4;5)$ & 4796825069                                & 4796417559 \cite{he2019hierarchical}\\
      \hline
        $A_3(13,4;5)$ & 1880918252176932                          & 1880918023783990 \cite{he2019hierarchical}\\
      \hline
        $A_4(13,4;5)$ & 18525690519076963184                      & 18525690479132333173 \cite{he2019hierarchical}\\
      \hline
        $A_5(13,4;5)$ & 23322304251254415373950                   & 23322304248923865096456 \cite{he2019hierarchical}\\
      \hline
        $A_7(13,4;5)$ & 1104898620940893642387898640              & 1104898620939789578683671514 \cite{he2019hierarchical}\\
      \hline
        $A_8(13,4;5)$ & 79247846163928378274466378432             & 79247846163915655208442806985 \cite{he2019hierarchical}\\
      \hline
        $A_9(13,4;5)$ & 3434214279120463268840947142394           & 3434214279120353599762054717228 \cite{he2019hierarchical}\\
      \hline
        $A_2(14,4;5)$ & 76749681496                               & 76745404672 \cite{he2019hierarchical}\\
      \hline
        $A_3(14,4;5)$ & 152354381482802889                        & 152354354408240436 \cite{he2019hierarchical}\\
      \hline
        $A_4(14,4;5)$ & 4742576773448941977744                    & 4742576757171205745408 \cite{he2019hierarchical}\\
      \hline
        $A_5(14,4;5)$ & 14576440157067948253027775                & 14576440154794852120820500 \cite{he2019hierarchical}\\
      \hline
        $A_7(14,4;5)$ & 2652861588879102858398037208053           & 2652861588875282767080909163052 \cite{he2019hierarchical}\\
      \hline
        $A_8(14,4;5)$ & 324599177887450844354761706030144         & 324599177887378338095324360943616 \cite{he2019hierarchical}\\
      \hline
        $A_9(14,4;5)$ & 22531879885309361371329789150871659       & 22531879885308389971852673089028988 \cite{he2019hierarchical}\\
      \hline     
        $A_2(13,4;6)$ & 38327432465                               & 38325131657 \cite{he2019hierarchical}\\
      \hline
        $A_3(13,4;6)$ & 50782269101589019                         & 50782269101569336 \cite{he2019hierarchical}\\
      \hline
        $A_4(13,4;6)$ & 1185639430145591548865                    & 1185639430145591024577 \cite{he2019hierarchical}\\
      \hline
        $A_5(13,4;6)$ & 2915286427121720403833501                 & 2915286427121720397974126 \cite{he2019hierarchical}\\
      \hline
        $A_7(13,4;6)$ & 378980216844611802379905892367            & 378980216844611802379704124332 \cite{he2019hierarchical}\\
      \hline
        $A_8(13,4;6)$ & 40574896910456482568428114359809          & 40574896910456482568427309053441 \cite{he2019hierarchical}\\
      \hline
        $A_9(13,4;6)$ & 2503542202545027727509322118254705        & 2503542202545027727509319406311282 \cite{he2019hierarchical}\\
      \hline
        $A_2(14,4;6)$ & 1227267234053                             & 1227203232293 \cite{he2019hierarchical}\\
      \hline
        $A_3(14,4;6)$ & 12340234566815274820                      & 12340234566810426241 \cite{he2019hierarchical}\\
      \hline
        $A_4(14,4;6)$ & 1214095649227435851637265                 & 1214095649227435312865809 \cite{he2019hierarchical}\\
      \hline
        $A_5(14,4;6)$ & 9110270832108553596766282526              & 9110270832108553578429954401 \cite{he2019hierarchical}\\
      \hline
        $A_7(14,4;6)$ & 6369520523839151727825310046433656        & 6369520523839151727821917674342793 \cite{he2019hierarchical}\\
      \hline
        $A_8(14,4;6)$ & 1329558223045414729174436187566776385     & 1329558223045414729174409793499570241 \cite{he2019hierarchical}\\
      \hline
        $A_9(14,4;6)$ & 147831663555770209444761899682581639650   & 147831663555770209444761739522908440329 \cite{he2019hierarchical}\\
      \hline            
        $A_2(15,4;6)$ & 39273527139056                            & 39267675031563 \cite{he2019hierarchical}\\
      \hline
        $A_3(15,4;6)$ & 2998677038028861083854                    & 2998676636295383433055 \cite{he2019hierarchical}\\
      \hline
        $A_4(15,4;6)$ & 1243233944887303700963929029              & 1243233943362040432057180581 \cite{he2019hierarchical}\\
      \hline
        $A_5(15,4;6)$ & 28469596350369158409165204552256          & 28469596349440811610995019309681 \cite{he2019hierarchical}\\
      \hline
        $A_7(15,4;6)$ & 107052531444164864809019626364096185842   & 107052531444149851093995761837649623043 \cite{he2019hierarchical}\\
      \hline
        $A_8(15,4;6)$ & 43566963852752158504511827652740340200969 & 43566963852751451455533741586485379191945 \cite{he2019hierarchical}\\
      \hline
        $A_9(15,4;6)$ & 87293119013046753015079265888857783791854 & 87293119013046541123889721019479429050233 \\
                      & 44                                        & 81 \cite{he2019hierarchical}\\
      \hline
        $A_2(16,4;6)$ & 1256765678235469                          & 1256703351587805 \cite{he2019hierarchical}\\
      \hline
        $A_3(16,4;6)$ & 728678523485248028210242                  & 728678523483522880513165 \cite{he2019hierarchical}\\
      \hline
        $A_4(16,4;6)$ & 1273071559584675915665964748625           & 1273071559584674249524907514705 \cite{he2019hierarchical}\\
      \hline
        $A_5(16,4;6)$ & 88967488594921579438014390756229526       & 88967488594921579094529973785226401 \cite{he2019hierarchical}\\
      \hline
        $A_7(16,4;6)$ & 17992318959820794052190487867406638758853 & 17992318959820794052179837126816751848892 \\
                      & 38                                        & 05 \cite{he2019hierarchical}\\
      \hline
        $A_8(16,4;6)$ & 14276022715269827610082997266445121659360 & 14276022715269827610082737902847157431165\\ 
                      & 67137                                     & 87585 \cite{he2019hierarchical}\\
      \hline
    \end{tabular}
    \caption{Improvements based on Theorem~\ref{thm_EF}.} 
    \label{table_improved_ef_numerical}
  \end{center}     
\end{table}  

\begin{table}[htp!]      
  \begin{center}
    \begin{tabular}{|l|l|l|}
      \hline
      $A_q(n,4;k)$ & New & Old \\ 
      \hline
      \hline
        $A_9(16,4;6)$ & 51545713846013977302875956389540675158560 & 51545713846013977302875913003333046529357\\
                      & 7978830                                   & 7720745 \cite{he2019hierarchical}\\    
      \hline  
        $A_2(17,4;6)$ & 40214593296350543                         & 40210734642430233 \cite{he2019hierarchical}\\ 
      \hline
        $A_3(17,4;6)$ & 177068881245303116902339546               & 177068857538981556600415147 \cite{he2019hierarchical}\\
      \hline
        $A_4(17,4;6)$ & 1303625277015154268047530588005973        & 1303625275416014562978042328889121 \cite{he2019hierarchical}\\
      \hline
        $A_5(17,4;6)$ & 278023401859130594858143611346347770156   & 278023401850065051300001033963426380051 \cite{he2019hierarchical}\\
      \hline
        $A_7(17,4;6)$ & 30239690475770808604751881256801280488268 & 30239690475766567624866522189453324986367 \\
                      & 311358                                    & 024243 \cite{he2019hierarchical}\\
      \hline
        $A_8(17,4;6)$ & 46779671233396171116050537713696416375261 & 46779671233395411929881098933346105696388\\
                      & 813580361                                 & 708610177 \cite{he2019hierarchical}\\
      \hline
        $A_9(17,4;6)$ & 30437228568932793457735943799142491471855 & 30437228568932719575938776421624397010040\\
                      & 909627410848                              & 042569507531 \cite{he2019hierarchical}\\  
      \hline
        $A_2(15,4;7)$ & 313939996903443                           & 313923840120169 \cite{he2019hierarchical}\\
      \hline
        $A_3(15,4;7)$ & 80962390735680572668348                   & 80962387333738514962426 \cite{he2019hierarchical}\\
      \hline
        $A_4(15,4;7)$ & 79566863776146089092873059065             & 79566863724904828874349525569 \cite{he2019hierarchical}\\   
      \hline
        $A_5(15,4;7)$ & 3558699030838267109468908231148586        & 3558699030750375431966367668488876 \cite{he2019hierarchical}\\   
      \hline
        $A_7(15,4;7)$ & 36719018111675892766467253005230234472800 & 36719018111669485366326051726236511270134 \cite{he2019hierarchical}\\   
      \hline
        $A_8(15,4;7)$ & 22306285465490553868002054600166649426159 & 22306285465490013824377105083555006005322 \\
                      & 313                                       & 241 \cite{he2019hierarchical}\\   
      \hline
        $A_9(15,4;7)$ & 63636683737167279191968272796410112452114 & 63636683737167010339661770016731801211973\\
                      & 39310                                     & 68344 \cite{he2019hierarchical}\\   
      \hline
        $A_2(16,4;7)$ & 20093092605969267                         & 20090530823175168 \cite{TableSubspacecodes}\\
      \hline
        $A_3(16,4;7)$ & 59021599901810630842384564                & 59021591907098096238648717 \cite{he2019hierarchical}\\   
      \hline
        $A_4(16,4;7)$ & 325905875863710195123556939265157         & 325905875503966895183219736444928 \cite{he2019hierarchical}\\   
      \hline
        $A_5(16,4;7)$ & 55604672371466910516074031340231456506    & 55604672369921077974140644073486328125 \cite{he2019hierarchical}\\   
      \hline
        $A_7(16,4;7)$ & 43199557618314831473921116433704053011403 & 43199557618309916202868392749119713436798\\
                      & 04392                                     & 04101 \cite{he2019hierarchical}\\    
      \hline
        $A_8(16,4;7)$ & 58474588970678832897891612620850698965996 & 58474588970678073098225984456969906775715\\
                      & 13506569                                  & 97238272 \cite{he2019hierarchical}\\   
      \hline
        $A_9(16,4;7)$ & 33819142841966544686367447197905513991858 & 33819142841966479587063098904956143968989\\
                      & 06891914222                               & 57540722053 \cite{he2019hierarchical}\\
      \hline
        $A_2(17,4;7)$ & 1285973764408635208                       & 1285780755925958656 \cite{he2019hierarchical}\\   
      \hline
        $A_3(17,4;7)$ & 43026746493711590523056074275             & 43026740586030433477849264332 \cite{he2019hierarchical}\\   
      \hline
        $A_4(17,4;7)$ & 1334910467554985714047715561596194368     & 1334910466075153545917778842397704192 \cite{he2019hierarchical}\\   
      \hline
        $A_5(17,4;7)$ & 86882300580430813669650616907760767647637 & 86882300578011697800830006599426269531250\\
                      & 5                                         & 0 \cite{he2019hierarchical}\\   
      \hline
        $A_7(17,4;7)$ & 50823847542371226941060627330811586531967 & 50823847542365442865880653737635543921067\\
                      & 5271592683                                & 3571016516 \cite{he2019hierarchical}\\ 
      \hline
        $A_8(17,4;7)$ & 15328762651129632211002306082390007789441 & 15328762651129433020259724863226574216922\\
                      & 94589626442240                            & 82703427141632 \cite{he2019hierarchical}\\   
      \hline
        $A_9(17,4;7)$ & 17972879091077542502798423597933679533502 & 17972879091077507906314355839720905589861\\
                      & 14940243153381903                         & 92090028734124332 \cite{he2019hierarchical}\\    
      \hline
        $A_2(18,4;7)$ & 82302571633443282819                      & 82291970549255555624 \cite{TableSubspacecodes}\\
      \hline  
        $A_3(18,4;7)$ & 31366498204356725852831211333673          & 31366481390109307710095048570592 \cite{he2019hierarchical}\\   
      \hline
        $A_4(18,4;7)$ & 5467793275108157387543615361694729718789  & 5467793218060404819993105682513015078912 \cite{he2019hierarchical}\\   
      \hline
        $A_5(18,4;7)$ & 13575359465692365684898552467453540627598 & 13575359458996180137997725978493690490722\\
                      & 430631                                    & 656250 \cite{he2019hierarchical}\\   
      \hline
        $A_7(18,4;7)$ & 59793748395124324914384236296038210886457 & 59793748394835040265326746482178302037689\\
                      & 090991018946109                           & 516829976915564 \cite{he2019hierarchical}\\     
      \hline
        $A_8(18,4;7)$ & 40183431564177263067797197215986973010693 & 40183431564146475997690442413946479412809\\
                      & 5844085497257345033                       & 5421048885275000832 \cite{he2019hierarchical}\\   
      \hline
        $A_9(18,4;7)$ & 95515248370413402653106548360544977074290 & 95515248370399043860359211584572109956785\\
                      & 8421375268890568749019                    & 5543645243756034671222 \cite{he2019hierarchical}\\  
      \hline
    \end{tabular}
    \caption{Improvements based on Theorem~\ref{thm_EF} cont.} 
  \end{center}
\end{table}
      
\begin{table}[htp!]      
  \begin{center}
    \begin{tabular}{|l|l|l|}
      \hline
      $A_q(n,4;k)$ & New & Old \\ 
      \hline
      \hline      
        $A_2(19,4;7)$ & 5267367924445148864092                    & 5058097205000347197549 \cite{silberstein2015error}\\
      \hline
        $A_3(19,4;7)$ & 22866177190621892757679222318915925       & 22813524065165704375162588681706875 \cite{silberstein2015error}\\ 
      \hline
        $A_4(19,4;7)$ & 22396081254840251014738769405257047429902 & 22388542052518188728938813857333763516548\\
                      & 630                                       & 177 \cite{silberstein2015error}\\
      \hline
        $A_5(19,4;7)$ & 21211499165144290480179096294676880646602 & 21209813680090451272236350632742053588039\\
                      & 6993016307                                & 4984750151 \cite{silberstein2015error}\\ 
      \hline
        $A_7(19,4;7)$ & 70346747049379816915369395320383414712645 & 70346075648600428027267630601996366936175\\
                      & 47397843085217471817                      & 78982540335678343419 \cite{silberstein2015error}\\  
      \hline
        $A_8(19,4;7)$ & 10533845483959684448674419012790949807456 & 10533801523606741064564183358274781734572\\
                      & 1619935428942997712801930                 & 1344118067187280665018945 \cite{silberstein2015error}\\
      \hline
        $A_9(19,4;7)$ & 50760719109220869118965528769591000765261 & 50760616504319115174434295841792410847448\\
                      & 3420811027477177039803795175              & 9256852919544839333945104727 \cite{silberstein2015error}\\  
      \hline
        $A_2(18,4;8)$ & 1316667538397101428149                    & 1264601087568682942805 \cite{silberstein2015error}\\ 
      \hline                     
        $A_3(18,4;8)$ & 2540681546251523549771895302399518        & 2534836537121138399731153735006570 \cite{silberstein2015error}\\ 
      \hline
        $A_4(18,4;8)$ & 13997550036505113139933121791886674834127 & 13992838834950406313750769467311428427492\\
                      & 53                                        & 01 \cite{silberstein2015error}\\
      \hline
        $A_5(18,4;8)$ & 84845996396175179770599196313863292814169 & 84839254734261831517960672499611973762512\\
                      & 28756276                                  & 21110026 \cite{silberstein2015error}\\ 
      \hline
        $A_7(18,4;8)$ & 14356478989082867395087999913040607219209 & 14356341969123295746286010290803218889077\\
                      & 1727223436678152394                       & 1110979148690800890 \cite{silberstein2015error}\\
      \hline
        $A_8(18,4;8)$ & 16459133568567622026207934780783177598554 & 16459064880639275271231203508694685723605\\
                      & 08541359722889311735873                   & 96441738049974310277185 \cite{silberstein2015error}\\  
      \hline
        $A_9(18,4;8)$ & 62667554455729492107715380469110602507032 & 62667427783112758078838097157679059423642\\
                      & 42571284112409913460035888                & 52612630126994566337005626 \cite{silberstein2015error}\\  
      \hline     
        $A_2(19,4;8)$ & 168534060204346081643054                  & 161868939208791416732918\cite{silberstein2015error}\\ 
      \hline                    
        $A_3(19,4;8)$ & 5556470543990117093920796733024932097     & 5543687506683929680212033218489009029\cite{silberstein2015error}\\ 
      \hline
        $A_4(19,4;8)$ & 22933585979825739912419408751656551857477 & 22925867147182745704449260695243044338837\\
                      & 750740                                    & 225537\cite{silberstein2015error}\\
      \hline 
        $A_5(19,4;8)$ & 66285934684512179124945660759025164232422 & 66280667761142055873406775390321854501962\\
                      & 7507808667905                             & 6771492187626\cite{silberstein2015error}\\ 
      \hline
        $A_7(19,4;8)$ & 11823177776106271896222089221602557589481 & 11823064934277706348783619772918955293567\\
                      & 1105121482204632859666129                 & 2313049101078765671318222\cite{silberstein2015error}\\ 
      \hline
        $A_8(19,4;8)$ & 34517304881588725672548617783010034714431 & 34517160832562417413613060900666077554631\\
                      & 24750791181363571603492943432             & 69548983834979971097956450817\cite{silberstein2015error}\\ 
      \hline
        $A_9(19,4;8)$ & 29973697026756603314019362623281023278094 & 29973636439636704539558217472416705317244\\
                      & 694159378234061636310178009770385         & 006954378905881075605724274378842\cite{silberstein2015error}\\  
      \hline
        $A_2(19,4;9)$ & 1348002146261417447406857                 & 1289520797394170812563456 \cite{etzion2016optimal}\\ 
      \hline
        $A_3(19,4;9)$ & 150024595081884393007012600338236544880   & 149656439781495352647490043484854812672 \cite{etzion2016optimal}\\ 
      \hline
        $A_4(19,4;9)$ & 14677494853604982256349239711358766939534 & 14672330163008228150280295382725315655876\\
                      & 33631137                                  & 80198656 \cite{etzion2016optimal}\\ 
      \hline
        $A_5(19,4;9)$ & 82857418316609426935359109104035114463034 & 82850622250904212594169271527415901448337\\
                      & 035070083878426                           & 330052257546240 \cite{etzion2016optimal}\\  
      \hline
        $A_7(19,4;9)$ & 40553499771898134480982178003920834382400 & 40553105686976525598537903185027802877404\\
                      & 027150097709434199076454452               & 377818341526587313754210304 \cite{etzion2016optimal}\\ 
      \hline
        $A_8(19,4;9)$ & 17672860099364199843916651708091969078559 & 17672785292630676941870781019917737428555\\
                      & 21985408931899768042741180322945          & 76219388246358086384262766919680 \cite{etzion2016optimal}\\ 
      \hline
        $A_9(19,4;9)$ & 21850825132503488679494454968297144272741 & 21850780456810680205404210364952225732272\\
                      & 734751304289430201060161991357689262      & 903279105184451471542312700846538752 \cite{etzion2016optimal}\\ 
      \hline      
    \end{tabular}
    \caption{Improvements based on Theorem~\ref{thm_EF} cont.} 
  \end{center}
\end{table}

\normalsize

In Table~\ref{table_improved_ef_numerical} we compare the numerical improvements for $A_q(n,4;k)$ based on Theorem~\ref{thm_EF}, Theorem~\ref{thm_construction_d}, and 
Lemma~\ref{lemma_ef_comb_special}. We also list a few cases which result in the same lower bound that was previously known. For consistency reasons we also list the 
obtained lower bounds for $A_q(12,4;4)$ and $A_q(16,4;4)$ which are inferior to the lower bounds obtained in \cite{cossidente2019combining}.

\tiny 

\begin{table}[htp!]
  \begin{center}
    \begin{tabular}{|l|l|l|}
      \hline
      $A_q(n,4;k)$ & New & Old \\ 
      \hline
      \hline
        $A_3(10,4;3)$ & 5086975 & 5086975 \cite{heinlein2019generalized}\\ 
      \hline
        $A_4(10,4;3)$ & 273727509 & 273727509 \cite{heinlein2019generalized}\\ 
      \hline
        $A_5(10,4;3)$ & 6162421906 & 6162421906 \cite{heinlein2019generalized}\\ 
      \hline
        $A_7(10,4;3)$ & 680487816156 & 680487816156 \cite{heinlein2019generalized}\\ 
      \hline
        $A_8(10,4;3)$ & 4407724867657 & 4407724867657 \cite{heinlein2019generalized}\\ 
      \hline
        $A_9(10,4;3)$ & 22911698562814 & 22911698562814 \cite{heinlein2019generalized}\\ 
      \hline
        $A_3(11,4;3)$ & 45782788 & 45782686 \cite{heinlein2019generalized}\\ 
      \hline
        $A_4(11,4;3)$ & 4379640165 & 4379639873 \cite{heinlein2019generalized}\\ 
      \hline
        $A_5(11,4;3)$ & 154060547681 & 154060547001 \cite{heinlein2019generalized}\\ 
      \hline
        $A_7(11,4;3)$ & 33343902991701 & 33343902989195 \cite{heinlein2019generalized}\\ 
      \hline
        $A_8(11,4;3)$ & 282094391530121 & 282094391525889 \cite{heinlein2019generalized}\\ 
      \hline
        $A_9(11,4;3)$ & 1855847583588025 & 1855847583581293 \cite{heinlein2019generalized}\\ 
      \hline
        $A_3(12,4;3)$ & 412045132 & 412044676 \cite{heinlein2019generalized}\\ 
      \hline
        $A_4(12,4;3)$ & 70074242725 & 70074241065 \cite{heinlein2019generalized}\\ 
      \hline
        $A_5(12,4;3)$ & 3851513692181 & 3851513687561 \cite{heinlein2019generalized}\\ 
      \hline
        $A_7(12,4;3)$ & 1633851246593749 & 1633851246571461 \cite{heinlein2019generalized}\\ 
      \hline
        $A_8(12,4;3)$ & 18054041057928329 & 18054041057886353 \cite{heinlein2019generalized}\\ 
      \hline
        $A_9(12,4;3)$ & 150323654270630845 & 150323654270557225 \cite{heinlein2019generalized}\\ 
      \hline
        $A_3(13,4;3)$ & 3708406309 & 3708405100 \cite{heinlein2019generalized}\\ 
      \hline
        $A_4(13,4;3)$ & 1121187883941 & 1121187877725 \cite{heinlein2019generalized}\\ 
      \hline
        $A_5(13,4;3)$ & 96287842305306 & 96287842283141 \cite{heinlein2019generalized}\\ 
      \hline
        $A_7(13,4;3)$ & 80058711083096502 & 80058711082943685 \cite{heinlein2019generalized}\\ 
      \hline
        $A_8(13,4;3)$ & 1155458627707417737 & 1155458627707087193 \cite{heinlein2019generalized}\\ 
      \hline
        $A_9(13,4;3)$ & 12176215995921105826 & 12176215995920451445 \cite{heinlein2019generalized}\\ 
      \hline
        $A_2(14,4;3)$ & 6241671 & 6241671 \cite{heinlein2019generalized}\\ 
      \hline
        $A_3(14,4;3)$ & 33375657145 & 33375648396 \cite{heinlein2019generalized}\\ 
      \hline
        $A_4(14,4;3)$ & 17939006144421 & 17939006056956 \cite{heinlein2019generalized}\\ 
      \hline
        $A_5(14,4;3)$ & 2407196057636556 & 2407196057148120 \cite{heinlein2019generalized}\\ 
      \hline
        $A_7(14,4;3)$ & 3922876843071748206 & 3922876843065022206 \cite{heinlein2019generalized}\\ 
      \hline
        $A_8(14,4;3)$ & 73949352173274772617 & 73949352173255598072 \cite{heinlein2019generalized}\\ 
      \hline
        $A_9(14,4;3)$ & 986273495669609638336 & 986273495669561209956 \cite{heinlein2019generalized}\\ 
      \hline
        $A_3(15,4;3)$ & 300380915398 & 300380802505 \cite{heinlein2019generalized}\\ 
      \hline
        $A_4(15,4;3)$ & 287024098316197 & 287024096598921 \cite{heinlein2019generalized}\\ 
      \hline
        $A_5(15,4;3)$ & 60179901440933431 & 60179901426410736 \cite{heinlein2019generalized}\\ 
      \hline
        $A_7(15,4;3)$ & 192220965310515799351 & 192220965310140305694 \cite{heinlein2019generalized}\\ 
      \hline
        $A_8(15,4;3)$ & 4732758539089585747081 & 4732758539088207997713 \cite{heinlein2019generalized}\\ 
      \hline
        $A_9(15,4;3)$ & 79888153149238381303087 & 79888153149234028941436 \cite{heinlein2019generalized}\\ 
      \hline
        $A_2(16,4;3)$ & 102223687 & 102223687 \cite{heinlein2019generalized}\\ 
      \hline
        $A_3(16,4;3)$ & 2703428238582 & 2703427322125 \cite{heinlein2019generalized}\\ 
      \hline
        $A_4(16,4;3)$ & 4592385573059152 & 4592385547188501 \cite{heinlein2019generalized}\\ 
      \hline
        $A_5(16,4;3)$ & 1504497536023335775 & 1504497535674194406 \cite{heinlein2019generalized}\\ 
      \hline
        $A_7(16,4;3)$ & 9418827300215274168199 & 9418827300197242691478 \cite{heinlein2019generalized}\\ 
      \hline
        $A_8(16,4;3)$ & 302896546501733487813184 & 302896546501646667812937 \cite{heinlein2019generalized}\\ 
      \hline
        $A_9(16,4;3)$ & 6470940405088308885550047 & 6470940405087960642352846 \cite{heinlein2019generalized}\\ 
      \hline
    \end{tabular}
    \caption{Improvements based on Theorem~\ref{thm_EF}, Theorem~\ref{thm_construction_d}, and Lemma~\ref{lemma_ef_comb_special}.} 
    \label{table_improved_ef_linkage_numerical}
  \end{center}
\end{table}

\begin{table}[htp!]
  \begin{center}
    \begin{tabular}{|l|l|l|}
      \hline
      $A_q(n,4;k)$ & New & Old \\ 
      \hline
      \hline
        $A_2(17,4;3)$ & 408894755                                 & 408894729 \cite{heinlein2019generalized}\\ 
      \hline
        $A_3(17,4;3)$ & 24330854147239                            & 24330847680853 \cite{heinlein2019generalized}\\ 
      \hline
        $A_4(17,4;3)$ & 73478169168946433                         & 73478168809292217 \cite{heinlein2019generalized}\\ 
      \hline
        $A_5(17,4;3)$ & 37612438400583394376                      & 37612438392939375961 \cite{heinlein2019generalized}\\ 
      \hline
        $A_7(17,4;3)$ & 461522537710548434241752                  & 461522537709756797516943 \cite{heinlein2019generalized}\\ 
      \hline
        $A_8(17,4;3)$ & 19385378976110943220043777                & 19385378976105915500991697 \cite{heinlein2019generalized}\\ 
      \hline
        $A_9(17,4;3)$ & 524146172812153019729553808               & 524146172812127278980234877 \cite{heinlein2019generalized}\\ 
      \hline
        $A_2(18,4;3)$ & 1635579035                                & 1635578957 \cite{heinlein2019generalized}\\ 
      \hline
        $A_3(18,4;3)$ & 218977687325155                           & 218977629126520 \cite{heinlein2019generalized}\\ 
      \hline
        $A_4(18,4;3)$ & 1175650706703142933                       & 1175650700948669781 \cite{heinlein2019generalized}\\ 
      \hline
        $A_5(18,4;3)$ & 940310960014584859406                     & 940310959823484378906 \cite{heinlein2019generalized}\\ 
      \hline
        $A_7(18,4;3)$ & 22614604347816873277845856                & 22614604347778083078190614 \cite{heinlein2019generalized}\\ 
      \hline
        $A_8(18,4;3)$ & 1240664254471100366082801737              & 1240664254470778592063165001 \cite{heinlein2019generalized}\\ 
      \hline
        $A_9(18,4;3)$ & 42455839997784394598093858458             & 42455839997782309597398420706 \cite{heinlein2019generalized}\\ 
      \hline
        $A_2(19,4;3)$ & 6542316171                                & 6542316059 \cite{heinlein2019generalized}\\ 
      \hline
        $A_3(19,4;3)$ & 1970799185926408                          & 1970798662145206 \cite{heinlein2019generalized}\\ 
      \hline
        $A_4(19,4;3)$ & 18810411307250286949                      & 18810411215178781533 \cite{heinlein2019generalized}\\ 
      \hline
        $A_5(19,4;3)$ & 23507774000364621485181                   & 23507773995587109861266 \cite{heinlein2019generalized}\\ 
      \hline
        $A_7(19,4;3)$ & 1108115613043026790614447001              & 1108115613041126070837091728 \cite{heinlein2019generalized}\\ 
      \hline
        $A_8(19,4;3)$ & 79402512286150423429299311241             & 79402512286129829892059310425 \cite{heinlein2019generalized}\\ 
      \hline
        $A_9(19,4;3)$ & 3438923039820535962445602535189           & 3438923039820367077389315073958 \cite{heinlein2019generalized}\\ 
      \hline
        $A_2(12,4;4)$ & 19674269                                  & 19676797 \cite{cossidente2019combining}\\ 
      \hline
        $A_3(12,4;4)$ & 288648673507                              & 288648887023 \cite{cossidente2019combining}\\ 
      \hline
        $A_4(12,4;4)$ & 283104148286289                           & 283104153226065 \cite{cossidente2019combining}\\ 
      \hline
        $A_5(12,4;4)$ & 59732550564570151                         & 59732550620930151 \cite{cossidente2019combining}\\ 
      \hline
        $A_7(12,4;4)$ & 191677878196845899475                     & 191677878199060649103 \cite{cossidente2019combining}\\ 
      \hline
        $A_8(12,4;4)$ & 4723722950504908124737                    & 4723722950514423444033 \cite{cossidente2019combining}\\ 
      \hline
        $A_9(12,4;4)$ & 79780441020720237308359                   & 79780441020754680563815 \cite{cossidente2019combining}\\ 
      \hline
        $A_2(13,4;4)$ & 157396313                                 & 157332190 \cite{cossidente2019combining}\\ 
      \hline
        $A_3(13,4;4)$ & 7793514240823                             & 7793495430036 \cite{cossidente2019combining}\\ 
      \hline
        $A_4(13,4;4)$ & 18118665490931521                         & 18118664249474716 \cite{cossidente2019combining}\\ 
      \hline
        $A_5(13,4;4)$ & 7466568820575245751                       & 7466568787180077320 \cite{cossidente2019combining}\\ 
      \hline
        $A_7(13,4;4)$ & 65745512221518213208951                   & 65745512216555289614188 \cite{cossidente2019combining}\\ 
      \hline
        $A_8(13,4;4)$ & 2418546150658513179095553                 & 2418546150622126921477496 \cite{cossidente2019combining}\\ 
      \hline
        $A_9(13,4;4)$ & 58159941504105053602711351                & 58159941503893673245551936 \cite{cossidente2019combining}\\ 
      \hline
        $A_2(14,4;4)$ & 1259181253                                & 1258757174 \cite{cossidente2019combining}\\ 
      \hline
        $A_3(14,4;4)$ & 210424885316173                           & 210424421624298 \cite{cossidente2019combining}\\ 
      \hline
        $A_4(14,4;4)$ & 1159594591440766481                       & 1159594516050838620 \cite{cossidente2019combining}\\ 
      \hline
        $A_5(14,4;4)$ & 933321102572187566901                     & 933321098538702991570 \cite{cossidente2019combining}\\ 
      \hline
        $A_7(14,4;4)$ & 22550710691980761977054117                & 22550710690309028764671498 \cite{cossidente2019combining}\\ 
      \hline
        $A_8(14,4;4)$ & 1238295629137158820145963073              & 1238295629118788686643907448 \cite{cossidente2019combining}\\ 
      \hline
        $A_9(14,4;4)$ & 42398597356492584370698238141             & 42398597356340204444957848530 \cite{cossidente2019combining}\\ 
      \hline
        $A_2(15,4;4)$ & 10073483841                               & 10071464646 \cite{cossidente2019combining}\\ 
      \hline
        $A_3(15,4;4)$ & 5681471907358915                          & 5681463153275925 \cite{cossidente2019combining}\\ 
      \hline
        $A_4(15,4;4)$ & 74214053852334056793                      & 74214050169101548368 \cite{cossidente2019combining}\\ 
      \hline
        $A_5(15,4;4)$ & 116665137821525417847661                  & 116665137415279661027650 \cite{cossidente2019combining}\\ 
      \hline
        $A_7(15,4;4)$ & 7734893767349401492485075543              & 7734893766857015258769289566 \cite{cossidente2019combining}\\ 
      \hline
        $A_8(15,4;4)$ & 634007362118225316643319878225            & 634007362109986775858834010688 \cite{cossidente2019combining}\\ 
      \hline
        $A_9(15,4;4)$ & 30908577472883094009493870125553          & 30908577472784286989399940957138 \cite{cossidente2019combining}\\ 
      \hline      
        $A_2(16,4;4)$ & 80596312221                               & 80596320222 \cite{cossidente2019combining}\\ 
      \hline
        $A_3(16,4;4)$ & 153399853228893616                        & 153399853246113244 \cite{cossidente2019combining}\\ 
      \hline
        $A_4(16,4;4)$ & 4749699529175914867537                    & 4749699529177178098513 \cite{cossidente2019combining}\\ 
      \hline
        $A_5(16,4;4)$ & 14583142241438194910273276                & 14583142241438230122883276 \cite{cossidente2019combining}\\ 
      \hline
        $A_7(16,4;4)$ & 2653068562231495127604331943392           & 2653068562231495132921600978204 \cite{cossidente2019combining}\\ 
      \hline
        $A_8(16,4;4)$ & 324611769405185201386954976457281         & 324611769405185201425928435085889 \cite{cossidente2019combining}\\ 
      \hline
        $A_9(16,4;4)$ & 22532352977741502934357344934823692       & 22532352977741502934583323007263948 \cite{cossidente2019combining}\\ 
      \hline     
    \end{tabular}
    \caption{Improvements based on Theorem~\ref{thm_EF}, Theorem~\ref{thm_construction_d}, and Lemma~\ref{lemma_ef_comb_special} cont.} 
  \end{center}
\end{table}
      
\begin{table}[htp!]
  \begin{center}
    \begin{tabular}{|l|l|l|}
      \hline
      $A_q(n,4;k)$ & New & Old \\ 
      \hline
      \hline
        $A_2(17,4;4)$ & 644770526929 & 644769492958 \cite{cossidente2019combining}\\ 
      \hline
        $A_3(17,4;4)$ & 4141796037183639766 & 4141796035658667783 \cite{cossidente2019combining}\\ 
      \hline
        $A_4(17,4;4)$ & 303980769867258670043393 & 303980769866940801875612 \cite{cossidente2019combining}\\ 
      \hline
        $A_5(17,4;4)$ & 1822892780179774365686344376 & 1822892780179753492675780445 \cite{cossidente2019combining}\\ 
      \hline
        $A_7(17,4;4)$ & 910002516845402828768417638272482 & 910002516845402816852352967620731 \cite{cossidente2019combining}\\ 
      \hline
        $A_8(17,4;4)$ & 166201225935454823110121665687982081      & 166201225935454822961083064550157688 \cite{cossidente2019combining}\\ 
      \hline
        $A_9(17,4;4)$ & 16426085320773555639146507673669565108    & 16426085320773555637759638858654591537 \cite{cossidente2019combining}\\ 
      \hline
        $A_2(18,4;4)$ & 5158164354661 & 5158157544758 \cite{cossidente2019combining}\\ 
      \hline
        $A_3(18,4;4)$ & 111828493003995506044 & 111828492966430770027 \cite{cossidente2019combining}\\ 
      \hline
        $A_4(18,4;4)$ & 19454769271504557075730961 & 19454769271485256959951964 \cite{cossidente2019combining}\\ 
      \hline
        $A_5(18,4;4)$ & 227861597522471795764877751276 & 227861597522469274830565882195 \cite{cossidente2019combining}\\ 
      \hline
        $A_7(18,4;4)$ & 312130863277973170267574410726205300 & 312130863277973166253742271546610147 \cite{cossidente2019combining}\\ 
      \hline
        $A_8(18,4;4)$ & 85095027678952869432382343292909789249    & 85095027678952869357138271983639907192 \cite{cossidente2019combining}\\ 
      \hline
        $A_9(18,4;4)$ & 11974616198843922060937804378020488874436 & 11974616198843922059938039662012799420059 \cite{cossidente2019combining}\\ 
      \hline
        $A_2(19,4;4)$ & 41265315376833 & 41265282958278 \cite{cossidente2019combining}\\ 
      \hline
        $A_3(19,4;4)$ & 3019369311108187930600 & 3019369310399000457648 \cite{cossidente2019combining}\\ 
      \hline
        $A_4(19,4;4)$ & 1245105233376291684834977113 & 1245105233375348762973895504 \cite{cossidente2019combining}\\ 
      \hline
        $A_5(19,4;4)$ & 28482699690308974471842016207036 & 28482699690308720567594897355775 \cite{cossidente2019combining}\\ 
      \hline
        $A_7(19,4;4)$ & 107060886104344797401778345454852597360 & 107060886104344796219558793302396711773 \cite{cossidente2019combining}\\ 
      \hline
        $A_8(19,4;4)$ & 43568654171623869149379762750201287709265 & 43568654171623869115634697686019566694976 \cite{cossidente2019combining}\\ 
      \hline
        $A_9(19,4;4)$ & 87294952089572191824236594129317862395535 & 87294952089572191817753865390378237377288\\
                      & 08                                        & 19 \cite{cossidente2019combining}\\ 
      \hline
    \end{tabular}
    \caption{Improvements based on Theorem~\ref{thm_EF}, Theorem~\ref{thm_construction_d}, and Lemma~\ref{lemma_ef_comb_special} cont.} 
  \end{center}
\end{table}

\normalsize

\section{Skeleton codes for the multilevel construction}
\label{sec_skeleton_codes}
In this appendix we list the used skeleton codes from the results of Section~\ref{sec_results}. For the ease of a more compact representation we 
replace each vector $v\in\F_2^n$ by the integer $\sum_{i=1}^n v_i\cdot 2^{n-i}$. As an example, the integer $6168$ corresponds to the vector
$1100000011000\in\F_2^{13}$. Starting from an integer, the value of $n$ needs to be clear from the context.  

\medskip

\tiny

$\cS_{13,4,5}^1=\{$7936, 7360, 6816, 1984, 3488, 3680, 5776, 6496, 6544, 6736, 7216, 
3720, 5456, 5704, 3400, 5000, 5508, 2948, 4912, 5416, 5668, 7180, 2856,
3604, 4932, 1816, 2882, 6666, 1826, 3346, 4802, 2753, 4336, 5218, 5281, 
6406, 6409, 6661, 7171, 1256, 1380, 2706, 2833, 3217, 4514, 4545, 5258, 
740, 929, 1618, 2264, 2388, 2636, 3206, 3333, 4705, 696, 1236, 2274, 
2442, 3114, 1417, 1669, 2228, 4300, 4426, 4636, 4867, 5189, 376, 466,
841, 1202, 1329, 1577, 2598, 5142, 5145, 428, 618, 790, 1347, 1550, 
2161, 2217, 4389, 1084, 2339, 405, 4154, 597, 4243, 651, 563, 2078, 
2123, 118, 199, 109, 283, 1063, 4111$\}$

\medskip

$\cS_{13,4,5}^2=\{$7936, 7360, 6816, 3008, 3488, 3680, 5776, 6496, 6544, 6736, 7216,3720, 5456, 
5512, 5704, 3400, 1840, 1924, 5668, 7180, 3604, 4904, 4932, 4994, 1858, 2840, 2852, 5410, 6666, 
3346, 1264, 1698, 4516, 4801, 5281, 6406, 6409, 6661, 7171, 1380, 1473, 2706, 3217, 4328, 4706, 
4748, 4881, 740, 929, 1801, 2264, 2388, 2444, 2636, 3333, 5254, 696, 1617, 2274, 3114, 3142, 
4308, 1228, 2228, 2609, 2819, 4426, 5146, 5189, 376, 466, 714, 1308, 4274, 4630, 426, 1550, 
2217, 2245, 4209, 409, 617, 4156, 661, 355, 1078, 1081, 309, 333, 2131, 4235, 391, 1099,2078, 
583, 110, 539, 2087$\}$

\medskip

$\cS_{14,4,5}^1=\{$15872, 14720, 13632, 3968, 7456, 11072, 11456, 12992, 13088, 13472, 14432, 11536, 6816, 7312, 5904,
6920, 10896, 3680, 6736, 7240, 9992, 14360, 5768, 9808, 11304, 9860, 13332, 5508, 5700, 6468,
10788, 992, 2864, 3608, 9602, 12812, 12818, 13322, 14342, 1488, 1712, 2760, 3394, 4994, 5232,
5666, 6338, 9508, 10530, 3236, 4528, 4552, 6418, 8872, 9089, 9368, 9761, 10508, 10762, 12561,
1448, 1730, 1857, 4712, 5313, 6284, 7173, 8560, 8644, 8980, 10324, 10401, 920, 2288, 2408,
3590, 4756, 6196, 9314, 11267, 12425, 2452, 2497, 3210, 3337, 6659, 9292, 3122, 3153, 4324,
4450, 5164, 6186, 10313, 12357, 844, 850, 2821, 5379, 4657, 8372, 8402, 8522, 8774, 12323,
628, 721, 810, 1308, 1329, 2601, 5145, 4393, 8729, 678, 1557, 4188, 4250, 421, 1114,
1129, 1174, 2150, 2201, 4366, 234, 345, 1187, 1547, 8250, 8334, 403, 611, 653, 310,
4179, 8455, 542, 2183, 8237, 1095$\}$

\medskip

$\cS_{14,4,5}^2=\{$15872, 14720, 13632, 6016, 6976, 7360, 11552, 12992, 13088, 13472, 14432, 7440, 10912, 11024, 11408,
6800, 3680, 3848, 11336, 14360, 7208, 9808, 9864, 9988, 3716, 5680, 5704, 10820, 13332, 6692,
992, 3396, 9032, 9602, 10562, 12812, 12818, 13322, 14342, 2760, 2946, 3457, 5412, 6434, 8656,
9412, 9496, 9762, 1480, 1858, 3602, 4528, 4776, 4888, 5272, 6412, 9089, 10433, 10762, 12561,
1392, 3234, 4548, 6228, 6305, 6665, 8616, 10380, 11269, 1729, 2288, 2408, 2456, 2468, 2849,
5218, 5638, 6282, 7171, 8852, 10292, 10505, 12425, 1428, 1809, 2616, 5385, 8548, 9313, 852,
908, 4705, 5201, 9260, 12357, 2641, 3100, 3121, 4434, 4869, 5253, 8418, 1234, 8312, 8498,
8753, 12323, 690, 1322, 1577, 4739, 710, 806, 4204, 8771, 10259, 618, 2326, 4262, 8346,
451, 665, 677, 2138, 2150, 8462, 220, 233, 316, 345, 1166, 2197, 4154, 8278, 779,
1078, 2567, 589, 1287, 4149, 542, 2093, 4171, 8327, 115, 8221, 1051$\}$

\medskip

$\cS_{13,4,6}^1=\{$8064, 7776, 7504, 6624, 6864, 6960, 6984, 7344, 7368, 7464, 7704, 3908, 2016, 3752, 5828,
5924, 5954, 6820, 3778, 3874, 6552, 7686, 3492, 3732, 5032, 5794, 3521, 3857, 5524, 6794,
7429, 1944, 3474, 5057, 5537, 5777, 5897, 2956, 2977, 4948, 5010, 5452, 5514, 6264, 6534,
6915, 7299, 2898, 3402, 3657, 5330, 5426, 6697, 6725, 2516, 3188, 3356, 5292, 1656, 1925,
2738, 3621, 2668, 4722, 4764, 5226, 5233, 6246, 948, 2474, 2673, 3225, 4340, 6293, 970,
1478, 1814, 3130, 3214, 4878, 504, 1385, 1638, 3171, 4412, 4442, 4697, 6227, 1244, 1621,
2236, 2266, 2281, 2393, 2405, 6174, 867, 1253, 1675, 2358, 4282, 4493, 4661, 825, 1333,
1363, 1587, 2695, 4302, 4323, 5149, 5191, 723, 845, 1206, 1326, 2589, 3095, 4395, 435,
685, 4205, 4375, 414, 606, 2319, 63$\}$

\medskip

$\cS_{13,4,6}^2=\{$8064, 7776, 7504, 6624, 6864, 6960, 6984, 7344, 7368, 7464, 7704, 3908, 6824, 2016, 5828,
5924, 5954, 3748, 3778, 3874, 6552, 7686, 5794, 3521, 3857, 5524, 7429, 1944, 2964, 3468,
3474, 5057, 5537, 5777, 5897, 2977, 3284, 3380, 3721, 4948, 5004, 5010, 5452, 5514, 6264,
6534, 6789, 6915, 7299, 2898, 2954, 3402, 5330, 5426, 2764, 2860, 4788, 5236, 5292, 1656,
1925, 2676, 2738, 3180, 3186, 3242, 4810, 4906, 4716, 4722, 5226, 6246, 2666, 1478, 1814,
2396, 2417, 3161, 6229, 504, 1265, 1638, 4316, 4412, 4442, 4457, 4697, 5177, 934, 1381,
1621, 1678, 2236, 2266, 2281, 2362, 2617, 6174, 6189, 6195, 6219, 867, 4282, 4325, 1363,
1581, 1587, 1611, 3143, 723, 821, 845, 1206, 1229, 1326, 4679, 5159, 469, 1309, 2599,
435, 459, 683, 795, 374, 429, 669, 1118, 1179, 2327, 414, 4247, 4367, 238, 574,
2191, 123$\}$

\medskip

$\cS_{14,4,6}^1=\{$16128, 15552, 15008, 13248, 13728, 13920, 13968, 14688, 14736, 14928, 15408, 4032, 7816, 11600, 11656,
11848, 11908, 13640, 7556, 7748, 13104, 15372, 5968, 6984, 7464, 3888, 7042, 7490, 7714, 11048,
11076, 11556, 13588, 14858, 6948, 10114, 11137, 11794, 11809, 6017, 7256, 7697, 9896, 10008, 10904,
10946, 11426, 12528, 13068, 13830, 13833, 14598, 14601, 14853, 15363, 5796, 5826, 6804, 6849, 7314,
7329, 9668, 10049, 13394, 13450, 3312, 3852, 5032, 5528, 6376, 6712, 11409, 5476, 5898, 6484,
9124, 9812, 10468, 10804, 13445, 1896, 9432, 9570, 11334, 12492, 12625, 13059, 1940, 2914, 3425,
4824, 5332, 6450, 8680, 10360, 10570, 12586, 1008, 2520, 3276, 3843, 4961, 5660, 9042, 9396,
9441, 10524, 10545, 2897, 4578, 5425, 6342, 10441, 12454, 12457, 12581, 1745, 2484, 2732, 2738,
3354, 4472, 4724, 4934, 5321, 6257, 9324, 9513, 12348, 12483, 1490, 3178, 4785, 6297, 8564,
8626, 10390, 972, 1452, 1650, 1829, 2412, 2474, 2652, 2838, 3132, 3267, 3349, 4714, 4889,
8817, 8901, 9486, 1690, 2501, 2665, 5178, 5389, 8806, 8981, 9491, 10325, 10339, 963, 1372,
1699, 4502, 4750, 5198, 5219, 6190, 6221, 8601, 8762, 12339, 1593, 1677, 8851, 10285, 828,
1254, 2451, 2699, 3123, 4325, 4691, 9291, 1334, 1419, 1582, 4653, 8525, 9245, 10267, 726,
819, 1607, 2166, 4427, 6167, 8583, 9255, 12303, 1141, 2599, 252, 3087, 8286, 243, 359,
783, 1175, 207, 63$\}$

\medskip

$\cS_{14,4,6}^2=\{$16128, 15552, 15008, 13248, 13728, 13920, 13968, 14688, 14736, 14928, 15408, 4032, 7816, 13648, 11656,
11848, 11908, 7496, 7556, 7748, 13104, 15372, 11588, 3888, 7042, 7714, 11048, 14858, 5928, 6936,
6948, 10114, 11074, 11137, 11554, 11794, 11809, 5954, 6017, 6568, 6760, 6977, 7442, 7457, 7697,
9896, 10008, 10020, 10904, 11028, 12528, 13068, 13578, 13830, 13833, 14598, 14601, 14853, 15363, 5796,
5908, 6804, 10049, 10660, 10852, 11537, 3312, 3852, 5528, 5720, 9576, 10472, 10584, 13573, 5352,
5476, 6360, 6372, 6484, 9620, 9812, 9432, 9444, 10452, 12492, 13059, 5332, 1008, 2786, 3276,
3843, 4792, 6322, 12458, 4578, 4818, 4833, 8632, 8824, 8884, 9394, 10354, 10417, 2732, 2769,
3242, 4472, 4532, 4724, 5234, 5297, 6257, 8658, 8673, 8906, 12348, 12378, 12390, 12393, 12438,
12441, 12453, 12483, 1489, 1737, 2506, 8564, 9329, 972, 1452, 1644, 1692, 1734, 2412, 2460,
2652, 3132, 3162, 3174, 3177, 3222, 3225, 3237, 3267, 4553, 6286, 8901, 12373, 2501, 9358,
10318, 10381, 937, 963, 1372, 1699, 2618, 3157, 5198, 5261, 6221, 12339, 874, 922, 934,
1593, 2467, 8803, 8851, 9293, 828, 857, 869, 917, 1338, 1379, 1590, 1619, 2358, 2361,
2613, 2699, 3123, 4499, 4654, 6187, 854, 1333, 1419, 2387, 4683, 4743, 8494, 8734, 8749,
9259, 10267, 10279, 819, 2631, 4382, 4397, 4637, 5147, 5159, 6167, 8523, 8583, 12303, 4423,
8477, 9239, 252, 3087, 243, 783, 207, 63$\}$

\medskip

$\cS_{15,4,6}^1=\{$32256, 31104, 30016, 26496, 27456, 27840, 27936, 29376, 29472, 29856, 30816, 8064, 15632, 23200, 23312,
23696, 23816, 27280, 15112, 15496, 26208, 30744, 11936, 13968, 14928, 7776, 14084, 14980, 15428, 22096,
22152, 23112, 27176, 29716, 13896, 20228, 22274, 23588, 23618, 12034, 14512, 15394, 19792, 20016, 21808,
21892, 22852, 25056, 26136, 27660, 27666, 29196, 29202, 29706, 30726, 11592, 11652, 13608, 13698, 14628,
14658, 19336, 20098, 26788, 26900, 6624, 7704, 10064, 11056, 12752, 13424, 14881, 22818, 28945, 3792,
10952, 11796, 11841, 12968, 13185, 14529, 18248, 19624, 19841, 20936, 21608, 22049, 26890, 3880, 7489,
18864, 19140, 21313, 21697, 22668, 24984, 25250, 26118, 27141, 28809, 2016, 5828, 6850, 9648, 10664,
12900, 17360, 20720, 21140, 25172, 25745, 25865, 26705, 5040, 6552, 7329, 7686, 9922, 11320, 18084,
18792, 18882, 19233, 19977, 21048, 21090, 23043, 28741, 5794, 9156, 10017, 10850, 12684, 13827, 19553,
20882, 20897, 22577, 24908, 24914, 25162, 26755, 2980, 3428, 3490, 3724, 3857, 4968, 5464, 6708,
6801, 8944, 9448, 9665, 10642, 12514, 12641, 13445, 18648, 19026, 24696, 24966, 25137, 25347, 28707,
1944, 2904, 3009, 5474, 5897, 6356, 6917, 12594, 17128, 17252, 20780, 25637, 25667, 3658, 4824,
4948, 5304, 5521, 5676, 6264, 6534, 6698, 9428, 9778, 9865, 13337, 14357, 17634, 17802, 18972,
2744, 3252, 5002, 5330, 10356, 10778, 17612, 17962, 18982, 20650, 20678, 1844, 3356, 4833, 6441,
9004, 9542, 10396, 10438, 10565, 12380, 12442, 14347, 17202, 17524, 17989, 19477, 21517, 21523, 24678,
1656, 1925, 2508, 2886, 3186, 3273, 3621, 9033, 9498, 10521, 17702, 20570, 1713, 2668, 3226,
4902, 5398, 5413, 5653, 6246, 8650, 9382, 10345, 10531, 18582, 1452, 1859, 2673, 6419, 8982,
9493, 11271, 17050, 17177, 18490, 18701, 20534, 20747, 1393, 4332, 4549, 5262, 8869, 12334, 12371,
12551, 17091, 17166, 17497, 18510, 24606, 504, 937, 1638, 1686, 2282, 2723, 3339, 6221, 8846,
2453, 3118, 4686, 4749, 6174, 8851, 9293, 16572, 1253, 2358, 5195, 8409, 16725, 16739, 17543,
4659, 9259, 16941, 486, 725, 2259, 4277, 8733, 8775, 1563, 2221, 4381, 16563, 1181, 2583,
619, 1118, 365, 411, 574, 207, 1079$\}$

\medskip

$\cS_{15,4,6}^2=\{$32256, 31104, 30016, 26496, 27456, 27840, 27936, 29376, 29472, 29856, 30816, 8064, 15632, 27296, 23312,
23696, 23816, 14992, 15112, 15496, 26208, 30744, 23176, 7776, 14084, 22096, 23620, 29716, 11856, 13872,
13896, 14916, 15396, 15426, 20228, 22148, 22274, 11908, 12034, 13136, 13520, 13954, 19792, 20016, 20040,
21808, 22056, 23076, 23106, 23586, 25056, 26136, 27156, 27660, 27666, 29196, 29202, 29706, 30726, 11592,
11816, 13608, 14882, 20098, 21320, 21704, 6624, 7704, 11056, 11440, 14657, 19152, 20944, 21168, 21889,
27146, 28945, 10704, 10952, 11649, 12720, 12744, 12968, 13185, 19240, 19624, 22721, 22817, 14497, 18864,
18888, 19329, 20904, 24984, 26118, 26769, 26889, 28809, 2016, 10664, 3524, 6552, 7686, 9584, 12644,
13409, 18241, 24916, 25681, 25861, 28741, 9921, 10017, 17264, 17348, 17648, 17768, 17828, 17858, 18788,
19553, 20708, 20834, 21089, 2980, 3010, 3490, 3745, 3857, 5464, 5524, 6484, 7249, 7429, 8944,
9064, 9448, 10468, 10594, 10849, 12514, 13445, 13571, 21253, 24696, 24756, 24780, 24786, 24876, 24882,
24906, 24966, 25137, 25161, 25221, 25347, 25641, 25731, 26661, 26691, 28707, 1944, 2904, 3288, 3384,
5026, 5777, 5897, 9108, 9612, 9618, 11013, 17128, 18658, 19589, 19715, 4824, 4920, 5004, 5304,
5514, 6264, 6324, 6348, 6354, 6444, 6450, 6474, 6534, 6705, 6729, 6789, 6915, 7209, 7299,
11333, 12572, 13337, 14357, 17298, 18057, 21123, 21541, 21571, 24746, 1876, 2744, 9098, 10883, 12837,
12867, 18716, 19013, 19481, 20636, 20762, 21017, 22541, 22547, 3188, 5446, 5701, 6314, 10396, 10522,
10777, 11299, 12442, 14347, 24678, 1656, 1716, 1740, 1746, 1836, 1842, 1866, 1925, 3651, 4724,
5228, 5234, 9030, 9414, 9510, 18586, 18979, 1489, 2668, 2674, 3178, 3350, 4806, 5286, 5667,
6246, 9308, 9749, 12374, 17190, 1706, 3597, 4714, 4886, 8870, 16988, 17468, 17498, 17558, 17678,
17939, 18518, 20534, 20558, 945, 969, 1449, 2417, 2710, 2830, 3214, 4549, 8764, 8794, 9274,
9739, 10294, 10318, 12334, 24606, 504, 1638, 4337, 4457, 8613, 8643, 16954, 17038, 18478, 1381,
2281, 2453, 6174, 8537, 16803, 17485, 867, 4499, 16601, 16697, 16781, 741, 1251, 2381, 2443,
4405, 8377, 8405, 16723, 16941, 486, 1566, 1309, 2227, 4269, 4299, 5143, 8491, 16501, 795,
2247, 2343, 8301, 8307, 669, 1179, 4189, 4623, 8471, 16491, 414, 2109, 2139, 599, 4155,
8335, 1071, 126$\}$

\medskip

$\cS_{16,4,6}^1=\{$64512, 62208, 60032, 52992, 54912, 55680, 55872, 58752, 58944, 59712, 61632, 16128, 31264, 46400, 46624,
47392, 47632, 54560, 30224, 30992, 52416, 61488, 23872, 27936, 29856, 15552, 28168, 29960, 30856, 44192,
44304, 46224, 54352, 59432, 27792, 40456, 44548, 47176, 47236, 24068, 29024, 30788, 39584, 40032, 43616,
43784, 45704, 50112, 52272, 55320, 55332, 58392, 58404, 59412, 61452, 23184, 23304, 27216, 27396, 29256,
29316, 38672, 40196, 53576, 53800, 13248, 15408, 20128, 22112, 25504, 26848, 29762, 45636, 57890, 7584,
21904, 23592, 23682, 25936, 26370, 29058, 36496, 39248, 39682, 41872, 43216, 44098, 53780, 4032, 7760,
14978, 27713, 37728, 38280, 40065, 42626, 42753, 43394, 45336, 45441, 49968, 50500, 52236, 54282, 57618,
57633, 57873, 61443, 11656, 13700, 19296, 21328, 22273, 25800, 26241, 27009, 34720, 41440, 42280, 50344,
51490, 51730, 51745, 53410, 10080, 13104, 14658, 14913, 15372, 19844, 22640, 36168, 37584, 37764, 38466,
39954, 42096, 42180, 46086, 46089, 53521, 54277, 57482, 11588, 18312, 20034, 21700, 23569, 25368, 27654,
39106, 41764, 41794, 45154, 49816, 49828, 50324, 53345, 53510, 3888, 5960, 6856, 6980, 7448, 7714,
9936, 10928, 13416, 13602, 17888, 18896, 19330, 21284, 22721, 22817, 25028, 25282, 25409, 26890, 28817,
36417, 37296, 38052, 49392, 49932, 50274, 50694, 50697, 51345, 51465, 52227, 53385, 57414, 57417, 57477,
5808, 6018, 7041, 10948, 11794, 11809, 12712, 13834, 14497, 25188, 28933, 34256, 34504, 41560, 41665,
45137, 51274, 51334, 7316, 9648, 9896, 10608, 11042, 11352, 12528, 13068, 13396, 13585, 14601, 15363,
18856, 19556, 19730, 26674, 27139, 28714, 35268, 35604, 36129, 37944, 39429, 43145, 43269, 3688, 5488,
6504, 10004, 10660, 13829, 20712, 21153, 21556, 35224, 35924, 37964, 41300, 41356, 43057, 51269, 3312,
3852, 5825, 6712, 9666, 11025, 12882, 18008, 19084, 20792, 20876, 21130, 24760, 24884, 25649, 28694,
34404, 35048, 35978, 37641, 38954, 43034, 43046, 45093, 49356, 49923, 5016, 5772, 6372, 6546, 7242,
18066, 18996, 19593, 21042, 21577, 26649, 26661, 34570, 35404, 38147, 41140, 2904, 3426, 5336, 6452,
9804, 10796, 10826, 11306, 11397, 12492, 13059, 17300, 18764, 20690, 21062, 25129, 37164, 37425, 3473,
3718, 5346, 7237, 12838, 14357, 17964, 18986, 19017, 22542, 34100, 34354, 36980, 37402, 37507, 41068,
41494, 42005, 42019, 1008, 2786, 3276, 3843, 8664, 9098, 9441, 10524, 12825, 17738, 21539, 24668,
24742, 25102, 25613, 34182, 34332, 34437, 34994, 35865, 37020, 37189, 49212, 49347, 1873, 4564, 5446,
5673, 6678, 9545, 12442, 17617, 17692, 17989, 20803, 41485, 2529, 2769, 4906, 6236, 9093, 9372,
9498, 9795, 12348, 12483, 17702, 17705, 17795, 18586, 21005, 33144, 33969, 35139, 2506, 2857, 3605,
4716, 5413, 6297, 10345, 10390, 16818, 18709, 24675, 33450, 33478, 34981, 35086, 41107, 41227, 972,
1452, 3132, 3267, 4553, 5651, 6667, 8817, 9318, 17573, 18518, 18595, 20565, 33882, 33897, 49203,
1477, 1689, 4518, 4885, 6243, 8554, 16753, 17466, 17550, 33193, 34061, 963, 1699, 2661, 4442,
6407, 8762, 10323, 12339, 16793, 20615, 33126, 36939, 1622, 2362, 5166, 5259, 6189, 8549, 8611,
9479, 33173, 33365, 33379, 745, 828, 2445, 3123, 8505, 16979, 20507, 32985, 49167, 730, 5149,
8405, 33863, 36887, 819, 1206, 2222, 4217, 4277, 12303, 1141, 8365, 16494, 16935, 17431, 252,
3087, 243, 414, 663, 783, 343, 207, 63$\}$

\medskip

$\cS_{16,4,6}^2=\{$64512, 62208, 60032, 52992, 54912, 55680, 55872, 58752, 58944, 59712, 61632, 16128, 31264, 54592, 46624,
47392, 47632, 29984, 30224, 30992, 52416, 61488, 46352, 15552, 28168, 44192, 47240, 59432, 23712, 27744,
27792, 29832, 30792, 30852, 40456, 44296, 44548, 23816, 24068, 26272, 27040, 27908, 39584, 40032, 40080,
43616, 44112, 46152, 46212, 47172, 50112, 52272, 54312, 55320, 55332, 58392, 58404, 59412, 61452, 23184,
23632, 27216, 29764, 40196, 42640, 43408, 13248, 15408, 22112, 22880, 29314, 38304, 41888, 42336, 43778,
54292, 57890, 21408, 21904, 23298, 25440, 25488, 25936, 26370, 27393, 38480, 39248, 45442, 45634, 45697,
4032, 28994, 29057, 29249, 37728, 37776, 38658, 39681, 41808, 42753, 49968, 52236, 53538, 53778, 53793,
57618, 57633, 57873, 61443, 21328, 22273, 45377, 7048, 13104, 15372, 19168, 25288, 26818, 36482, 49832,
51362, 51722, 53521, 57482, 7554, 7746, 7809, 10120, 11080, 11140, 34528, 35296, 35536, 37576, 39106,
41416, 41668, 42178, 43201, 3888, 5960, 6020, 6980, 10928, 11586, 11649, 11810, 11841, 12968, 14498,
14858, 17888, 18128, 18896, 19240, 20936, 21188, 21698, 22721, 25028, 25793, 26890, 27142, 42506, 43529,
49392, 49512, 49560, 49572, 49752, 49764, 49812, 49932, 50274, 50322, 50337, 50442, 50694, 50697, 51282,
51297, 51345, 51462, 51465, 51717, 52227, 53322, 53382, 53385, 57414, 57417, 57477, 5808, 6576, 6768,
7489, 10052, 19746, 19986, 20001, 22026, 23049, 34256, 34600, 35608, 35620, 37316, 38081, 39178, 39430,
3752, 9648, 9840, 10608, 12528, 12648, 12696, 12708, 12888, 12900, 12948, 13068, 13410, 13458, 13473,
13578, 13830, 13833, 14418, 14433, 14481, 14598, 14601, 14853, 15363, 18200, 18212, 19220, 22666, 25144,
25865, 26117, 26674, 26761, 28714, 36114, 36129, 36369, 42246, 43082, 43142, 43269, 49492, 50257, 50437,
53317, 5488, 19729, 21766, 22789, 25674, 25734, 34580, 37432, 38026, 38153, 38405, 38962, 41272, 41524,
42034, 42121, 43057, 45082, 45094, 45097, 3312, 3432, 3480, 3492, 3672, 3684, 3732, 3852, 6376,
11402, 12628, 13393, 13573, 19084, 20792, 21044, 21554, 22577, 22598, 24884, 25649, 26693, 28694, 28697,
28709, 38985, 39045, 49356, 49923, 9448, 10456, 10468, 21577, 21637, 34444, 35212, 35404, 37172, 37937,
37958, 42053, 45077, 2978, 3412, 5336, 5348, 6356, 6700, 7210, 12492, 13059, 17804, 17996, 18616,
18764, 24748, 1953, 9428, 9772, 10540, 10780, 11305, 33976, 34124, 34936, 34996, 35878, 37036, 41068,
41116, 1008, 1890, 1938, 2898, 2913, 2961, 3276, 3843, 4834, 5420, 5660, 6428, 7190, 7193,
7205, 10883, 17290, 17528, 17588, 18548, 20588, 20636, 24668, 49212, 49347, 1873, 8674, 8914, 8929,
9500, 11285, 18051, 18819, 19011, 33610, 33670, 33673, 33908, 36956, 2762, 4562, 4577, 4817, 4906,
6691, 12348, 12483, 17074, 17222, 17225, 17285, 17578, 18538, 24739, 25099, 34179, 34371, 34970, 34985,
35139, 1737, 8657, 8986, 8998, 9001, 10531, 17731, 18582, 18597, 33202, 33394, 33457, 33605, 37027,
37387, 41059, 41107, 41227, 41479, 972, 1482, 1734, 2502, 2505, 2757, 3132, 3267, 4762, 4886,
4889, 4901, 5411, 5651, 6419, 8810, 16754, 16817, 17009, 17561, 18521, 20579, 20627, 20747, 20999,
24659, 24839, 33882, 33897, 49203, 1477, 4713, 8857, 8981, 9491, 17494, 17509, 33137, 33941, 34901,
36947, 37127, 963, 2618, 4454, 8598, 8846, 12339, 33002, 1593, 4494, 4501, 4686, 4749, 8549,
16602, 16614, 16617, 828, 1338, 1590, 2358, 2361, 2613, 3123, 8378, 8526, 8589, 8781, 16942,
32982, 32985, 32997, 49167, 1333, 4218, 4278, 4281, 4429, 10267, 16597, 33070, 33310, 33325, 819,
2222, 8310, 8313, 8373, 9255, 12303, 16670, 16685, 16925, 1197, 4213, 33053, 252, 683, 1134,
1182, 2142, 2157, 2205, 3087, 423, 1117, 243, 363, 411, 603, 615, 663, 783, 343,
207, 63$\}$

\medskip

$\cS_{17,4,6}=\{$129024, 124416, 120064, 105984, 109824, 111360, 111744, 117504, 117888, 119424, 123264, 32256, 62528, 109184, 93248,
94784, 95264, 59968, 60448, 61984, 104832, 122976, 92704, 31104, 56336, 88384, 94480, 118864, 47424, 55488,
55584, 59664, 61584, 61704, 80912, 88592, 89096, 47632, 48136, 52544, 54080, 55816, 79168, 80064, 80160,
87232, 88224, 92304, 92424, 94344, 100224, 104544, 108624, 110640, 110664, 116784, 116808, 118824, 122904, 46368,
47264, 54432, 59528, 80392, 85280, 86816, 26496, 30816, 44224, 45760, 58628, 76608, 83776, 84672, 87556,
108584, 115780, 42816, 43808, 46596, 50880, 50976, 51872, 52740, 54786, 76960, 78496, 90884, 91268, 91394,
8064, 57988, 58114, 58498, 75456, 75552, 77316, 79362, 83616, 85506, 99936, 104472, 107076, 107556, 107586,
115236, 115266, 115746, 122886, 42656, 44546, 90754, 14096, 26208, 29953, 30744, 38336, 50576, 53636, 72964,
99664, 102724, 103444, 103489, 107042, 114964, 115009, 115729, 118789, 23684, 23810, 36624, 38544, 38664, 39684,
69056, 70592, 71072, 75152, 78212, 82832, 83336, 84356, 86402, 7776, 11920, 12040, 13960, 15106, 15428,
15490, 21856, 22096, 25936, 27393, 27777, 28996, 29313, 29716, 35776, 36256, 37792, 40193, 41872, 42376,
43396, 45442, 50056, 51586, 53780, 54284, 72324, 85012, 87058, 98784, 99024, 99120, 99144, 99504, 99528,
99624, 99864, 100548, 100644, 100674, 100884, 100929, 101388, 101394, 101409, 102564, 102594, 102690, 102924, 102930,
102945, 103434, 104454, 106644, 106689, 106764, 106770, 106785, 107025, 107529, 108549, 110595, 114828, 114834, 114849,
114954, 115209, 116739, 11616, 13152, 13536, 20104, 23108, 23170, 39972, 40002, 44052, 46098, 68512, 69200,
71216, 71240, 72449, 72833, 74632, 76162, 78356, 78860, 7504, 14884, 19296, 19680, 20016, 21216, 22056,
23617, 25056, 25296, 25392, 25416, 25776, 25800, 25896, 26136, 26820, 26916, 26946, 27156, 27660, 27666,
28836, 28866, 28962, 29196, 29202, 29706, 30726, 36424, 39553, 45332, 50288, 51730, 52234, 53348, 53522,
57428, 72258, 72738, 84492, 86164, 86284, 86538, 98984, 100514, 100874, 102673, 106634, 10976, 14913, 15393,
39458, 43532, 45578, 51348, 51468, 69160, 74864, 76052, 76306, 76810, 77924, 82544, 83048, 84068, 84242,
86114, 90164, 90188, 90194, 6624, 6864, 6960, 6984, 7344, 7368, 7464, 7704, 12752, 21784, 22804,
23073, 25256, 26786, 27146, 41584, 42088, 43108, 45154, 45196, 49768, 51298, 53386, 57388, 57394, 57418,
77970, 78090, 98712, 99846, 100497, 100617, 102537, 14484, 14604, 18896, 20912, 20936, 43154, 43274, 68888,
70424, 70808, 74344, 75874, 75916, 84106, 90154, 3908, 6824, 10672, 10696, 12712, 13400, 14609, 19224,
21144, 24984, 26118, 28753, 35992, 37232, 38996, 42257, 49496, 49937, 50441, 51281, 53297, 53321, 99589,
2016, 5828, 5924, 5954, 14418, 18856, 19544, 22673, 22793, 37464, 37944, 67952, 68248, 69872, 69992,
71732, 71756, 74072, 82136, 82232, 3748, 3778, 3874, 6552, 7686, 9668, 10840, 11320, 12856, 13457,
13574, 14473, 18196, 19729, 22572, 22578, 22602, 26673, 26697, 28713, 35056, 35176, 37096, 41176, 41272,
41737, 45125, 49336, 50309, 51241, 70921, 71761, 86041, 98424, 98694, 98949, 99075, 99459, 5794, 14378,
17348, 17828, 17858, 21254, 21641, 35384, 36102, 38022, 45081, 67220, 67340, 67346, 67816, 70289, 73912,
3476, 9124, 9154, 9634, 9812, 11345, 11525, 12934, 13061, 18066, 18186, 21574, 24696, 24966, 25669,
34148, 34444, 43075, 49478, 49733, 49795, 50198, 50243, 68358, 68742, 69970, 71721, 75801, 82981, 1944,
5012, 5516, 5522, 17314, 17996, 34356, 34386, 35462, 37164, 37194, 37446, 66404, 66788, 66914, 67210,
68681, 68867, 74054, 74774, 82118, 82214, 82454, 82958, 2956, 2962, 3466, 6264, 6534, 9524, 9772,
9778, 9802, 10822, 10885, 11302, 12617, 13379, 17620, 25125, 25155, 25635, 33508, 33620, 33634, 34018,
35122, 35909, 37042, 37073, 41137, 41158, 41254, 41494, 41998, 49318, 49445, 49678, 67794, 68145, 70182,
70213, 82467, 98406, 5002, 9426, 9987, 10793, 17234, 17714, 18982, 34346, 34988, 35018, 37923, 38925,
38931, 66274, 67882, 68227, 69802, 73894, 74254, 1873, 3188, 5425, 6469, 7189, 8908, 12835, 17076,
17196, 20572, 20677, 20803, 24678, 24853, 34076, 35365, 66004, 66769, 69925, 98389, 1656, 1925, 4724,
5228, 5234, 8988, 9002, 9372, 9498, 17098, 17578, 33204, 33228, 33234, 67139, 1713, 1737, 1833,
2668, 2674, 3178, 3241, 6246, 6309, 6339, 6435, 7179, 16756, 20643, 24717, 24723, 24843, 33562,
35093, 65964, 65970, 65994, 66652, 70663, 98334, 98349, 98355, 98379, 2883, 3597, 3603, 4714, 8436,
8556, 8562, 8858, 9381, 12342, 12366, 17468, 18518, 19463, 33194, 2396, 16620, 16626, 16746, 24606,
33433, 33881, 37127, 49181, 66108, 66138, 66325, 66618, 504, 1637, 4883, 8426, 16954, 16985, 41003,
66701, 2645, 4284, 4314, 4410, 4749, 6174, 9355, 1366, 2234, 2265, 2361, 67655, 1357, 18459,
20503, 483, 723, 822, 846, 1206, 1230, 1326, 1566, 10263, 65717, 4381, 683, 65819, 8285,
407, 2319, 32911, 16463, 4143, 119$\}$

\medskip

$\cS_{15,4,7}^1=\{$32512, 31936, 31392, 29632, 30112, 30304, 30352, 31072, 31120, 31312, 31792, 28296, 30032, 12224, 15752,
15944, 16004, 31756, 20384, 23432, 23880, 23940, 24104, 24132, 24194, 27440, 27464, 27524, 27944, 27972,
28034, 28196, 28226, 29480, 31242, 8032, 8080, 14128, 15172, 15234, 15652, 15682, 15906, 20304, 29460,
15128, 22216, 23332, 23362, 23704, 23842, 26056, 26308, 26392, 26497, 27736, 27796, 28177, 28912, 29962,
30214, 30217, 30982, 30985, 31237, 31747, 11952, 13992, 26402, 27329, 27809, 29324, 13764, 14018, 14145,
15137, 15444, 15506, 15633, 21424, 22180, 23188, 26856, 27282, 29957, 11632, 12052, 13672, 14092, 14785,
15457, 21954, 22290, 22305, 23313, 23634, 29004, 29234, 7873, 11176, 13220, 14552, 14564, 14900, 14993,
19696, 20236, 21860, 22840, 22868, 22945, 23137, 25444, 26962, 29443, 12042, 19905, 21336, 22754, 25304,
25784, 25826, 25908, 25953, 27164, 28874, 28977, 6896, 7400, 7945, 13154, 14642, 18280, 20968, 21716,
21905, 22097, 22732, 25482, 26740, 27177, 7942, 13012, 13492, 18324, 19172, 20018, 21624, 21729, 22044,
22737, 23589, 25044, 25268, 25425, 25809, 26822, 28838, 28841, 28869, 5616, 6050, 7729, 9200, 10130,
10978, 11468, 12728, 12920, 13025, 13193, 13850, 14506, 14915, 18904, 19128, 20041, 20101, 20227, 21202,
21382, 21682, 24952, 25708, 26917, 28732, 5972, 11593, 11653, 11845, 11907, 12754, 13426, 13513, 13699,
13861, 19628, 19658, 19761, 21108, 25004, 25010, 26777, 28758, 28761, 3800, 3896, 7045, 7452, 7493,
10676, 12660, 13590, 19242, 19738, 21162, 21317, 25194, 26682, 28771, 28819, 3556, 5962, 6828, 6858,
7366, 9844, 10054, 10698, 10860, 11046, 11370, 11430, 17378, 20850, 20892, 25742, 25875, 9164, 9884,
10025, 10865, 11324, 11555, 14393, 14414, 14477, 18097, 18796, 19222, 19558, 21145, 22582, 25027, 25251,
25357, 26701, 3028, 3041, 3538, 5073, 5548, 5740, 6513, 6748, 7257, 7317, 7331, 9642, 9649,
10588, 10842, 10917, 18074, 18854, 21582, 21645, 22667, 3753, 3877, 6506, 6553, 6758, 7226, 10894,
11027, 12686, 12693, 13358, 17353, 17756, 17833, 18115, 19541, 19603, 25675, 2994, 5530, 5782, 6550,
9562, 9817, 11347, 12877, 18006, 18021, 18837, 19022, 22557, 22599, 25629, 3683, 4969, 9020, 12935,
13383, 14375, 17317, 18787, 18997, 19027, 20782, 21037, 21067, 21547, 24862, 25159, 25639, 26647, 28687,
1905, 1989, 5433, 5475, 5771, 9430, 9445, 9557, 10453, 12589, 12619, 12843, 16881, 17966, 20661,
24875, 25115, 1778, 1945, 2905, 3467, 4922, 8678, 8681, 8890, 9779, 17209, 17299, 18553, 21527,
3382, 3406, 4581, 4917, 5331, 6683, 10419, 10567, 17109, 17799, 18590, 20763, 1934, 3613, 5237,
6253, 7183, 8444, 8915, 16846, 2765, 4787, 16634, 17523, 986, 4345, 9371, 9743, 18703, 1703,
1815, 2477, 4342, 4523, 8563, 2711, 4445, 4879, 16621, 886, 1389, 2283, 2622, 4695, 17039,
1262, 2363, 16727, 1245, 4303, 701, 2167, 8367, 415, 623, 8287, 1087$\}$

\medskip

$\cS_{15,4,7}^2=\{$32512, 31936, 31392, 29632, 30112, 30304, 30352, 31072, 31120, 31312, 31792, 28296, 30032, 12224, 15944,
23944, 24196, 27976, 28036, 28228, 29488, 31756, 15240, 16002, 20384, 24104, 15656, 15908, 23368, 23876,
24130, 27432, 27522, 28194, 31242, 8032, 8080, 14212, 15172, 15682, 15745, 20304, 12080, 15512, 15636,
22216, 22402, 23332, 23425, 23842, 26392, 26404, 26434, 26497, 27412, 27457, 27922, 27937, 28177, 28912,
29452, 29962, 30214, 30217, 30982, 30985, 31237, 31747, 13768, 13992, 15138, 23192, 23640, 14145, 14756,
14996, 15041, 22292, 22978, 23314, 23698, 23713, 23825, 26856, 29957, 14098, 14680, 14904, 15121, 15457,
21416, 21864, 22305, 11504, 12044, 13160, 13668, 13908, 19396, 20161, 21953, 22756, 22840, 23137, 25816,
25828, 26836, 28876, 29443, 11682, 11937, 14562, 19184, 20234, 7400, 7945, 11858, 13016, 13028, 13217,
20020, 21346, 21688, 21716, 21730, 22066, 22097, 25272, 25314, 26802, 28842, 5872, 6640, 7942, 13432,
13492, 13522, 13873, 14452, 17904, 20229, 20952, 22737, 5976, 9200, 11601, 12035, 12728, 12756, 14513,
18316, 19660, 19756, 21112, 21172, 21202, 22642, 24952, 25012, 25042, 25057, 25204, 25297, 25714, 25777,
26737, 28732, 28762, 28774, 28777, 28822, 28825, 28837, 28867, 7708, 7749, 10122, 10956, 11082, 11436,
11466, 12978, 3800, 3812, 11141, 13509, 13593, 19116, 19146, 19249, 19626, 19843, 20852, 20914, 21317,
26766, 28757, 3048, 5932, 6604, 7366, 7450, 7461, 7715, 10922, 11356, 12658, 13603, 18249, 21273,
3512, 3540, 3553, 6764, 6825, 6979, 9836, 9884, 9926, 9929, 10025, 10604, 10652, 10694, 10697,
11366, 11414, 13454, 17866, 19036, 19561, 19609, 21605, 22606, 22669, 25678, 25741, 26701, 28723, 3762,
5073, 5553, 6570, 6822, 7317, 9132, 10857, 10905, 11322, 14382, 14419, 14475, 18026, 18074, 18086,
18089, 18794, 18842, 18854, 18857, 19046, 19094, 19542, 21134, 21155, 22581, 2018, 2996, 3026, 3697,
5066, 5532, 5737, 5785, 5797, 5827, 6505, 6553, 6745, 7225, 9578, 10005, 12878, 12899, 12941,
17260, 17756, 19002, 20885, 21550, 21643, 25134, 25227, 26667, 3442, 5542, 5734, 5782, 6502, 6550,
6742, 6803, 7222, 9052, 9532, 9637, 9667, 10806, 10837, 11317, 12691, 13357, 13387, 14365, 17766,
17814, 18069, 18195, 18777, 18789, 1908, 6332, 6362, 6485, 9062, 9110, 9558, 9875, 10469, 10553,
10595, 12597, 12621, 17242, 17317, 17347, 17722, 19027, 19507, 20819, 21037, 21067, 21533, 22555, 24862,
24877, 24907, 24967, 25117, 25159, 25627, 25639, 26647, 28687, 1989, 2979, 3723, 9018, 18659, 3469,
3662, 5454, 6451, 12827, 13335, 17206, 17717, 20779, 1516, 2958, 3283, 3403, 3629, 4581, 9043,
21015, 2893, 4999, 7183, 8444, 8885, 17101, 1497, 3607, 4725, 4894, 5235, 9387, 10767, 16634,
17607, 1831, 2759, 2843, 3239, 4345, 17935, 18703, 953, 4342, 9487, 16629, 8435, 1627, 2286,
2398, 875, 1598, 2269, 1246, 2235, 1213, 4303, 471, 687, 2167, 446, 1135, 8351, 381,
16479, 4159$\}$

\medskip

$\cS_{16,4,7}=\{$65024, 63872, 62784, 59264, 60224, 60608, 60704, 62144, 62240, 62624, 63584, 56592, 60064, 31504, 32008,
40832, 48272, 55952, 56072, 56456, 58976, 63512, 16192, 31824, 46864, 48388, 24256, 24352, 30352, 31368,
31876, 47696, 48200, 54864, 55044, 56388, 62484, 16032, 44808, 46728, 47748, 47874, 28048, 28420, 30280,
30466, 31300, 40544, 47408, 47656, 52784, 52808, 52868, 52994, 54824, 54914, 55844, 55874, 56354, 57824,
58904, 59924, 60428, 60434, 61964, 61970, 62474, 63494, 30000, 30896, 43920, 44368, 46660, 28200, 29572,
30244, 31012, 31042, 31266, 31777, 44674, 45896, 46376, 46466, 53712, 54657, 55617, 59914, 61713, 26448,
27344, 28226, 28289, 44580, 45744, 46192, 46626, 47298, 15896, 22408, 23008, 23938, 27522, 29128, 29296,
29890, 30913, 42704, 43720, 44200, 44417, 44609, 45953, 51632, 51656, 52097, 53672, 55457, 57752, 58886,
59537, 59657, 61577, 13792, 14800, 24084, 27969, 29505, 39748, 40258, 45508, 46273, 12080, 15650, 23656,
26308, 26992, 27048, 27076, 27748, 27810, 36320, 37856, 40100, 40460, 40466, 42416, 42440, 42818, 48131,
50544, 50628, 51009, 53604, 54369, 57684, 58449, 58629, 61509, 11232, 24074, 25520, 27425, 29090, 29345,
40225, 43248, 43368, 43428, 44130, 45288, 15116, 15633, 15878, 18400, 20248, 22936, 23128, 25840, 25960,
26020, 28900, 39586, 39617, 39992, 41840, 41896, 42785, 45410, 45665, 49904, 50024, 50084, 50114, 50408,
50594, 50849, 51428, 51554, 51809, 53474, 57464, 57524, 57548, 57554, 57644, 57650, 57674, 57734, 57905,
57929, 57989, 58115, 58409, 58499, 59429, 59459, 61475, 7600, 7624, 13720, 14994, 15412, 15498, 23697,
36504, 36628, 38548, 39256, 39316, 39697, 42340, 6096, 11714, 11924, 13250, 14753, 21848, 21908, 22114,
22289, 22868, 23302, 23813, 25320, 25444, 26250, 26377, 37784, 38666, 40017, 43402, 43570, 43781, 44081,
53532, 54297, 55317, 57514, 7024, 7080, 13652, 13906, 14922, 20114, 26162, 38104, 38284, 41700, 43590,
46117, 7524, 7945, 11096, 11480, 11660, 11852, 13016, 13112, 14085, 14456, 14540, 14726, 14897, 15429,
19348, 19873, 21394, 21688, 22738, 22834, 23173, 26826, 26908, 27673, 28828, 28954, 28969, 29209, 30733,
30739, 36242, 38196, 39212, 39305, 51356, 51482, 51737, 53402, 55307, 57446, 4036, 5992, 6052, 10146,
10177, 13642, 14665, 19668, 19764, 19788, 19794, 20049, 21204, 21300, 21324, 21708, 21804, 22700, 23081,
23619, 25926, 28778, 34648, 35640, 35724, 37714, 38449, 38473, 39028, 42268, 45148, 45222, 45334, 45589,
3824, 6884, 7393, 10936, 12035, 13484, 13490, 13865, 13955, 18136, 19170, 21620, 22085, 25386, 25801,
35540, 36044, 36394, 37496, 37580, 37766, 42182, 43205, 50268, 50454, 50709, 53334, 54279, 11145, 12724,
19148, 19244, 19274, 19754, 20006, 21194, 22641, 25041, 26729, 26774, 34530, 34705, 36169, 37996, 38058,
38213, 39018, 41628, 41754, 41766, 42138, 43068, 43098, 43299, 45114, 45325, 45579, 3937, 7452, 7717,
9848, 10028, 10868, 11546, 11561, 12908, 12970, 14438, 14489, 18100, 18310, 21670, 25180, 25254, 25366,
25660, 25690, 25891, 26682, 27655, 28726, 29191, 34610, 36134, 37298, 37322, 37673, 41578, 42509, 42515,
45203, 49724, 49754, 49814, 49934, 50234, 50318, 50699, 51254, 51278, 51719, 53294, 57374, 3978, 7043,
7282, 7366, 9172, 10674, 12997, 18680, 19121, 19269, 20913, 20934, 25283, 25749, 26709, 34420, 34476,
35281, 35436, 39054, 5958, 6812, 6938, 9905, 11429, 12657, 13091, 13401, 17866, 18028, 19981, 35270,
35497, 35651, 36035, 37478, 37541, 38163, 41411, 41635, 42038, 42251, 49497, 49557, 3448, 3785, 6597,
9586, 11350, 13411, 18213, 18854, 21050, 21771, 34501, 35941, 35989, 37105, 38989, 39175, 1976, 3026,
5065, 5545, 5724, 5785, 7309, 8696, 9830, 10894, 11411, 14366, 17266, 18883, 19030, 19219, 20713,
24973, 34154, 37052, 38963, 41177, 41273, 49337, 49547, 3747, 5571, 6741, 6755, 8946, 9450, 9622,
10597, 10645, 11310, 12502, 14379, 16884, 17321, 17817, 19595, 20835, 21075, 21534, 21555, 24885, 25133,
34213, 35475, 3644, 3674, 5774, 6457, 6710, 9065, 9125, 9813, 10574, 10835, 11339, 12590, 12683,
17977, 22567, 33260, 34382, 35002, 35381, 35870, 37406, 43031, 49269, 49479, 4965, 5013, 5434, 6483,
10467, 17637, 17749, 24755, 34567, 41133, 2918, 4572, 5333, 6325, 7195, 8678, 9107, 9549, 9758,
10781, 12615, 13335, 33621, 34937, 41075, 49259, 49319, 1521, 2537, 3257, 3379, 3463, 3655, 6347,
10791, 16826, 17619, 35101, 37063, 1875, 1933, 2492, 3254, 20573, 20759, 33971, 1002, 1750, 2859,
4787, 5415, 6671, 12349, 17181, 33142, 33483, 36955, 4346, 16750, 18493, 18523, 18703, 20623, 24655,
33182, 892, 1276, 3165, 8437, 17051, 33853, 950, 967, 1454, 1643, 2459, 8975, 9275, 9359,
16606, 16999, 33879, 761, 2286, 5199, 8382, 8539, 33339, 34863, 1374, 1591, 17455, 2263, 687,
4215, 607, 319$\}$

\medskip

$\cS_{17,4,7}=\{$130048, 127744, 125568, 118528, 120448, 121216, 121408, 124288, 124480, 125248, 127168, 113184, 120128, 63008, 64016,
81664, 96544, 111904, 112144, 112912, 117952, 127024, 32384, 63648, 93728, 96776, 48512, 48704, 60704, 62736,
63752, 95392, 96400, 109728, 110088, 112776, 124968, 32064, 89616, 93456, 95496, 95748, 56096, 56840, 60560,
60932, 62600, 81088, 94816, 95312, 105568, 105616, 105736, 105988, 109648, 109828, 111688, 111748, 112708, 115648,
117808, 119848, 120856, 120868, 123928, 123940, 124948, 126988, 60000, 61792, 87840, 88736, 93320, 56400, 59144,
60488, 62024, 62084, 62532, 89348, 91792, 92752, 92932, 96322, 107424, 109314, 111234, 119828, 123426, 52896,
54688, 56452, 56578, 63553, 89160, 91488, 92384, 93252, 94596, 31792, 44816, 46016, 47876, 55044, 58256,
58592, 59780, 60161, 61826, 85408, 87440, 88400, 88834, 89218, 91906, 94849, 103264, 103312, 104194, 105217,
107344, 108289, 110914, 110977, 111169, 115504, 117772, 119074, 119314, 119329, 123154, 123169, 123409, 126979, 27584,
29600, 48168, 55938, 59010, 79496, 80516, 91016, 92546, 24160, 31300, 47312, 52616, 53984, 54096, 54152,
55496, 55620, 63494, 72640, 75712, 80200, 80920, 80932, 84832, 84880, 85636, 94529, 101088, 101256, 102018,
107208, 108738, 115368, 116898, 117258, 119057, 123018, 22464, 30465, 47745, 48148, 51040, 54850, 56353, 58180,
58690, 86496, 86736, 86856, 88260, 90576, 96261, 15200, 15248, 30232, 31105, 31266, 31756, 36800, 40496,
45872, 46256, 51680, 51920, 52040, 57800, 79172, 79234, 79426, 79984, 83680, 83792, 85570, 87617, 90820,
91330, 99808, 100048, 100168, 100228, 100816, 101188, 101698, 101761, 101953, 102856, 103108, 103618, 104641, 106948,
107713, 114928, 115048, 115096, 115108, 115288, 115300, 115348, 115468, 115810, 115858, 115873, 115978, 116230, 116233,
116818, 116833, 116881, 116998, 117001, 117253, 117763, 118858, 118918, 118921, 122950, 122953, 123013, 12192, 27440,
29988, 30824, 30996, 47394, 61522, 61585, 73008, 73256, 77096, 77441, 78512, 78632, 80402, 84680, 14048,
14160, 23428, 23848, 26500, 28225, 29506, 43696, 43816, 44228, 44578, 45736, 46401, 46604, 47626, 50640,
50888, 52500, 52545, 52754, 54465, 54801, 55394, 58049, 61701, 75568, 77332, 79393, 80034, 86804, 87140,
87562, 107064, 108594, 110634, 115028, 115793, 115973, 118853, 15048, 27304, 27812, 29844, 40228, 51073, 52324,
59529, 76208, 76568, 77121, 83400, 87180, 92675, 8072, 11984, 12104, 22192, 22960, 23320, 23704, 26032,
26224, 28170, 28912, 29080, 29452, 29794, 30858, 31747, 38696, 39746, 42788, 43376, 45476, 45668, 46346,
53652, 53816, 54537, 57656, 57908, 58418, 59441, 61481, 72484, 76392, 78424, 79121, 80134, 88114, 94234,
94246, 102712, 102964, 103474, 104497, 106804, 107569, 110614, 110617, 110629, 114892, 115459, 7648, 13768, 14760,
14788, 15634, 15649, 15889, 20292, 20354, 23361, 27284, 29330, 39528, 39576, 39588, 40098, 42408, 42600,
42648, 43416, 43608, 45400, 51852, 57556, 59418, 69296, 71280, 71448, 71553, 75428, 76577, 76898, 76946,
78056, 84417, 84536, 90296, 90444, 90668, 90698, 91178, 14114, 15878, 21872, 22978, 23825, 26968, 26980,
27730, 27910, 27913, 36272, 38340, 43240, 43457, 44129, 44170, 45833, 46595, 50772, 51602, 57681, 71080,
72088, 72788, 72865, 74992, 75160, 75532, 84364, 86410, 87089, 87299, 88105, 92185, 100536, 100908, 101418,
106668, 108558, 7876, 14904, 22290, 24069, 25448, 26049, 26273, 29265, 29769, 38296, 38488, 38548, 39508,
40012, 42388, 45282, 50082, 53548, 54298, 54310, 69060, 69410, 71330, 72290, 75992, 76116, 76426, 78036,
78597, 83256, 83508, 83532, 84276, 86136, 86196, 86598, 90228, 91158, 7956, 11714, 11876, 15434, 15497,
19696, 20056, 20257, 21736, 23092, 23122, 23185, 25816, 25940, 28876, 28978, 36200, 36620, 38561, 39308,
39329, 42578, 43340, 43654, 44291, 50360, 50508, 50732, 51320, 51380, 51782, 53364, 53426, 55317, 57452,
57638, 58382, 58389, 72268, 73219, 74596, 74644, 75169, 82849, 83156, 84298, 85018, 85030, 88078, 99448,
99508, 99628, 99868, 100468, 100636, 101398, 101401, 101413, 102508, 102556, 103438, 104461, 106588, 107533, 114748,
114883, 11916, 13624, 13876, 18344, 21348, 23622, 25994, 37360, 38242, 38538, 41826, 41868, 50566, 51498,
51753, 68840, 68952, 69393, 70872, 72033, 72273, 72329, 75090, 75345, 76169, 76358, 78108, 86553, 6896,
7826, 14516, 14561, 14673, 15107, 15429, 19810, 22150, 22153, 25314, 26182, 26802, 26833, 29221, 35732,
36056, 36497, 42629, 43589, 45257, 47123, 50481, 53571, 53774, 55307, 57869, 70540, 71890, 74956, 82822,
83270, 83587, 84076, 84502, 85013, 86309, 90389, 98994, 99114, 100899, 106659, 107019, 3952, 6052, 7074,
11448, 11572, 11665, 12037, 13194, 13522, 19172, 21834, 22700, 26822, 36426, 37708, 39241, 40003, 42100,
45169, 45210, 45589, 51397, 51493, 53413, 69002, 70482, 71978, 72069, 74210, 74450, 74465, 75290, 78350,
79883, 7753, 9200, 10130, 11090, 13484, 13701, 19660, 21788, 22822, 26947, 28732, 28867, 34532, 35681,
37586, 37766, 38469, 39110, 39194, 41426, 43286, 49778, 49841, 49946, 50723, 53337, 57443, 57611, 68308,
70985, 74104, 74569, 75046, 75973, 78115, 82354, 82757, 84133, 84259, 86179, 90259, 98674, 98737, 98929,
99094, 99097, 99109, 99619, 99859, 100627, 102499, 102547, 102667, 102919, 106579, 106759, 114739, 5985, 7288,
7377, 7558, 9700, 11817, 14428, 18900, 19244, 19337, 19845, 21194, 21290, 21297, 21829, 22641, 25004,
25475, 29195, 33768, 35634, 36012, 38275, 42150, 42265, 43068, 43110, 43173, 50281, 71190, 76301, 78855,
82586, 83097, 10121, 12660, 12914, 12956, 12998, 17368, 18130, 19626, 20945, 22678, 25180, 25366, 25369,
25891, 26131, 34641, 35954, 38166, 45134, 49577, 49941, 68764, 70762, 70853, 71740, 74812, 74857, 77965,
82289, 82537, 86099, 98538, 98958, 100491, 3809, 6953, 9930, 11046, 12721, 14605, 18249, 21189, 25701,
34630, 35274, 35498, 35529, 35715, 37290, 37948, 39013, 50317, 66488, 68721, 69134, 70086, 71181, 74393,
74901, 5076, 10865, 12650, 12905, 13475, 17332, 18076, 18214, 19050, 19516, 19651, 25146, 25230, 25657,
26670, 26701, 37667, 38966, 41321, 41747, 42131, 45191, 49370, 49382, 66932, 67025, 67270, 67273, 67369,
68262, 68419, 70309, 70339, 75019, 75829, 82573, 98518, 98521, 98533, 98638, 98701, 98893, 99403, 99463,
100423, 114703, 4984, 5554, 5740, 6508, 6514, 10051, 10947, 11353, 12694, 20841, 21651, 24809, 33652,
34249, 35238, 39175, 43083, 49803, 66506, 68186, 68877, 69873, 70202, 74126, 74294, 82989, 2018, 3020,
3025, 3866, 5065, 5364, 6926, 7435, 9585, 19107, 19219, 20822, 21091, 25735, 34499, 35057, 35925,
36359, 37517, 41518, 50247, 67942, 67990, 68197, 70795, 82222, 3497, 6378, 7221, 8940, 9157, 10554,
10659, 13342, 19029, 20877, 21255, 28695, 34138, 34213, 34362, 37205, 41082, 41145, 42027, 66396, 68243,
69942, 74375, 1784, 1969, 2552, 5689, 6547, 7214, 9020, 11315, 14375, 18745, 20665, 21547, 24757,
33500, 34571, 35381, 35411, 37451, 66028, 67171, 74029, 82471, 2981, 3429, 3670, 3723, 5715, 6707,
6791, 8668, 9550, 9805, 10782, 16874, 66746, 67335, 67790, 67811, 69806, 71703, 73838, 73931, 82014,
82203, 3302, 3411, 6349, 6685, 12615, 17230, 17530, 17590, 18550, 18718, 20590, 33212, 33998, 2748,
4570, 4581, 5421, 12571, 16764, 34974, 49303, 66846, 1934, 4790, 5671, 10397, 11279, 24667, 33251,
66355, 69725, 1749, 4725, 5277, 8883, 9003, 18587, 18959, 73787, 1494, 3179, 4779, 6235, 8438,
9495, 17501, 36923, 41015, 1709, 1821, 2503, 3223, 4702, 10327, 17679, 33447, 65963, 8807, 16627,
16807, 16957, 33179, 66647, 2397, 5199, 34863, 877, 919, 65725, 859, 1243, 1339, 2359, 4311,
8317, 33823, 719, 66079, 8367$\}$

\medskip

$\cS_{18,4,7}=\{$260096, 255488, 251136, 237056, 240896, 242432, 242816, 248576, 248960, 250496, 254336, 226368, 240256, 126016, 128032,
163328, 193088, 223808, 224288, 225824, 235904, 254048, 64768, 127296, 187456, 193552, 97024, 97408, 121408, 125472,
127504, 190784, 192800, 219456, 220176, 225552, 249936, 64128, 179232, 186912, 190992, 191496, 112192, 113680, 121120,
121864, 125200, 162176, 189632, 190624, 211136, 211232, 211472, 211976, 219296, 219656, 223376, 223496, 225416, 231296,
235616, 239696, 241712, 241736, 247856, 247880, 249896, 253976, 120000, 123584, 175680, 177472, 186640, 112800, 118288,
120976, 124048, 124168, 125064, 178696, 183584, 185504, 185864, 192644, 214848, 218628, 222468, 239656, 246852, 105792,
109376, 112904, 113156, 127106, 178320, 182976, 184768, 186504, 189192, 63584, 89632, 92032, 95752, 110088, 116512,
117184, 119560, 120322, 123652, 170816, 174880, 176800, 177668, 178436, 183812, 189698, 206528, 206624, 208388, 210434,
214688, 216578, 221828, 221954, 222338, 231008, 235544, 238148, 238628, 238658, 246308, 246338, 246818, 253958, 55168,
59200, 96336, 111876, 118020, 158992, 161032, 182032, 185092, 48320, 62600, 94624, 105232, 107968, 108192, 108304,
110992, 111240, 126988, 145280, 151424, 160400, 161840, 161864, 169664, 169760, 171272, 189058, 202176, 202512, 204036,
210177, 214416, 216321, 217476, 217857, 218241, 230736, 233796, 234516, 234561, 238114, 246036, 246081, 246801, 249861,
44928, 60930, 62977, 95490, 96296, 102080, 109700, 112706, 116360, 117380, 172992, 173472, 173712, 176520, 181152,
192522, 30400, 30496, 60464, 62210, 62532, 63512, 73600, 80992, 91744, 92512, 93697, 103360, 103840, 104080,
115600, 158344, 158468, 158852, 159968, 167360, 167584, 171140, 175234, 178241, 181640, 182660, 189473, 199616, 200096,
200336, 200456, 201632, 202376, 203396, 203522, 203906, 205712, 206216, 207236, 207617, 208001, 209282, 209537, 213896,
215426, 215681, 229856, 230096, 230192, 230216, 230576, 230600, 230696, 230936, 231620, 231716, 231746, 231956, 232001,
232460, 232466, 232481, 233636, 233666, 233762, 233996, 234002, 234017, 234506, 235526, 237716, 237761, 237836, 237842,
237857, 238097, 238601, 239621, 241667, 245900, 245906, 245921, 246026, 246281, 247811, 16192, 54880, 59976, 61648,
61992, 94788, 119169, 123044, 123170, 146016, 146512, 154192, 154882, 157024, 157264, 160804, 169360, 28096, 28320,
46856, 47696, 53000, 56450, 59012, 60161, 60545, 87392, 87632, 88456, 89156, 91472, 92802, 93208, 95252,
101280, 101776, 105000, 105090, 105508, 105601, 108930, 109602, 110788, 116098, 123402, 151136, 154664, 158786, 160068,
173608, 173825, 174280, 175124, 181889, 214128, 217188, 221268, 230056, 231586, 231946, 233745, 237706, 24336, 30096,
54608, 55624, 59688, 80456, 102146, 104193, 104648, 116289, 119058, 123909, 152416, 153136, 154242, 166800, 170625,
174360, 174465, 184849, 185350, 188689, 15776, 16016, 23488, 44384, 45920, 46640, 47408, 52064, 52448, 56340,
57824, 58160, 58904, 59588, 61716, 63494, 77392, 79492, 85576, 86752, 90952, 91336, 92545, 92692, 107304,
107632, 109074, 115312, 115816, 116836, 118882, 122962, 144968, 152784, 156848, 158242, 160268, 176228, 188468, 188492,
188933, 205424, 205928, 206948, 208994, 213608, 215138, 221228, 221234, 221258, 229784, 230918, 231569, 231689, 233609,
24200, 27536, 29520, 29576, 31268, 31298, 31778, 31809, 40708, 46722, 54568, 56353, 58660, 79056, 79152,
79176, 80196, 84816, 85200, 85296, 86832, 87216, 90800, 94401, 96261, 103704, 107393, 115112, 118836, 138592,
142560, 142896, 143106, 150856, 153154, 153796, 153892, 156112, 157057, 168834, 169072, 171025, 180592, 180888, 181336,
181396, 182356, 184417, 184585, 188553, 28228, 31756, 40520, 43744, 45956, 47650, 53936, 53960, 55460, 55820,
55826, 61601, 72544, 76680, 86480, 86914, 88258, 88340, 91666, 93190, 101544, 103204, 110689, 115362, 122929,
122953, 142160, 144176, 145576, 145730, 149984, 150320, 151064, 157889, 157985, 161795, 168728, 170273, 172820, 174178,
174598, 176210, 184370, 201072, 201816, 202836, 208977, 213336, 215121, 217116, 217137, 217161, 230661, 15912, 23940,
29808, 39816, 44580, 46465, 48138, 50896, 52098, 52546, 55489, 58530, 59538, 76592, 76976, 77096, 79016,
80024, 80417, 84776, 88593, 90564, 91425, 94473, 100164, 107096, 108596, 108620, 138820, 139009, 142660, 144580,
151984, 152232, 152852, 152897, 156072, 156481, 157194, 166512, 167016, 167064, 168552, 172272, 172392, 173196, 180456,
182316, 30868, 30994, 39392, 40112, 43472, 46184, 46244, 46370, 51632, 51880, 54081, 57752, 57956, 72400,
73240, 73282, 77122, 78616, 78658, 85156, 86680, 87308, 88582, 94723, 100720, 101016, 101464, 102640, 102760,
103564, 106728, 106852, 110634, 114904, 115276, 116764, 116778, 138440, 144536, 144929, 146438, 149192, 149288, 150338,
154629, 165698, 166312, 168596, 170036, 170060, 176156, 184451, 188483, 198896, 199016, 199256, 199736, 200936, 201272,
202796, 202802, 202826, 205016, 205112, 206876, 206897, 206921, 208922, 208937, 213176, 215066, 215081, 229496, 229766,
230021, 230147, 230531, 15234, 15640, 22344, 27248, 27752, 31241, 36688, 42696, 47244, 48133, 51988, 59473,
74720, 76484, 77076, 79377, 80145, 83652, 83736, 87313, 88161, 89091, 90897, 101132, 102996, 103506, 104529,
104709, 108035, 118809, 137680, 137904, 138024, 138786, 141744, 144066, 144546, 144658, 150180, 150690, 152257, 152338,
152716, 153745, 156216, 156321, 167433, 172641, 173106, 173137, 7904, 14792, 29032, 29122, 29346, 30214, 30858,
39620, 44300, 44306, 44561, 45761, 50628, 52364, 52497, 53604, 53666, 54417, 58001, 58442, 71464, 72112,
72994, 85258, 87178, 90514, 94246, 100962, 101893, 107142, 107548, 110614, 115738, 141080, 143780, 145673, 145925,
149912, 165644, 166540, 167174, 168152, 169004, 169475, 170026, 172618, 180778, 182297, 197988, 198228, 201798, 205893,
213318, 213573, 214038, 214053, 214083, 15124, 22896, 23144, 23330, 23640, 26388, 27044, 38344, 38568, 43668,
45400, 45841, 53644, 54533, 55555, 72852, 75416, 75585, 78482, 80006, 84200, 90338, 90420, 90761, 91178,
102794, 102986, 103473, 106826, 106833, 108569, 140964, 143809, 143956, 144138, 148420, 148900, 148930, 150580, 156700,
156745, 159766, 159769, 180433, 14576, 18400, 20260, 22180, 26968, 27402, 27913, 29281, 40034, 43576, 45644,
53894, 54345, 57464, 57734, 59429, 69064, 75172, 75532, 76938, 78220, 78388, 82852, 83393, 85129, 86572,
92227, 99556, 99682, 99892, 101446, 106674, 114886, 115222, 136136, 136616, 138401, 141970, 148208, 149138, 150092,
150281, 151946, 156230, 164708, 165514, 166661, 168266, 168518, 172358, 180518, 197348, 197474, 197858, 198188, 198194,
198218, 199238, 199718, 201254, 204998, 205094, 205334, 205349, 205379, 205838, 205859, 213158, 213518, 213539, 229478,
10192, 11970, 13924, 15412, 19880, 20289, 23201, 28856, 29445, 29829, 38674, 39316, 42388, 42580, 42594,
43658, 45282, 47145, 47235, 50008, 50593, 50950, 58390, 61459, 71330, 71361, 72024, 72129, 72969, 76433,
76550, 83601, 84300, 84530, 84741, 86136, 86220, 86346, 86661, 87081, 100562, 100577, 100913, 104483, 114981,
134932, 137634, 142380, 152602, 157710, 157717, 165172, 165217, 166194, 166225, 166981, 168325, 174099, 181261, 12065,
15497, 15875, 20116, 21736, 25320, 25828, 25912, 25996, 34736, 36260, 39521, 41890, 44105, 45356, 50360,
50732, 50738, 51593, 51782, 52262, 53553, 57669, 69220, 71908, 76332, 78979, 90268, 99154, 99169, 99721,
99882, 134849, 137612, 141524, 141706, 143480, 149253, 149624, 149714, 155930, 164578, 165074, 165425, 168113, 172198,
197076, 197457, 197841, 197916, 197937, 197961, 198981, 199701, 200901, 200982, 200997, 201027, 201237, 201741, 201747,
205077, 213141, 213261, 213267, 229461, 7618, 10152, 11112, 11916, 13652, 13906, 14689, 21396, 25442, 27269,
29210, 36498, 39250, 42378, 46105, 46147, 51402, 53797, 57891, 59403, 67440, 69281, 69393, 71058, 71244,
71430, 72329, 74580, 75153, 75896, 78026, 78885, 83715, 86341, 100634, 107021, 136849, 138268, 140172, 140937,
142362, 148369, 148705, 148786, 149802, 152611, 155761, 156691, 164745, 6084, 7076, 7946, 11201, 11480, 13016,
13217, 19796, 19858, 20050, 25300, 25810, 26161, 26950, 34664, 36152, 38100, 38481, 39032, 42209, 43589,
49889, 50292, 50460, 51314, 53340, 53402, 67522, 68500, 75334, 77932, 78001, 82642, 83494, 84145, 84262,
90382, 98740, 98764, 133864, 136524, 136838, 137801, 140618, 140678, 141865, 144405, 150038, 151658, 165146, 166435,
173063, 197036, 197042, 197066, 197276, 197297, 197321, 197402, 197417, 197786, 197801, 198806, 198821, 198851, 198926,
198947, 199181, 199187, 199691, 200846, 200867, 201227, 204941, 204947, 205067, 213131, 229406, 229421, 229427, 229451,
6040, 7569, 13852, 14634, 14725, 19160, 19340, 21897, 22058, 22706, 23109, 25388, 26243, 34753, 38278,
38533, 41682, 42289, 43185, 43302, 49618, 50053, 67304, 68961, 69254, 70545, 75305, 78350, 79883, 84165,
84249, 86166, 86307, 90213, 99606, 102549, 103431, 132592, 133026, 134049, 134482, 136372, 137754, 138307, 139746,
140404, 148252, 148588, 151717, 165978, 172089, 7729, 11601, 19170, 21830, 22825, 26793, 29707, 36380, 37304,
38438, 39449, 41644, 41801, 42182, 43331, 49961, 50709, 51305, 51470, 53771, 70114, 70772, 71045, 71850,
72718, 75034, 76069, 76813, 76819, 83036, 99097, 100494, 106759, 134537, 139985, 140585, 141590, 143651, 143891,
148761, 164444, 165029, 5104, 5985, 6988, 6993, 7468, 14428, 17880, 18314, 20009, 20961, 21108, 22041,
25923, 26684, 27155, 35658, 36049, 37674, 41754, 42510, 43285, 50371, 51349, 51363, 57390, 70473, 74410,
82164, 82290, 83469, 84054, 84491, 99011, 102477, 132792, 134022, 136486, 136579, 137329, 139884, 140163, 141465,
147913, 148750, 149605, 151637, 152071, 155911, 3938, 5940, 7493, 9848, 10034, 11378, 11651, 11798, 12722,
13094, 14438, 19654, 25009, 25753, 28750, 36419, 37329, 37490, 41330, 41765, 43094, 45141, 45323, 45575,
49514, 66520, 67000, 68298, 70342, 70762, 74185, 74902, 75939, 83107, 86099, 90247, 134693, 136005, 137493,
140483, 148058, 149645, 164291, 164942, 166027, 10117, 10612, 11430, 11461, 12906, 13582, 17336, 18897, 19569,
19982, 21564, 28725, 33748, 35436, 35529, 41385, 49939, 50443, 57373, 68483, 70257, 70809, 74531, 74841,
75861, 82281, 82521, 98723, 135622, 136297, 139612, 143406, 143437, 143499, 147676, 147862, 155691, 164028, 164406,
164499, 180247, 5496, 5834, 6812, 7322, 7446, 7693, 9140, 10652, 11033, 13370, 19747, 21262, 21605,
25230, 25685, 26701, 35060, 35270, 35494, 37100, 37436, 38051, 41180, 42131, 53319, 70057, 71737, 71827,
90139, 114703, 131960, 132050, 135410, 147875, 149555, 149579, 167963, 167975, 6851, 9162, 9580, 10476, 10842,
12697, 12949, 14419, 19514, 22574, 25142, 33528, 35164, 37645, 67365, 68201, 69948, 77853, 77895, 82485,
98489, 98606, 133692, 134233, 134286, 139497, 139621, 141447, 148615, 164173, 164907, 3754, 8696, 11321, 11342,
12857, 20710, 20794, 20825, 21139, 22663, 25645, 34394, 34457, 34595, 38942, 38987, 49334, 66412, 67804,
100375, 132549, 132710, 133478, 133781, 134291, 135779, 135821, 139450, 147765, 148013, 2994, 5361, 6377, 6742,
9073, 9450, 10851, 12502, 12517, 22557, 24971, 25163, 34405, 35173, 35893, 37174, 37219, 49357, 66501,
68871, 98855, 132883, 135481, 139603, 147574, 6919, 9404, 9995, 10805, 16876, 17347, 17813, 18019, 34137,
35002, 35219, 35975, 66958, 67981, 73846, 82030, 82093, 139445, 139805, 1780, 2801, 3669, 4954, 6357,
6539, 7211, 8934, 9805, 17238, 17253, 33678, 41163, 66033, 66901, 67155, 67211, 67923, 68126, 164367,
1884, 3299, 4572, 5459, 6350, 6451, 17741, 20683, 21527, 35867, 41069, 66358, 131918, 132218, 132638,
133421, 137231, 1942, 2538, 3196, 3403, 4794, 5325, 5703, 17582, 18590, 18983, 34093, 34989, 35115,
49243, 66170, 66741, 69726, 69915, 73835, 1513, 1849, 5405, 8661, 9767, 12567, 33461, 33950, 34119,
66859, 131758, 133403, 1510, 2489, 2862, 4533, 8622, 8883, 9031, 10767, 16634, 18711, 33907, 66669,
66759, 131701, 132253, 135319, 4717, 17053, 17067, 37007, 41015, 66622, 67687, 131783, 131879, 1459, 1709,
2715, 4462, 4775, 33181, 33387, 69687, 133207, 2471, 9359, 66199, 66319, 131390, 131435, 8509, 33367,
131643, 132151, 32894, 32999, 65755, 131471, 16509, 16599, 734, 2171, 4639, 16687, 351, 191$\}$

\medskip

$\cS_{19,4,7}=\{$520192, 510976, 502272, 474112, 481792, 484864, 485632, 497152, 497920, 500992, 508672, 452736, 480512, 195584, 255104,
383104, 387136, 447616, 448576, 451648, 471808, 508096, 129536, 243840, 256032, 385664, 250496, 254528, 325120, 325888,
373888, 382016, 386080, 438912, 440352, 451104, 499872, 128256, 227392, 242752, 250912, 251920, 193280, 248192, 250176,
355456, 358432, 373312, 374800, 381472, 422272, 422464, 422944, 423952, 438592, 439312, 446752, 446992, 450832, 462592,
471232, 479392, 483424, 483472, 495712, 495760, 499792, 507952, 220288, 223872, 242208, 371072, 378240, 226320, 236096,
239936, 240656, 254216, 356672, 367648, 373024, 379168, 379408, 381200, 429696, 437256, 444936, 479312, 493704, 225568,
234880, 238464, 241936, 247312, 342656, 349824, 356880, 357384, 385284, 118528, 127168, 210560, 218688, 222528, 224264,
225800, 236552, 248324, 310336, 322576, 351248, 364096, 365440, 370192, 371716, 378376, 413056, 413248, 416776, 420868,
429376, 433156, 443656, 443908, 444676, 462016, 471088, 476296, 477256, 477316, 492616, 492676, 493636, 507916, 175872,
186912, 190992, 192672, 232992, 239112, 354824, 367112, 96640, 110208, 125200, 189728, 208256, 208448, 211472, 247044,
290560, 302848, 320320, 323680, 323728, 341536, 347008, 347456, 347680, 353056, 353552, 385048, 404352, 405024, 408072,
420354, 428832, 432642, 434952, 435714, 436482, 461472, 467592, 469032, 469122, 476228, 492072, 492162, 493602, 499722,
121860, 125954, 155392, 192592, 214912, 215872, 216352, 219400, 221968, 225412, 231232, 253972, 322052, 335232, 363792,
365832, 56960, 60800, 60992, 81664, 117952, 120928, 124420, 125064, 127024, 185616, 185864, 188864, 203648, 204096,
232208, 234248, 247874, 293056, 316096, 318466, 337792, 338752, 339232, 342280, 350468, 356482, 362272, 399232, 400192,
400672, 400912, 403264, 404752, 406792, 407044, 407812, 411424, 412432, 414472, 415234, 416002, 418564, 419074, 427792,
430852, 431362, 459712, 460192, 460384, 460432, 461152, 461200, 461392, 461872, 463240, 463432, 463492, 463912, 464002,
464920, 464932, 464962, 465025, 467272, 467332, 467524, 467992, 468004, 468034, 468097, 469012, 469057, 471052, 475432,
475522, 475672, 475684, 475714, 475777, 476194, 477202, 477217, 479242, 483334, 483337, 491800, 491812, 491842, 491905,
492052, 492097, 492562, 492577, 493585, 495622, 495625, 499717, 507907, 119952, 123296, 123984, 160960, 175296, 178440,
178692, 182976, 183456, 186504, 190536, 207648, 320648, 369410, 377160, 377412, 56640, 93712, 95392, 106000, 112900,
113665, 117408, 120322, 121090, 189064, 189953, 190721, 210180, 211202, 216144, 216578, 217488, 217860, 232706, 246849,
305856, 306336, 307984, 309384, 313792, 314464, 316676, 317488, 317953, 321576, 333632, 334624, 341072, 352648, 363268,
377876, 428256, 434376, 442536, 460112, 463172, 463892, 463937, 467490, 475412, 475457, 476177, 479237, 32288, 48000,
60192, 109664, 109832, 111248, 112712, 119376, 174752, 177288, 202528, 204292, 217648, 238626, 239628, 246306, 291488,
293889, 312992, 320769, 339458, 340368, 341250, 342084, 348930, 363650, 366657, 369188, 369729, 370721, 378890, 30592,
31552, 48416, 63620, 81088, 88768, 89608, 91840, 94816, 104128, 104896, 107968, 115648, 116320, 117808, 119176,
123432, 123649, 126988, 158816, 158864, 160400, 162305, 171104, 171152, 182624, 185412, 189464, 215264, 219170, 221384,
230624, 231632, 232577, 233672, 235553, 237764, 245864, 245912, 245924, 246794, 247817, 285856, 301728, 314500, 316456,
339016, 345680, 345860, 349220, 410848, 411856, 413896, 417988, 427216, 430276, 442456, 442468, 442516, 459568, 461836,
463138, 463378, 463393, 467218, 467233, 467473, 471043, 32016, 46912, 48648, 55072, 56080, 59152, 62536, 93444,
95745, 109136, 111656, 117320, 146112, 146592, 154048, 157344, 158112, 158472, 161832, 161922, 169632, 170400, 173664,
173712, 174480, 176712, 177409, 178242, 181152, 183042, 184708, 192522, 215592, 221569, 222241, 223249, 230704, 231720,
238098, 246034, 246049, 286212, 288192, 289376, 289424, 291972, 300480, 305504, 306689, 319874, 335361, 338144, 338480,
352833, 361184, 361296, 362672, 363137, 364712, 368744, 368834, 376994, 378885, 62594, 63512, 87488, 87824, 91912,
92512, 95300, 112674, 123202, 168896, 171521, 177188, 184898, 206384, 206594, 209474, 218124, 234513, 276160, 277648,
300640, 300688, 301456, 301648, 304032, 304528, 304900, 307588, 307752, 312208, 312656, 313090, 313986, 317452, 323590,
323593, 334160, 348356, 348545, 352420, 361796, 368788, 402144, 403632, 405672, 417954, 426672, 430242, 434232, 434274,
434322, 461322, 463121, 31880, 47712, 59616, 62020, 63553, 92930, 96276, 101792, 103776, 104196, 105092, 107360,
110978, 111636, 117060, 119076, 145168, 146512, 146946, 152992, 154192, 154672, 157264, 158032, 160072, 169552, 170320,
173392, 181072, 181890, 183316, 190470, 192517, 200328, 203056, 206032, 213712, 214192, 217282, 239619, 276896, 285024,
285320, 291120, 291585, 291906, 304712, 308258, 315714, 320018, 333024, 344544, 346224, 346280, 346392, 348264, 348312,
348684, 350217, 361584, 364632, 365073, 368913, 369161, 376914, 376929, 376977, 60452, 61736, 62081, 78784, 80944,
86944, 89345, 92368, 92740, 105217, 105538, 115504, 116008, 116353, 122980, 123913, 126979, 145760, 154244, 156512,
157064, 160292, 173444, 178193, 181448, 182552, 201440, 201552, 205280, 206120, 206960, 207128, 209000, 209048, 209217,
209930, 210953, 214296, 217188, 221240, 230000, 233528, 233556, 237665, 237713, 237830, 284048, 288304, 289096, 289156,
301384, 302104, 302209, 305688, 308236, 315556, 316433, 319681, 331396, 333104, 334024, 336592, 337072, 344752, 344776,
346306, 352340, 353285, 354307, 361624, 363529, 365573, 397792, 398032, 398512, 399472, 401872, 402544, 405592, 405604,
405652, 410032, 410224, 413752, 413794, 413842, 417844, 417874, 426352, 430132, 430162, 458992, 459532, 460042, 460294,
460297, 461062, 461065, 461317, 461827, 30048, 54496, 55504, 61986, 62497, 73376, 85392, 93249, 94488, 96266,
108098, 110756, 123026, 144288, 154690, 167304, 167496, 169288, 170520, 172848, 174372, 174657, 175137, 176418, 181360,
181570, 181633, 184513, 202264, 202920, 202948, 203794, 207012, 209058, 221701, 233989, 277250, 277544, 284232, 288324,
288386, 289826, 298544, 304424, 305698, 312104, 312897, 313496, 315914, 334356, 337092, 338072, 338961, 340490, 347142,
368690, 46640, 46724, 58064, 58244, 58544, 58692, 61716, 79240, 79432, 80424, 88240, 88612, 91522, 100240,
101256, 104728, 107208, 107812, 116884, 118945, 119305, 172868, 186371, 201840, 207878, 213360, 213444, 216069, 218115,
231476, 274000, 274056, 275296, 283472, 285250, 301348, 304450, 305428, 309253, 312514, 313618, 315476, 319564, 331288,
334858, 336664, 337420, 340066, 340129, 345281, 345377, 352330, 395976, 396456, 403596, 411786, 426636, 427146, 428076,
428106, 428166, 460037, 28292, 46288, 47280, 52776, 54088, 55681, 76688, 77136, 91682, 92705, 108065, 108737,
109066, 111110, 144936, 150856, 153217, 153892, 156546, 160018, 160774, 160777, 166788, 182418, 184585, 188460, 188466,
201892, 217225, 229780, 229794, 230049, 230665, 231685, 233731, 275336, 277028, 281904, 283524, 288036, 291084, 296840,
298184, 301250, 303728, 307466, 313505, 332400, 333986, 336240, 340485, 345144, 345234, 348210, 352305, 12224, 30232,
30832, 45536, 48146, 53936, 54146, 54804, 59660, 78640, 80449, 80908, 88594, 91240, 91288, 92322, 102832,
103234, 107284, 107794, 108600, 110674, 110857, 114928, 115468, 115852, 115985, 118858, 118918, 123013, 144770, 145793,
151064, 151076, 156948, 157196, 157889, 169154, 172456, 172738, 174164, 199112, 199784, 199956, 205400, 205460, 205505,
205964, 206097, 213644, 213770, 229964, 230444, 230915, 275780, 282241, 283760, 289057, 304228, 313609, 332482, 338002,
344849, 360844, 360850, 362524, 362755, 364675, 394696, 394948, 395716, 396376, 396388, 396436, 398476, 399436, 402508,
409996, 410188, 410668, 410698, 410758, 411676, 411718, 426316, 427036, 427078, 458956, 459523, 15688, 15748, 15940,
20384, 27944, 29136, 29384, 31010, 31244, 43848, 46360, 55458, 56329, 57712, 58890, 59913, 77348, 87330,
88152, 94353, 94470, 101697, 101900, 103016, 103586, 104594, 106904, 107170, 109061, 116780, 116817, 117251, 142660,
142977, 144961, 150192, 150724, 152232, 152680, 152728, 152897, 154122, 159841, 165672, 167442, 172272, 173322, 174218,
176262, 176643, 180792, 184348, 184369, 188547, 199362, 199842, 200209, 206922, 213601, 215299, 221251, 231450, 269954,
273794, 275841, 276801, 281412, 284193, 284820, 287832, 297794, 298274, 299810, 299841, 300321, 305297, 307281, 311906,
329384, 330344, 332196, 336290, 336321, 337034, 337161, 338181, 340041, 345603, 361034, 361514, 23938, 27970, 36704,
40484, 44564, 45860, 51632, 51824, 55814, 58466, 59474, 72520, 73000, 76968, 79378, 80161, 85156, 86849,
88204, 92425, 100720, 101154, 101464, 107596, 108678, 115338, 118853, 135056, 138440, 143984, 148944, 149316, 150401,
152354, 153761, 153873, 158723, 168344, 168548, 168596, 168722, 168737, 169478, 170246, 180881, 182314, 182405, 184358,
199268, 202833, 203267, 205140, 205577, 217110, 229713, 273192, 277009, 277510, 282130, 283048, 286960, 287905, 288265,
297384, 298252, 298506, 299248, 300553, 303448, 311860, 312370, 319525, 330514, 330836, 330850, 330913, 332876, 332945,
344710, 348227, 394152, 394914, 395682, 395832, 395874, 395922, 397962, 399402, 401802, 401964, 401994, 402054, 402474,
403482, 403494, 410154, 426282, 426522, 426534, 458922, 20304, 22224, 23336, 24129, 27304, 39624, 39748, 40104,
40132, 42882, 44184, 50884, 52033, 52513, 52753, 54538, 75588, 83778, 84632, 84756, 86676, 90516, 91188,
100004, 104585, 115782, 134880, 138520, 141000, 142146, 142626, 143812, 146437, 150802, 151792, 166338, 166968, 167185,
170034, 176169, 180450, 180532, 180618, 180997, 198082, 201121, 201361, 205897, 208933, 269634, 280344, 281252, 281378,
283236, 287372, 287498, 288849, 288905, 290858, 297572, 297620, 303754, 304262, 305222, 305225, 311945, 313382, 329064,
332553, 333873, 339990, 8072, 12080, 15128, 27330, 29092, 36560, 40212, 42664, 43460, 50840, 53900, 54354,
57953, 59441, 71554, 72856, 76948, 78064, 78604, 85258, 87137, 87178, 90706, 94249, 100584, 102964, 140720,
141092, 141928, 141976, 142100, 144482, 145460, 148424, 165569, 170057, 180809, 181289, 197480, 197528, 197540, 197988,
198228, 201260, 201798, 215054, 229913, 265696, 266056, 268072, 268996, 269185, 273601, 274882, 281176, 283913, 287074,
290885, 295876, 296356, 297745, 298066, 299320, 299596, 300106, 304177, 313411, 315427, 329108, 332129, 336172, 344390,
344617, 360649, 360745, 394072, 394084, 394132, 394552, 394594, 394642, 394804, 394834, 395572, 395602, 397612, 397642,
397702, 397852, 397894, 398362, 398374, 399382, 401692, 401734, 402454, 409882, 409894, 410134, 426262, 458812, 458842,
458854, 458857, 458902, 458905, 458917, 458947, 20168, 23652, 23820, 25520, 26433, 27032, 27732, 30981, 38785,
43794, 46342, 47148, 52300, 52362, 52486, 57745, 68576, 76226, 77321, 78180, 78420, 84417, 84578, 86356,
88134, 90444, 90761, 99896, 100106, 100961, 102732, 103497, 107651, 115749, 137026, 144012, 144138, 144529, 149140,
149908, 156805, 166665, 168585, 169002, 169093, 172572, 172586, 172677, 173125, 199218, 272036, 272162, 275090, 275508,
275553, 276530, 276554, 283281, 296532, 299654, 301097, 303497, 304154, 311580, 311601, 311685, 330060, 330122, 336515,
337934, 360995, 7632, 14996, 15010, 22180, 27409, 29740, 29746, 29777, 36632, 43426, 43660, 47363, 50580,
50594, 53601, 69232, 69505, 74664, 78982, 80005, 85129, 87173, 88195, 90337, 90418, 90449, 90652, 100753,
103450, 106794, 115075, 115222, 137996, 142601, 143672, 149090, 149265, 153669, 156209, 156230, 157733, 164584, 172361,
198218, 199196, 199241, 200988, 213084, 265168, 266020, 267184, 268052, 269585, 273057, 275971, 276508, 279490, 281861,
283676, 284698, 288790, 303276, 303685, 312334, 328546, 328609, 329266, 330310, 332326, 336931, 458837, 6112, 14113,
14568, 23697, 27188, 28856, 28898, 29066, 29318, 29955, 37744, 37800, 39690, 44106, 44113, 44293, 50020,
50953, 53368, 53554, 55333, 57546, 58390, 69060, 71268, 72084, 72966, 76172, 82776, 83532, 84306, 85020,
86570, 99540, 99921, 100682, 101414, 102609, 102789, 104462, 104469, 106821, 134512, 134977, 140737, 142380, 144133,
144412, 155946, 156186, 164706, 165217, 166092, 166449, 166470, 168108, 180750, 181267, 198277, 204969, 205326, 213523,
229518, 271592, 272145, 272972, 276611, 280977, 281673, 282849, 283779, 287270, 290835, 295640, 296492, 297080, 298053,
300054, 303398, 328152, 329937, 332058, 7784, 7970, 11504, 12044, 14676, 14738, 14930, 21400, 21848, 25828,
26257, 26374, 36514, 39266, 39308, 39569, 41921, 42401, 45260, 49640, 51740, 57500, 57626, 67528, 71308,
72330, 75106, 76084, 77059, 78154, 78595, 82660, 82834, 83170, 84358, 84613, 86321, 87065, 90310, 91150,
98788, 99718, 102556, 134564, 136596, 136850, 137825, 138374, 141538, 143570, 149624, 151878, 152613, 155973, 165044,
166531, 166979, 168261, 198833, 213141, 213155, 213261, 265624, 267682, 271756, 272596, 279764, 279884, 279946, 280780,
280876, 281137, 296161, 296581, 299557, 301075, 329285, 344153, 458803, 11172, 11874, 11922, 13962, 14985, 23594,
25428, 27674, 38290, 39985, 41700, 45172, 45361, 45594, 46147, 57510, 57513, 57877, 83505, 90405, 91155,
99026, 99530, 99612, 102678, 114746, 134922, 136780, 136969, 137610, 138499, 141702, 145422, 145443, 149676, 151939,
152089, 155929, 166926, 172122, 172309, 180310, 197446, 264048, 266984, 267480, 268472, 275481, 279393, 280754, 286874,
295394, 296265, 299162, 328998, 329091, 333831, 335973, 344205, 344331, 360501, 6896, 7106, 7585, 10088, 10136,
13650, 21812, 22277, 22858, 23299, 25810, 29251, 39137, 42104, 45226, 50385, 50726, 51398, 51401, 57635,
70120, 70360, 72773, 78361, 79891, 82360, 83273, 83482, 85027, 99145, 99505, 101395, 102502, 106585, 138281,
140177, 141617, 143561, 143894, 148268, 149801, 150083, 155798, 165158, 165413, 166057, 197786, 204854, 229451, 265825,
267857, 268873, 270820, 271032, 271156, 271186, 273445, 274729, 282902, 282917, 286825, 286997, 299285, 303267, 311815,
328348, 331918, 344110, 11617, 11665, 14545, 19340, 19640, 21708, 22086, 23574, 25208, 25905, 26732, 27011,
28764, 34288, 35745, 36489, 37602, 38188, 38225, 43305, 49618, 53401, 53779, 68408, 70866, 71796, 74630,
75141, 76310, 86076, 100517, 135988, 136417, 136498, 136753, 137766, 147937, 148650, 150037, 151642, 166179, 197923,
198742, 200803, 264129, 272554, 272963, 274598, 278904, 280682, 281102, 295377, 295811, 297485, 329786, 360519, 7394,
7813, 11146, 11468, 13025, 13137, 13610, 13862, 14534, 18916, 19156, 19681, 25801, 26149, 27661, 35288,
35512, 35636, 36611, 37332, 38469, 39491, 41396, 42521, 45581, 51477, 54279, 68914, 69162, 71209, 75868,
78019, 78101, 82378, 82726, 83222, 84073, 99013, 99109, 132080, 133004, 133025, 134354, 137372, 139954, 140486,
149657, 151694, 164273, 165017, 165989, 168014, 168083, 197397, 198797, 200781, 204935, 229415, 269347, 271177, 271237,
272665, 279875, 280085, 296467, 299149, 360475, 458767, 7348, 7708, 9940, 13126, 25385, 34636, 34642, 35956,
37660, 39004, 41649, 45205, 50236, 67380, 67425, 68276, 68300, 68433, 70322, 71235, 71786, 72003, 74348,
74865, 75285, 75878, 77882, 86535, 106526, 106541, 114717, 132792, 133560, 140835, 141837, 141843, 148134, 148165,
148675, 148750, 165027, 197689, 204829, 265094, 266700, 267462, 269325, 273419, 278985, 279654, 279705, 280677, 282702,
295529, 295694, 299054, 331821, 6961, 10616, 13481, 13509, 18226, 19570, 21106, 26766, 26773, 28813, 37322,
38435, 41324, 49754, 49829, 49859, 70949, 72205, 74161, 74394, 74842, 84022, 99091, 99427, 100429, 102451,
133841, 135633, 137411, 139634, 141401, 143445, 148569, 149003, 196922, 197715, 264988, 265513, 266612, 267790, 268390,
268453, 268558, 286855, 297099, 303133, 329803, 344087, 5560, 6828, 9186, 11077, 11430, 11542, 12700, 12906,
14437, 18036, 19226, 19737, 20844, 26723, 27143, 28725, 34609, 41622, 43062, 57415, 66516, 70002, 70086,
74405, 83029, 133802, 139500, 148534, 164437, 197011, 197259, 198699, 200731, 263882, 264006, 271509, 274695, 279203,
279309, 295269, 296199, 328275, 328734, 5833, 6601, 7699, 9164, 9588, 9674, 10865, 12860, 13581, 14393,
17836, 18307, 18858, 19237, 20933, 21097, 21262, 24892, 35434, 35945, 35990, 37286, 38030, 42125, 42503,
49497, 66808, 67025, 67270, 68265, 68803, 70931, 73962, 74089, 75915, 98428, 98862, 133364, 133484, 135845,
139685, 139877, 149555, 151595, 164026, 164237, 165917, 263857, 265276, 266842, 268373, 270742, 278716, 295126, 5914,
6934, 7257, 11043, 13390, 17361, 17862, 18089, 18780, 18801, 22611, 33714, 35058, 37462, 37465, 41302,
41358, 41550, 41739, 42030, 49486, 50251, 66986, 67353, 68186, 68366, 68750, 68877, 70041, 70710, 73916,
74339, 82275, 98486, 132886, 135402, 135529, 135738, 140341, 140427, 140551, 172047, 264857, 267027, 270650, 271443,
295097, 9923, 10810, 12707, 17770, 21077, 25165, 34028, 34245, 39047, 41555, 45095, 51227, 69127, 70222,
82507, 132444, 132713, 134229, 141383, 147795, 163961, 262900, 263410, 263830, 264643, 265294, 265363, 266585, 266809,
266851, 270460, 270553, 279595, 327965, 3861, 5571, 6460, 7307, 10458, 12598, 18062, 19091, 19509, 20694,
21131, 22557, 35130, 37639, 49453, 49707, 67235, 67941, 68003, 69861, 73941, 74823, 82131, 135292, 135763,
136267, 263066, 263514, 264425, 264758, 264971, 270542, 270878, 270899, 2930, 5478, 9561, 10574, 11315, 17130,
17765, 19591, 24694, 24907, 33689, 37917, 49269, 49331, 66289, 66499, 69934, 77839, 82078, 98475, 99343,
131804, 133895, 135350, 135353, 147566, 164039, 164123, 264589, 266635, 279067, 295027, 2921, 3641, 3670, 4965,
5013, 9402, 10579, 18003, 35029, 66894, 69753, 132654, 263509, 266451, 1958, 3659, 6699, 13335, 34187,
34334, 34990, 37163, 66445, 74011, 131546, 131558, 131897, 132218, 133406, 133451, 134171, 134183, 264499, 2787,
6323, 9006, 10413, 12403, 18759, 20763, 35367, 36973, 41067, 41115, 66765, 131917, 132397, 132637, 132679,
135463, 147517, 263289, 263342, 264350, 1513, 1765, 2684, 3374, 4939, 5237, 16822, 20647, 49295, 66166,
67699, 135695, 262861, 263719, 2518, 8989, 9325, 17502, 24635, 66347, 67863, 69783, 262581, 262766, 956,
3229, 4525, 4766, 4807, 8875, 16633, 17053, 35087, 67645, 131445, 133271, 1459, 8647, 16958, 17175,
66141, 131499, 266295, 270383, 1691, 1807, 2459, 4446, 5179, 6223, 16999, 33341, 33879, 131315, 262507,
1438, 2703, 4317, 18479, 33175, 10271, 131639, 262375, 8375, 635, 65647, 131167, 254, 319$\}$

\medskip

$\cS_{18,4,8}=\{$261120, 258816, 256640, 249600, 251520, 252288, 252480, 255360, 255552, 256320, 258240, 244256, 251200, 212736, 226848,
227616, 227856, 242976, 243216, 243984, 249024, 258096, 65280, 129184, 129544, 179840, 192032, 226576, 240800, 241160,
243848, 256040, 97920, 114048, 114240, 122400, 126240, 126480, 193696, 194056, 194656, 194704, 194824, 195076, 128096,
128144, 128264, 128516, 129104, 129284, 163200, 163392, 179520, 187680, 187920, 191760, 212160, 220320, 220680, 224352,
224400, 224520, 224772, 226440, 227400, 227460, 236640, 236688, 236808, 237060, 240720, 240900, 242760, 242820, 243780,
246720, 248880, 250920, 251928, 251940, 255000, 255012, 256020, 258060, 97600, 122128, 193616, 193796, 64704, 124576,
125536, 126088, 187040, 190880, 220240, 220420, 226372, 238496, 240386, 242306, 250900, 254498, 121488, 125328, 187528,
190096, 191056, 191560, 191620, 120224, 121184, 121928, 121988, 126020, 185952, 189792, 209856, 212016, 218016, 219906,
222048, 222096, 222978, 224001, 225666, 225858, 225921, 234336, 234384, 235266, 236289, 238416, 239361, 241986, 242049,
242241, 246576, 248844, 250146, 250386, 250401, 254226, 254241, 254481, 258051, 120400, 124240, 185744, 186704, 187460,
62400, 64560, 109504, 112264, 113284, 113704, 117664, 123784, 162440, 178760, 189152, 217936, 218881, 225601, 232160,
232328, 233090, 238280, 239810, 246440, 247970, 248330, 250129, 254090, 96648, 112968, 119520, 123360, 123600, 158656,
162856, 170944, 174016, 177544, 178564, 179224, 179236, 182176, 183136, 183184, 185224, 189256, 189316, 64130, 89024,
92096, 95816, 96836, 97304, 97316, 104384, 113684, 116576, 116624, 117584, 119624, 119684, 123716, 161412, 177732,
184800, 185040, 188880, 200640, 209712, 211980, 214752, 214920, 215520, 215760, 215880, 215940, 216450, 216642, 216705,
217800, 219330, 221640, 221892, 222402, 223425, 230880, 231120, 231240, 231300, 231888, 232260, 232770, 232833, 233025,
233928, 234180, 234690, 235713, 238020, 238785, 246000, 246120, 246168, 246180, 246360, 246372, 246420, 246540, 246882,
246930, 246945, 247050, 247302, 247305, 247890, 247905, 247953, 248070, 248073, 248325, 248835, 249930, 249990, 249993,
254022, 254025, 254085, 95620, 111940, 119248, 153536, 161096, 162116, 162836, 182096, 185156, 48032, 53184, 62256,
62850, 63042, 63105, 63810, 63873, 64065, 64524, 93736, 96802, 109360, 110116, 110209, 111280, 111400, 112162,
112930, 113170, 117448, 124426, 125194, 125446, 159362, 171560, 175400, 175490, 175682, 190985, 207536, 207656, 208418,
209576, 211106, 211466, 214480, 214852, 216385, 217540, 218305, 238136, 239666, 241706, 246100, 246865, 247045, 249925,
79776, 80736, 80784, 106008, 109848, 158512, 159256, 160432, 160552, 162337, 170800, 173872, 175636, 176560, 176752,
176920, 176932, 177697, 178465, 178705, 181960, 182728, 182980, 46944, 46992, 47952, 60080, 60200, 62120, 62785,
88880, 89384, 91952, 93058, 93569, 93761, 94640, 94832, 95000, 95012, 95522, 95762, 96530, 104240, 105346,
105764, 108418, 109224, 109889, 110960, 111380, 111890, 116168, 116420, 117188, 120970, 124166, 143264, 146144, 155172,
159012, 159042, 174977, 189705, 189957, 190725, 200496, 202416, 202536, 203184, 203376, 203544, 203556, 204066, 204306,
204321, 206256, 206448, 206616, 206628, 207216, 207636, 208146, 208161, 208401, 209136, 209256, 209304, 209316, 209496,
209508, 209556, 209676, 210018, 210066, 210081, 210186, 210438, 210441, 211026, 211041, 211089, 211206, 211209, 211461,
211971, 217656, 219186, 221496, 221748, 222258, 223281, 225306, 225318, 225321, 233784, 234036, 234546, 235569, 237876,
238641, 241686, 241689, 241701, 245964, 246531, 79696, 89620, 93460, 121347, 153392, 154904, 160112, 160532, 161057,
161297, 162065, 171284, 173736, 174744, 174756, 177425, 181700, 16320, 24480, 28512, 28560, 30432, 30600, 31200,
31440, 31560, 31620, 44768, 44936, 53040, 54960, 55080, 55728, 55920, 56088, 56100, 58800, 58992, 59160,
59172, 59760, 60180, 60578, 61680, 61800, 61848, 61860, 62040, 62052, 62100, 62220, 64515, 80584, 87938,
88744, 92824, 92836, 102064, 102184, 104872, 105064, 105634, 107944, 108136, 108784, 108964, 109144, 109324, 109666,
109714, 110824, 117304, 120902, 123064, 123436, 125059, 139104, 139152, 143184, 144864, 145104, 145872, 146306, 154497,
157569, 158312, 169857, 176780, 185859, 186410, 186441, 186501, 186627, 189482, 190490, 190502, 200360, 202096, 202516,
204049, 207080, 207500, 208010, 209236, 210001, 210181, 217396, 218161, 225301, 231608, 231980, 232490, 237740, 239630,
56842, 73440, 73608, 77280, 77520, 77640, 77700, 92520, 94860, 104100, 108180, 110988, 111180, 120067, 120873,
123945, 124953, 124965, 151216, 151336, 153256, 154024, 157096, 157936, 158104, 158356, 158476, 158881, 159976, 167344,
167536, 167704, 167716, 169624, 170224, 170344, 170392, 170596, 170644, 170764, 171169, 173296, 173476, 173656, 173836,
174424, 174436, 175201, 175249, 176344, 176356, 181816, 182584, 182836, 184504, 184876, 188536, 188596, 188716, 188956,
189571, 190531, 24400, 30160, 30532, 40416, 40656, 40776, 40836, 44496, 44868, 52904, 54640, 55060, 56418,
56466, 59624, 60044, 60498, 61780, 79304, 79556, 80324, 85424, 85616, 85784, 85796, 87704, 88304, 88484,
88664, 88844, 91376, 91544, 91748, 91916, 94424, 94436, 101744, 102164, 103664, 103784, 103832, 104204, 104852,
105044, 105121, 105554, 108884, 110804, 116024, 116276, 117044, 118904, 118964, 119084, 119324, 119338, 119834, 119846,
122996, 123164, 123971, 143048, 145218, 146241, 160140, 160332, 169634, 170402, 176460, 178307, 186390, 189462, 199920,
200040, 200088, 200100, 200280, 200292, 200340, 200460, 201960, 202380, 202968, 202980, 203148, 203340, 203850, 203910,
203913, 206040, 206052, 206220, 206412, 207060, 207180, 207942, 207945, 208005, 209100, 209667, 214200, 214572, 215160,
215220, 215340, 215580, 216090, 216102, 216105, 217260, 219150, 221292, 221340, 222222, 223245, 230520, 230580, 230700,
230940, 231540, 231708, 232470, 232473, 232485, 233580, 233628, 234510, 235533, 237660, 238605, 245820, 245955, 32265,
48394, 48646, 56582, 56585, 56837, 73168, 73540, 89185, 89233, 92754, 93258, 93318, 94540, 104082, 108114,
111747, 112707, 120853, 123925, 145000, 150896, 151316, 152816, 152996, 153176, 153356, 154004, 154196, 154210, 154258,
157076, 157268, 157282, 158801, 159956, 169304, 169316, 170578, 170634, 171089, 173396, 181556, 184436, 184604, 185369,
185381, 16176, 24264, 28104, 28356, 44834, 47672, 47714, 52464, 52584, 52632, 52644, 52824, 52836, 52884,
53004, 54504, 54924, 55512, 55524, 55692, 55884, 58274, 58584, 58596, 58764, 58956, 59604, 59724, 59978,
61644, 62211, 85666, 87384, 87396, 87650, 87713, 88404, 90856, 91793, 92386, 101608, 102028, 104033, 107236,
107362, 107921, 108748, 108937, 116908, 122986, 123046, 138696, 138948, 142788, 158049, 168680, 172520, 172946, 173702,
174634, 177219, 182915, 184617, 188585, 199586, 200020, 201940, 202060, 203845, 205538, 205706, 207491, 214132, 214300,
216085, 217180, 218125, 230066, 230186, 231971, 237731, 238091, 31882, 32005, 47506, 47881, 48265, 55945, 59782,
79416, 80184, 80289, 80436, 88710, 91041, 91489, 92497, 103128, 104010, 104785, 106968, 107850, 109125, 116355,
119062, 145762, 150760, 151180, 152916, 152930, 156376, 157900, 158278, 167128, 167140, 167308, 167500, 170188, 172756,
172897, 173260, 181420, 182380, 182428, 184853, 16040, 24004, 40722, 40737, 44817, 46392, 46497, 46644, 47412,
52564, 54114, 54162, 54177, 54484, 54604, 58194, 58257, 80465, 85208, 85220, 85346, 85388, 85580, 86504,
88268, 88394, 91340, 92486, 94378, 101588, 101708, 101770, 102884, 103250, 103628, 103814, 104998, 110746, 115820,
115868, 115882, 116828, 117059, 143114, 145034, 145592, 145964, 145970, 151137, 152292, 152970, 153825, 154755, 156132,
157257, 157994, 167561, 169545, 170310, 174281, 188517, 197616, 198498, 198546, 198561, 199506, 199521, 199569, 199884,
200451, 201186, 201426, 201441, 201546, 201606, 201609, 202371, 203139, 203331, 205266, 205281, 205521, 205638, 205641,
205701, 206211, 206403, 207171, 208956, 209091, 213426, 213618, 213681, 213786, 213798, 213801, 214563, 215331, 215571,
217251, 217611, 221283, 221331, 221451, 221703, 229746, 229809, 230001, 230166, 230169, 230181, 230691, 230931, 231699,
233571, 233619, 233739, 233991, 237651, 237831, 245811, 31814, 46306, 46673, 46853, 47329, 48197, 79156, 86740,
86881, 86929, 87250, 87433, 90580, 90954, 93219, 96267, 101958, 104730, 107718, 108057, 110694, 116886, 144289,
150740, 150860, 150929, 152024, 152780, 153158, 156497, 158083, 168404, 168785, 169258, 173338, 174358, 176233, 181340,
181571, 15600, 15720, 15768, 15780, 15960, 15972, 16020, 16140, 24120, 27960, 28212, 28422, 29880, 30252,
30840, 30900, 31020, 31260, 44216, 45962, 47276, 47306, 50160, 51938, 52106, 52428, 52995, 54661, 55621,
59939, 61500, 61635, 80425, 87621, 88613, 90802, 92329, 92693, 94358, 99304, 100066, 100834, 101074, 101194,
106930, 107290, 107813, 108604, 108649, 109579, 115354, 139014, 139017, 142562, 143109, 144564, 144684, 144690, 144774,
144924, 145524, 145692, 153894, 154137, 156332, 156997, 159897, 159909, 160779, 161799, 165770, 166790, 167241, 169494,
172652, 172721, 176277, 180650, 180842, 180902, 181353, 198481, 199370, 201169, 201541, 202051, 213361, 213781, 214291,
217171, 217351, 229610, 230030, 231563, 30921, 40337, 44617, 45778, 46618, 47653, 53962, 57802, 58054, 58057,
72888, 73260, 76920, 76980, 77100, 77340, 79020, 79980, 80028, 83794, 84705, 84873, 86854, 88345, 90540,
95239, 100233, 101253, 103593, 103957, 104601, 106860, 108693, 115113, 115305, 115365, 115798, 138578, 142729, 145606,
148456, 152393, 156826, 156838, 157756, 164824, 164836, 165601, 166369, 166609, 168364, 169157, 170044, 171019, 172444,
173116, 173158, 175111, 180889, 181397, 182357, 15700, 16131, 23860, 29812, 29980, 40056, 40116, 44148, 44421,
45894, 46188, 46236, 46374, 47196, 50658, 50898, 51018, 51078, 51666, 52038, 52041, 58659, 58899, 59667,
73289, 73349, 79145, 79397, 82904, 82916, 84778, 86636, 87566, 88124, 90482, 90716, 90902, 91196, 91301,
92218, 99284, 99794, 103484, 104549, 105479, 107203, 107578, 115034, 115046, 142034, 142508, 142889, 144585, 152220,
153690, 156081, 156453, 165702, 166697, 168540, 172821, 174137, 174347, 180630, 180822, 197580, 198090, 198342, 198345,
199110, 199113, 199365, 199740, 199875, 208947, 213210, 213222, 213225, 213390, 213582, 213645, 214155, 215115, 215175,
229590, 229593, 229605, 229710, 229773, 229965, 230475, 230535, 231495, 245775, 29893, 30243, 31011, 31251, 40579,
45521, 53702, 53705, 53957, 55459, 57797, 72820, 72988, 75489, 78940, 80149, 83658, 83785, 83845, 84805,
86428, 87317, 91225, 100134, 100165, 102748, 103203, 103513, 108039, 115093, 115285, 138418, 138629, 141748, 143818,
148436, 148961, 149201, 150186, 150217, 151916, 152177, 152636, 152726, 153869, 155996, 156355, 165329, 166723, 168613,
168643, 169229, 172483, 180569, 13296, 15564, 23724, 24131, 27756, 27804, 29801, 30821, 39602, 39722, 40134,
41912, 42674, 42698, 45734, 50124, 52284, 52419, 57914, 61491, 75562, 76234, 76486, 76582, 78442, 79001,
86723, 90563, 102851, 107062, 107149, 136632, 138348, 138396, 140664, 142428, 144069, 144474, 145497, 152355, 156729,
156813, 168506, 198085, 199226, 204986, 205358, 206891, 213205, 213325, 214087, 22993, 23217, 23337, 26409, 27058,
27249, 27417, 27459, 29097, 29338, 29971, 40259, 43434, 43441, 45721, 46166, 55379, 70584, 71137, 71377,
74616, 74676, 83398, 84378, 84582, 85515, 90681, 102969, 106809, 136820, 138773, 142486, 144485, 145549, 149273,
149957, 150181, 152003, 153653, 165289, 11240, 14153, 15555, 23644, 26282, 37752, 37812, 38322, 38346, 38514,
39282, 41844, 42354, 42441, 42693, 42774, 42819, 43670, 44174, 45402, 47367, 53562, 53814, 54411, 57654,
57909, 71338, 71450, 72346, 72470, 75418, 75430, 76901, 78422, 78499, 83606, 90763, 103051, 134072, 137894,
137897, 142115, 142394, 150795, 151047, 166293, 168246, 197436, 197571, 197946, 198198, 198201, 198966, 198969, 199221,
199731, 200826, 200886, 200889, 201006, 201246, 201261, 201771, 202779, 202791, 204918, 204921, 204981, 205086, 205101,
205341, 205851, 205863, 206871, 208911, 14001, 15130, 21937, 22309, 23317, 26389, 27045, 27795, 29017, 29030,
29077, 29269, 38681, 45413, 70516, 76182, 76563, 80007, 83514, 85075, 85255, 87115, 90421, 102709, 136114,
137585, 140146, 140953, 140965, 141925, 143929, 143971, 144019, 148826, 149846, 156211, 165461, 172339, 172615, 4080,
6120, 7128, 7140, 10200, 10212, 11220, 13260, 15420, 18360, 19320, 19380, 21240, 21420, 22182, 25080,
25332, 25452, 25500, 25962, 26202, 26970, 35576, 35756, 36458, 36506, 36518, 39333, 39513, 43238, 47147,
49980, 50115, 50533, 50787, 51386, 51758, 52275, 54343, 61455, 71914, 72038, 72041, 72971, 74986, 75110,
76387, 78179, 78227, 83363, 86579, 99004, 100474, 100654, 100894, 102574, 133044, 134004, 134570, 138003, 138297,
138339, 140198, 141075, 141710, 142389, 152199, 164771, 165050, 165422, 166070, 168327, 180519, 197941, 198830, 200821,
200989, 201751, 11890, 13681, 13738, 19881, 20121, 22165, 27213, 29259, 39982, 41893, 43363, 43603, 43659,
44077, 50579, 57502, 67320, 67500, 68088, 68340, 68460, 68508, 74649, 75349, 75993, 83254, 86134, 99219,
99885, 136553, 137530, 141653, 143669, 143699, 149045, 149689, 166429, 6100, 7793, 11939, 13251, 14990, 15411,
18292, 20980, 21340, 22283, 25942, 27271, 34296, 34548, 34668, 34716, 34737, 35316, 35676, 35697, 36186,
36198, 36246, 37617, 39225, 41817, 50298, 50358, 50361, 50478, 50718, 51318, 51321, 51381, 51486, 51501,
53358, 53421, 53787, 70565, 70886, 70998, 71894, 74597, 86315, 99614, 100531, 133868, 133994, 134742, 136038,
136611, 140009, 140622, 140643, 141854, 141899, 143534, 147900, 148092, 148371, 149715, 164220, 166061, 168107, 172187,
196860, 197427, 197742, 197790, 197805, 198750, 198765, 198813, 199695, 13965, 14731, 19801, 20053, 21721, 23595,
25827, 25995, 35737, 36259, 36409, 38229, 39125, 39965, 45341, 57437, 67060, 67420, 69013, 69267, 70489,
71431, 75085, 75091, 75315, 78363, 82300, 83245, 83485, 84085, 99445, 132850, 133618, 134501, 136085, 136787,
140501, 143659, 148871, 151835, 155927, 166003, 4044, 11497, 13116, 13206, 13902, 14669, 14919, 18156, 18924,
19164, 19363, 22919, 26183, 26798, 37773, 41430, 41774, 42105, 42165, 42195, 42539, 45591, 49404, 49971,
50461, 51735, 52239, 70106, 71861, 71987, 72221, 82606, 86171, 98926, 132588, 132828, 133596, 133977, 135897,
137547, 139737, 139982, 141422, 141491, 148659, 151659, 155739, 168027, 197291, 197725, 11482, 20013, 21331, 23575,
28766, 34618, 35557, 35638, 36149, 39070, 53399, 70350, 90199, 99483, 102487, 106575, 134029, 135653, 136478,
140573, 148295, 165143, 4035, 7011, 7397, 7566, 7734, 10940, 11574, 11595, 11822, 13107, 13639, 15375,
17884, 22685, 26731, 34521, 35534, 37662, 37707, 38439, 41210, 49835, 68069, 68281, 70835, 76303, 133859,
140459, 149607, 172079, 196851, 196971, 197019, 197031, 197211, 197223, 197271, 197391, 3945, 4005, 6035, 10126,
11061, 11086, 11803, 21619, 23055, 26895, 41566, 67033, 69159, 70093, 70765, 90159, 132566, 132941, 140815,
3900, 5820, 5850, 5859, 6588, 6780, 6867, 7463, 9660, 9852, 9942, 10620, 10707, 12540, 13479,
18041, 35277, 37110, 37485, 37995, 42255, 49395, 49515, 49935, 67253, 68779, 74155, 74903, 136347, 137367,
196951, 6861, 18891, 19623, 24813, 26679, 73950, 132907, 164407, 3891, 5500, 5934, 9651, 10023, 12531,
13071, 18803, 21053, 21095, 37275, 38967, 38991, 49495, 70799, 82109, 83999, 132345, 133365, 133787, 134235,
135390, 148527, 196815, 6393, 10653, 11351, 17653, 19035, 24891, 34030, 34247, 43039, 67179, 67933, 73959,
131998, 136271, 3324, 3855, 5549, 20670, 25631, 34455, 35051, 49359, 68183, 82043, 147999, 5035, 5725,
6487, 6715, 9067, 18063, 20791, 20815, 66030, 66285, 67983, 3315, 3483, 5467, 8693, 12495, 33639,
34383, 37167, 41079, 132327, 133679, 196671, 2967, 5335, 8923, 9047, 33214, 131773, 1020, 3279, 6319,
16871, 17111, 49215, 133407, 4791, 33403, 1011, 1515, 8879, 9327, 12351, 1723, 975, 1655, 2415,
3135, 16799, 831, 255$\}$

\medskip

$\cS_{19,4,8}=\{$522240, 517632, 513280, 499200, 503040, 504576, 504960, 510720, 511104, 512640, 516480, 488512, 502400, 425472, 453696,
455232, 455712, 485952, 486432, 487968, 498048, 516192, 130560, 228608, 258368, 384064, 390160, 453152, 481600, 482320,
487696, 512080, 244800, 252480, 252960, 257040, 258576, 259080, 326912, 359168, 359552, 387392, 389312, 389408, 195328,
195712, 227968, 256192, 256288, 258208, 375360, 375840, 383520, 387600, 388104, 389640, 424320, 440640, 441360, 448704,
448800, 449040, 449544, 452880, 454800, 454920, 473280, 473376, 473616, 474120, 481440, 481800, 485520, 485640, 487560,
493440, 497760, 501840, 503856, 503880, 510000, 510024, 512040, 516120, 244256, 256520, 326272, 387232, 129408, 249152,
251072, 252176, 374080, 381760, 440480, 440840, 452744, 476992, 480772, 484612, 501800, 508996, 242976, 250656, 375056,
380192, 382112, 383120, 383240, 240448, 242368, 243856, 243976, 252040, 371904, 379584, 419712, 424032, 436032, 439812,
444096, 444192, 445956, 448002, 451332, 451716, 451842, 468672, 468768, 470532, 472578, 476832, 478722, 483972, 484098,
484482, 493152, 497688, 500292, 500772, 500802, 508452, 508482, 508962, 516102, 240800, 248480, 371488, 373408, 374920,
124800, 129120, 193808, 219008, 226448, 227408, 235328, 247568, 251393, 255233, 355600, 357640, 378304, 435872, 437762,
451202, 464320, 464656, 466180, 472321, 476560, 478465, 479620, 480001, 480385, 492880, 495940, 496660, 496705, 500258,
508180, 508225, 508945, 512005, 224016, 226056, 239040, 246720, 247200, 317312, 324368, 325712, 341888, 348032, 357008,
358448, 358472, 364352, 366272, 366368, 370448, 374273, 378512, 378632, 380417, 385793, 386177, 128260, 178048, 184192,
191752, 194608, 194632, 208768, 224392, 227368, 233152, 233248, 235168, 239248, 239368, 241153, 247432, 254593, 322704,
324744, 369600, 370080, 377760, 401280, 419424, 423960, 429504, 429840, 431040, 431520, 431760, 431880, 432900, 433284,
433410, 435600, 438660, 443280, 443784, 444804, 446850, 461760, 462240, 462480, 462600, 463776, 464520, 465540, 465666,
466050, 467856, 468360, 469380, 469761, 470145, 471426, 471681, 476040, 477570, 477825, 492000, 492240, 492336, 492360,
492720, 492744, 492840, 493080, 493764, 493860, 493890, 494100, 494145, 494604, 494610, 494625, 495780, 495810, 495906,
496140, 496146, 496161, 496650, 497670, 499860, 499905, 499980, 499986, 500001, 500241, 500745, 501765, 503811, 508044,
508050, 508065, 508170, 508425, 509955, 191120, 193160, 238496, 307072, 322312, 325672, 354952, 364192, 370312, 63296,
106368, 124512, 125700, 126084, 126210, 127620, 127746, 128130, 129048, 187472, 187652, 193604, 218720, 219908, 220232,
220292, 222560, 222800, 224324, 225860, 226340, 234896, 236801, 250241, 250898, 343120, 350800, 351490, 379924, 381460,
381964, 415072, 415312, 416836, 419152, 422212, 422932, 428960, 429704, 432770, 435080, 436610, 476272, 479332, 483412,
492200, 493730, 494090, 495889, 499850, 94016, 95936, 96032, 212016, 219696, 317024, 318512, 320864, 321104, 324674,
341600, 347744, 351272, 353120, 353504, 353840, 353864, 355394, 356930, 357410, 363920, 365456, 365960, 367361, 367745,
373121, 379265, 61120, 61216, 63136, 120160, 120400, 124240, 125570, 128065, 155456, 161216, 177760, 178768, 183904,
187012, 189280, 189664, 190000, 190024, 191044, 191524, 193060, 208480, 210692, 211528, 218448, 221920, 222760, 223780,
232336, 232840, 234376, 236161, 240001, 241940, 242705, 248338, 248842, 250004, 250124, 250378, 310344, 318024, 318210,
318594, 350850, 379404, 381202, 400992, 404832, 405072, 406368, 406752, 407088, 407112, 408132, 408612, 408642, 412512,
412896, 413232, 413256, 414432, 415272, 416292, 416322, 416802, 418272, 418512, 418608, 418632, 418992, 419016, 419112,
419352, 420036, 420132, 420162, 420372, 420876, 420882, 422052, 422082, 422178, 422412, 422418, 422922, 423942, 435312,
438372, 442992, 443496, 444516, 446562, 450612, 450636, 450642, 467568, 468072, 469092, 471138, 475752, 477282, 483372,
483378, 483402, 491928, 493062, 493713, 493833, 495753, 32640, 93856, 179240, 186920, 242694, 278336, 286400, 286496,
290240, 291776, 292256, 306784, 309808, 316752, 320224, 321064, 322114, 322594, 324130, 342568, 349392, 349488, 354850,
356673, 363400, 371729, 56768, 57104, 106080, 109920, 110160, 111456, 111840, 112176, 112200, 117600, 117984, 118320,
118344, 119520, 120360, 121156, 123360, 123600, 123696, 123720, 124080, 124104, 124200, 124440, 125505, 125985, 127521,
129030, 147136, 147232, 155296, 158656, 159136, 160672, 175876, 185168, 185672, 204128, 204368, 208208, 209744, 210248,
211268, 216272, 216368, 217568, 217808, 218408, 218648, 221648, 222488, 223553, 225473, 225569, 227333, 234608, 241804,
242185, 246128, 246872, 247892, 248837, 249953, 254033, 292612, 308994, 339714, 347464, 349508, 350404, 350500, 371718,
372820, 372882, 373002, 373254, 381066, 381955, 400720, 404192, 405032, 408098, 414160, 415000, 416020, 418472, 420002,
420362, 434792, 436322, 450602, 463216, 463960, 464980, 471121, 475480, 477265, 479260, 479281, 479305, 492805, 81344,
81680, 89024, 89504, 89744, 89864, 113684, 161412, 161538, 177360, 185520, 186690, 215880, 219330, 219426, 240134,
241746, 248070, 278176, 289696, 302432, 302672, 306512, 308528, 314192, 315872, 316112, 316712, 316952, 319952, 320792,
325637, 334688, 335072, 335408, 335432, 338768, 340448, 340784, 340808, 341168, 341192, 341528, 342338, 346592, 346928,
347312, 347672, 348872, 348968, 349761, 352688, 352712, 353048, 353432, 354497, 354593, 356513, 357381, 358403, 363632,
365168, 365672, 369008, 369752, 371209, 372833, 377072, 377192, 377432, 377912, 378962, 378977, 385073, 385097, 48064,
48544, 48784, 48904, 56224, 56968, 105808, 109280, 110120, 112836, 112932, 119248, 120088, 120996, 121873, 123560,
127249, 155024, 162881, 170848, 171232, 171568, 171592, 175408, 176608, 176944, 176968, 177688, 182752, 183088, 183496,
183832, 185876, 188848, 188872, 189208, 189592, 194563, 203488, 204328, 207328, 207568, 208168, 208408, 209584, 210088,
210116, 210497, 211108, 216596, 216641, 217768, 217922, 221608, 221848, 223393, 225542, 226307, 232048, 232552, 232724,
234088, 237808, 237928, 238168, 238648, 238676, 239668, 239692, 239713, 245992, 246328, 246836, 247852, 249926, 253993,
290436, 292482, 314692, 316612, 317761, 339268, 341284, 346948, 349714, 372780, 378118, 379014, 399840, 400080, 400176,
400200, 400560, 400584, 400680, 400920, 403920, 404760, 405936, 405960, 406296, 406680, 407700, 407820, 407826, 412080,
412104, 412440, 412824, 414120, 414360, 415884, 415890, 416010, 418200, 419334, 428400, 429144, 430320, 430440, 430680,
431160, 432180, 432204, 432210, 434520, 438300, 442584, 442680, 444444, 446490, 461040, 461160, 461400, 461880, 463080,
463416, 464940, 464946, 464970, 467160, 467256, 469020, 469041, 469065, 471066, 471081, 475320, 477210, 477225, 491640,
491910, 492165, 492291, 492675, 32352, 64020, 64524, 64530, 80800, 81544, 112961, 113164, 113170, 113674, 121025,
121121, 127028, 127052, 159362, 159489, 177474, 178370, 178466, 185634, 186530, 215748, 218274, 241706, 246348, 246918,
277904, 285584, 286088, 293441, 293921, 301792, 302632, 305632, 305872, 306472, 306712, 307912, 308008, 308036, 314032,
314536, 316072, 316196, 319912, 320152, 323846, 338608, 339112, 339521, 340676, 342178, 348836, 350481, 354566, 354821,
356486, 363112, 368872, 369208, 370738, 370762, 377930, 378922, 380966, 56900, 62576, 62849, 95554, 95764, 96532,
104928, 105168, 105264, 105288, 105648, 105672, 105768, 106008, 109008, 109848, 111024, 111048, 111384, 111768, 116548,
117168, 117192, 117528, 117912, 119208, 119448, 121353, 123288, 124422, 125073, 125193, 127113, 146320, 146824, 154504,
162337, 174792, 174888, 175300, 177316, 181712, 184772, 185026, 192646, 192773, 203216, 204056, 207528, 208066, 209698,
210194, 211217, 214472, 214820, 216332, 217496, 219281, 219401, 222725, 223366, 223747, 232204, 233816, 234068, 234572,
235601, 241733, 246533, 249881, 309441, 309537, 314898, 317601, 337360, 345040, 345796, 348610, 363282, 363786, 365650,
369234, 370947, 377394, 399172, 400040, 403880, 404120, 407690, 411076, 411412, 414982, 428264, 428600, 432170, 434360,
436250, 460132, 460372, 463942, 468037, 475462, 475717, 476182, 476197, 476227, 32080, 64010, 81476, 89636, 89666,
92968, 93296, 93476, 93569, 94832, 95105, 95336, 96785, 97289, 112801, 175425, 177420, 179205, 182082, 183105,
183444, 183570, 183585, 185484, 206256, 206658, 208020, 213936, 214722, 215460, 215841, 234246, 238124, 238725, 277384,
301520, 302360, 305832, 306340, 308372, 308386, 312752, 314561, 315800, 321029, 334256, 334280, 334616, 335000, 335169,
338626, 340376, 342161, 342281, 345512, 345890, 346520, 346772, 347276, 350345, 353361, 353795, 362840, 364760, 364856,
365702, 366641, 366665, 369706, 372761, 378905, 48452, 48676, 48706, 56866, 60016, 60289, 60520, 60692, 62056,
62993, 93708, 93714, 94660, 95754, 96396, 96402, 96522, 105128, 108228, 108324, 108968, 109208, 116388, 120882,
120906, 153392, 157508, 160112, 160532, 160577, 160856, 170416, 170440, 170776, 171160, 173008, 174756, 176536, 176788,
176802, 182680, 184993, 188756, 189955, 190597, 200516, 203176, 203416, 203585, 203969, 205768, 206500, 207256, 211081,
218373, 221492, 221516, 222257, 222281, 231640, 231736, 233656, 235561, 239641, 284488, 286228, 291940, 292033, 292114,
304584, 305940, 305985, 307650, 310275, 312264, 313153, 313794, 314145, 314641, 316561, 323715, 333634, 338700, 339090,
346890, 348934, 352594, 365130, 377475, 395232, 396996, 397092, 397122, 399012, 399042, 399138, 399768, 400902, 402372,
402852, 402882, 403092, 403212, 403218, 404742, 406278, 406662, 410532, 410562, 411042, 411276, 411282, 411402, 412422,
412806, 414342, 417912, 418182, 426852, 427236, 427362, 427572, 427596, 427602, 429126, 430662, 431142, 434502, 435222,
442566, 442662, 442902, 443406, 459492, 459618, 460002, 460332, 460338, 460362, 461382, 461862, 463398, 467142, 467238,
467478, 467493, 467523, 467982, 468003, 475302, 475662, 475683, 491622, 31200, 31440, 31536, 31560, 31920, 31944,
32040, 32280, 59844, 61860, 62145, 62220, 62730, 81442, 94882, 108353, 111881, 116417, 116513, 146964, 153288,
155148, 155154, 158312, 158498, 160866, 160929, 161289, 169668, 169764, 173480, 173762, 173858, 174866, 175250, 177290,
181160, 181922, 182028, 183050, 185097, 207626, 209546, 215690, 276176, 277700, 284336, 285860, 288578, 289136, 289880,
290321, 291056, 291176, 291416, 291617, 291896, 301480, 301720, 304048, 305560, 309382, 316677, 317510, 334529, 336808,
337570, 338996, 345873, 353321, 362680, 55664, 60937, 61784, 62098, 88680, 92610, 92760, 100320, 103876, 104212,
104856, 105990, 109329, 111249, 116116, 119941, 123000, 123270, 143176, 146242, 151336, 157936, 158476, 158776, 159976,
160312, 174436, 174484, 177202, 181604, 184628, 184882, 186437, 188716, 190486, 198608, 201668, 202148, 202388, 202529,
204293, 206609, 208133, 209948, 214244, 214673, 214793, 215378, 217208, 217292, 217733, 219158, 221573, 230753, 230993,
234533, 245957, 274224, 277284, 277794, 278028, 278034, 286218, 289380, 289554, 290058, 291220, 291594, 304833, 304929,
306314, 308305, 312993, 313994, 315724, 316489, 319794, 319818, 331540, 333580, 340741, 342085, 345738, 352554, 361300,
361684, 361780, 361804, 363589, 365125, 368965, 377125, 396962, 398740, 402338, 403082, 404102, 426722, 427562, 428582,
434342, 434702, 459220, 459601, 459985, 460060, 460081, 460105, 461125, 461845, 463045, 463126, 463141, 463171, 463381,
463885, 463891, 467221, 475285, 475405, 475411, 491605, 31400, 58818, 59028, 59154, 59810, 60044, 61794, 80240,
81164, 81185, 88304, 88850, 89172, 92392, 94424, 94520, 108193, 109705, 117830, 143024, 146594, 146954, 153960,
154200, 154305, 157377, 158040, 158865, 160140, 160146, 160401, 169746, 173708, 175178, 177225, 182604, 182921, 185385,
185475, 200466, 202506, 205668, 206412, 207410, 209196, 211011, 213844, 216195, 230226, 230610, 230706, 230730, 232515,
234051, 237891, 246051, 274120, 285452, 289336, 291380, 296912, 298690, 298786, 302597, 304786, 305740, 306437, 313556,
313652, 315512, 315572, 315953, 316037, 319877, 321550, 329648, 329672, 331202, 331553, 333218, 333473, 335363, 337250,
337553, 337673, 338132, 338314, 340088, 341123, 342038, 345441, 346232, 346697, 348338, 349209, 350222, 352643, 361313,
361697, 362033, 362057, 47344, 47464, 48401, 55528, 55700, 56417, 56585, 59608, 59704, 60497, 61624, 61652,
80532, 85444, 88408, 88737, 89226, 91044, 91554, 94764, 101316, 101796, 102036, 103332, 103362, 104076, 104082,
107412, 107916, 107922, 115596, 115602, 116106, 117073, 117381, 117507, 119113, 123925, 154385, 154676, 157076, 158290,
160330, 165808, 165832, 169409, 170324, 171523, 175132, 176248, 176465, 176899, 180964, 182392, 182482, 184529, 186382,
188803, 198568, 200332, 200353, 203060, 206473, 206968, 207493, 209012, 210958, 215218, 216089, 218126, 230068, 230113,
230572, 230953, 231973, 283504, 285090, 285778, 289164, 289929, 291465, 298764, 300817, 308250, 308265, 312162, 312546,
312966, 316438, 321555, 330660, 344802, 346314, 361164, 361260, 368805, 395160, 396180, 396684, 396690, 398220, 398226,
398730, 399480, 399750, 417894, 426420, 426444, 426450, 426780, 427164, 427290, 428310, 430230, 430350, 459180, 459186,
459210, 459420, 459441, 459465, 459546, 459561, 459930, 459945, 460950, 460965, 460995, 461070, 461091, 461325, 461331,
461835, 462990, 463011, 463371, 467085, 467091, 467211, 475275, 491550, 491565, 491571, 491595, 26592, 31128, 32261,
55890, 55953, 59018, 59730, 59745, 62022, 62502, 87489, 93226, 94793, 101825, 103841, 110918, 117286, 118961,
145857, 146060, 146508, 150978, 153828, 154017, 154761, 158090, 166850, 167330, 167508, 174641, 181893, 184428, 188949,
189453, 199586, 206049, 206214, 213894, 215196, 215241, 230090, 230186, 237731, 276712, 276824, 277265, 283880, 284888,
284984, 289169, 291020, 296872, 299938, 300434, 300682, 304714, 305272, 305797, 307529, 313925, 314405, 315971, 331346,
333201, 334034, 337478, 338458, 345477, 347171, 348266, 363043, 32006, 40769, 44706, 46676, 47500, 47506, 47713,
47878, 51056, 52568, 54488, 55046, 55690, 56069, 58700, 59572, 60547, 61642, 79252, 81030, 86884, 86932,
87436, 87442, 87628, 91348, 92466, 92742, 100248, 102026, 102153, 104568, 104838, 107402, 111651, 122982, 123075,
123405, 123411, 123915, 138692, 142192, 144740, 145042, 145624, 145720, 146185, 153784, 154883, 167569, 169609, 170284,
174289, 174723, 176305, 176425, 182449, 200081, 201617, 202121, 203141, 217619, 221285, 229841, 230681, 273648, 282004,
285234, 287572, 288052, 288978, 289861, 298401, 300594, 301362, 301386, 304518, 307314, 307749, 312436, 313626, 315918,
317451, 319638, 331398, 331401, 332682, 333354, 334122, 334885, 336609, 336774, 340515, 344692, 345196, 345289, 345370,
346435, 348451, 348685, 361795, 362773, 364693, 364813, 396170, 398452, 409972, 410716, 413782, 426410, 426650, 428174,
22480, 47753, 52578, 52620, 54072, 54858, 75632, 78568, 79442, 80056, 83696, 83816, 85766, 87224, 87686,
87689, 88604, 101281, 104195, 107729, 107825, 110758, 115826, 116905, 142568, 150868, 152788, 156882, 156978, 157002,
157065, 157868, 158001, 167468, 167558, 169350, 171030, 172934, 173637, 174250, 175139, 182819, 201930, 202371, 203811,
205596, 205938, 207126, 207891, 209046, 213618, 214122, 215317, 231699, 233526, 233619, 233739, 269252, 273256, 275796,
275852, 277066, 281320, 281537, 282377, 289097, 289411, 290986, 299858, 301445, 304678, 305329, 305446, 312707, 319589,
320011, 330641, 331052, 331058, 332972, 334211, 338957, 352355, 360881, 360905, 361126, 361241, 361625, 28244, 30049,
31366, 31491, 40609, 42736, 42856, 43864, 45796, 45985, 46726, 46853, 50920, 54113, 58193, 58540, 59941,
76737, 77016, 77028, 78732, 78738, 79242, 87601, 88645, 90828, 103249, 107124, 107186, 107628, 115398, 138148,
144785, 144940, 149336, 150418, 152396, 153158, 156340, 156788, 158233, 160025, 169265, 170371, 180892, 181361, 184595,
188515, 199561, 199796, 201905, 205699, 217173, 229801, 230041, 231565, 280420, 282156, 284419, 284788, 285225, 287532,
298188, 299913, 300358, 301590, 302094, 307477, 307731, 311916, 312474, 333097, 336492, 337174, 340243, 344745, 376903,
394872, 395142, 395892, 396396, 396402, 397932, 397938, 398442, 399462, 401652, 401772, 401778, 402012, 402492, 402522,
403542, 405558, 405582, 409836, 409842, 409962, 410172, 410202, 410682, 411702, 411726, 413742, 417822, 8160, 12240,
14256, 14280, 20400, 20424, 22440, 23908, 26520, 27540, 28050, 30258, 30840, 31109, 43874, 44216, 46284,
46634, 51938, 53964, 55590, 58154, 77386, 79185, 79493, 80074, 80515, 84664, 91461, 103594, 103718, 104981,
105485, 107209, 107305, 108694, 108814, 115306, 115313, 115817, 119303, 136944, 137064, 142241, 146457, 149409, 153714,
156842, 167301, 168620, 168837, 172658, 173349, 174605, 181525, 199116, 199850, 205418, 209415, 269969, 272216, 275128,
275170, 275786, 280280, 280376, 282182, 284198, 287410, 287434, 287858, 287942, 288874, 297862, 298162, 298313, 303814,
303898, 304003, 313613, 315477, 332913, 346259, 348243, 362635, 24145, 27873, 29522, 38384, 38744, 40148, 40244,
44803, 51089, 52012, 53876, 54821, 55388, 58033, 58565, 72532, 75240, 76596, 77361, 79401, 79980, 80421,
83794, 84690, 85545, 92220, 94294, 94357, 94477, 99960, 100230, 102772, 104550, 115036, 115875, 115981, 115987,
122910, 141016, 141112, 142918, 143828, 144097, 144682, 149176, 150857, 151078, 152362, 156358, 156457, 157891, 158731,
166578, 169486, 170149, 172713, 173158, 180841, 181334, 198008, 198533, 199276, 201329, 201833, 202853, 205036, 205169,
206933, 213233, 213353, 213443, 213593, 214157, 217133, 221213, 245783, 268225, 269217, 269644, 272100, 276037, 277525,
283804, 285709, 288933, 300124, 303729, 307341, 311665, 312085, 330100, 330410, 330844, 330858, 331033, 332020, 332140,
332380, 336114, 336220, 340021, 344284, 344380, 344470, 346165, 364587, 395882, 397660, 401642, 401978, 403502, 12200,
23394, 23778, 29980, 30275, 36708, 38754, 42468, 44148, 50898, 52038, 52803, 53682, 54147, 54595, 57948,
58022, 58458, 58659, 69104, 69352, 69464, 70640, 71384, 71396, 71480, 73034, 75448, 77129, 90902, 90947,
91301, 91406, 93191, 103484, 103619, 116813, 135105, 138450, 138546, 141153, 145514, 145521, 145701, 150405, 152010,
152730, 153123, 156273, 167081, 168306, 169050, 172485, 173331, 174347, 180570, 180630, 180645, 181515, 184395, 199770,
200946, 201066, 201306, 201786, 205146, 208947, 213210, 213306, 281894, 283461, 284867, 287429, 291335, 296312, 296837,
297586, 297673, 300197, 300323, 301157, 305237, 312461, 315437, 319517, 328568, 328952, 329541, 329603, 330353, 332393,
332438, 332441, 333923, 336233, 336473, 338003, 340043, 352283, 368663, 15945, 23857, 24195, 29137, 35824, 36328,
36568, 36664, 36754, 37864, 38584, 39604, 40108, 42889, 44566, 46358, 47164, 47194, 50132, 50657, 52390,
52750, 52757, 56327, 57586, 61523, 69474, 72396, 72492, 83740, 86386, 86809, 87130, 87142, 100596, 100716,
100722, 100956, 101667, 101907, 102636, 103513, 106716, 106812, 107574, 114876, 115758, 136664, 142723, 142885, 145939,
149969, 150725, 152357, 154123, 159797, 159821, 159883, 164728, 165112, 165763, 166597, 167182, 169109, 170083, 180883,
188443, 197368, 201485, 205369, 206893, 213557, 215111, 229685, 229709, 237583, 267748, 268084, 271800, 272820, 273222,
273564, 275737, 280018, 281713, 283057, 283081, 284821, 287337, 287893, 287907, 288013, 298044, 299580, 301203, 303885,
304227, 304395, 311529, 311638, 311694, 312403, 329964, 330153, 330300, 332122, 332854, 336060, 338055, 344293, 344405,
393720, 394854, 395484, 395580, 395610, 397500, 397530, 397626, 399390, 8088, 15569, 20196, 20321, 22241, 26232,
26417, 26441, 26501, 27337, 27461, 27929, 27971, 28195, 29340, 29478, 29849, 30821, 30883, 39626, 39722,
40050, 42780, 42822, 43466, 45482, 47253, 54330, 55353, 55437, 58425, 59655, 72370, 74722, 75562, 79961,
82913, 83316, 84394, 85142, 86438, 86723, 90682, 91283, 92299, 101059, 101434, 102970, 103075, 106841, 107086,
107107, 107783, 114905, 114998, 115001, 134120, 136674, 137041, 140748, 141746, 144579, 148908, 149100, 152677, 155868,
156323, 156435, 166250, 166490, 166677, 167013, 168156, 168170, 168549, 172250, 172346, 180410, 180782, 184359, 199014,
204986, 265688, 265912, 267224, 267980, 269482, 271841, 275267, 276538, 280940, 281180, 281254, 290867, 295672, 296726,
298133, 298507, 304181, 304205, 305197, 307271, 333063, 336523, 337963, 360659, 361511, 22988, 22994, 23834, 29097,
29349, 29846, 36682, 38609, 39532, 41932, 46181, 51862, 53849, 59467, 61483, 71497, 72873, 72982, 76209,
76454, 78438, 82890, 83377, 85093, 87307, 88327, 90709, 98808, 99942, 100586, 104478, 108615, 134584, 136120,
136988, 137905, 137926, 141766, 143782, 144470, 151910, 156301, 156935, 165212, 168505, 197852, 198245, 199001, 199267,
199757, 200921, 201017, 202781, 265953, 266065, 271314, 273510, 273513, 273603, 274801, 274885, 275979, 280170, 281241,
283299, 283726, 283747, 287790, 299238, 303321, 303417, 312363, 394582, 395450, 14211, 14964, 15132, 21880, 21940,
26281, 27290, 29781, 30774, 30798, 38346, 38697, 42354, 43244, 46227, 47181, 50748, 51770, 69329, 76138,
76393, 76579, 76886, 76963, 78065, 78245, 78275, 78926, 84749, 86249, 86582, 94247, 100921, 101429, 102614,
102755, 106681, 107143, 134858, 134961, 136874, 138006, 138517, 140204, 141884, 142478, 143961, 144525, 148425, 149738,
150117, 150627, 151957, 152659, 153735, 166076, 166118, 172619, 176151, 180807, 198051, 199431, 204981, 265186, 268058,
269462, 273683, 274748, 274790, 275030, 275086, 275093, 276743, 279452, 279466, 280259, 283701, 283787, 287315, 287819,
296549, 297178, 299194, 299783, 301085, 303443, 329102, 329315, 329945, 333851, 8018, 11236, 11746, 12166, 13778,
15045, 15514, 15525, 38252, 38300, 43429, 45289, 45622, 47155, 53477, 57653, 72853, 75414, 75459, 78603,
84198, 85075, 86350, 88107, 90294, 99670, 101511, 102957, 136052, 137833, 138019, 140698, 141673, 142390, 155833,
157735, 165475, 169031, 170011, 197973, 198841, 199198, 205003, 267634, 269401, 271305, 272149, 273927, 274851, 279961,
280764, 281187, 282934, 297315, 297525, 311475, 311847, 328410, 344599, 345103, 360975, 393702, 393942, 394038, 394062,
394422, 394446, 394542, 394782, 7800, 8069, 11640, 13176, 13560, 13937, 15459, 15627, 19320, 19704, 20082,
20124, 20138, 21240, 21958, 22293, 25080, 26213, 26982, 27701, 28902, 29741, 30749, 36294, 36465, 39699,
42218, 42585, 44078, 49644, 51603, 52275, 54301, 69020, 71004, 72041, 74228, 75161, 76133, 78134, 83289,
86325, 136428, 137786, 143715, 148713, 149077, 150558, 150573, 151734, 151779, 152094, 172149, 172317, 172583, 268529,
272950, 273163, 279281, 284711, 296277, 303214, 328589, 11884, 12058, 19881, 35740, 37618, 37786, 39257, 40075,
42389, 43603, 50765, 51406, 53422, 53534, 70569, 72750, 76173, 78366, 84181, 84531, 98790, 99534, 99870,
100683, 103447, 137813, 140185, 140633, 141067, 141541, 148793, 148883, 149131, 149806, 180317, 197093, 197333, 197933,
198173, 265669, 267114, 268630, 268941, 271269, 271929, 272013, 279385, 280971, 283015, 287047, 299211, 303275, 311358,
328501, 328885, 393902, 7121, 7625, 7974, 11000, 11929, 13254, 14094, 14905, 20243, 21923, 22869, 23093,
23627, 25942, 26190, 27182, 29003, 36249, 41457, 41660, 41829, 42198, 43235, 45447, 45597, 51318, 68506,
71013, 71253, 72275, 78029, 78875, 83531, 84254, 86174, 100531, 132978, 133620, 133980, 134385, 134745, 134925,
136131, 136758, 137550, 140118, 140883, 141651, 143646, 143915, 144423, 148323, 165069, 166091, 166215, 197427, 197451,
268846, 269363, 272590, 274643, 279443, 280115, 280781, 280883, 281627, 282909, 283163, 295397, 295637, 296110, 296237,
296477, 299291, 328163, 329243, 19772, 25486, 26925, 35689, 36493, 39559, 42541, 42547, 43563, 44059, 45214,
50357, 51499, 69174, 70332, 72775, 76317, 82734, 114831, 135665, 136781, 136967, 140501, 140935, 164323, 164653,
165043, 165403, 197293, 197831, 263922, 267859, 268597, 272670, 295731, 295755, 297159, 393630, 6648, 7781, 12780,
14550, 18342, 19363, 21389, 21837, 25061, 26141, 37710, 38070, 46095, 50379, 53403, 71987, 74573, 74873,
74926, 77995, 83495, 99143, 137390, 137421, 140085, 149787, 150039, 155799, 263156, 264041, 265018, 265061, 265805,
266981, 267065, 267449, 267477, 267821, 267911, 268871, 274541, 295597, 328391, 7571, 7822, 11155, 11605, 13018,
13141, 13619, 14638, 19790, 21205, 21806, 26739, 28951, 35638, 36179, 37685, 41398, 66540, 68325, 68430,
69003, 70074, 74190, 76059, 82381, 98718, 136523, 141469, 148253, 151823, 164926, 166039, 168079, 197021, 200783,
263658, 281103, 286863, 295518, 10716, 13639, 17394, 17884, 19847, 21299, 25206, 35258, 35723, 37305, 37333,
49835, 49947, 67034, 68054, 68218, 70535, 82526, 82541, 82589, 90191, 134699, 138255, 139934, 140397, 148647,
151611, 198711, 266702, 270766, 271655, 272539, 278957, 279831, 295325, 297019, 299063, 328091, 328791, 329807, 6038,
6883, 11469, 12979, 13101, 13847, 14863, 19275, 25011, 34606, 35502, 35998, 36391, 42151, 49607, 50279,
71847, 75279, 83005, 98907, 132058, 133843, 147883, 164251, 265389, 265515, 280663, 393342, 2040, 6629, 6989,
7469, 7707, 10038, 10933, 17850, 18221, 18869, 19117, 21102, 22031, 24989, 33753, 35277, 37995, 39183,
41310, 41627, 41751, 43151, 66998, 67277, 74855, 100399, 132789, 134439, 135597, 135837, 135963, 141371, 264569,
264659, 264862, 264883, 264989, 17337, 18247, 19143, 19549, 35065, 35189, 38999, 41323, 67251, 82279, 98430,
98871, 134299, 135847, 136253, 139507, 141391, 148559, 196733, 264438, 270951, 3798, 5835, 6582, 9146, 9675,
10667, 11415, 12699, 13374, 24811, 24894, 34503, 51231, 66422, 66515, 66926, 70799, 74135, 147799, 263603,
263851, 274463, 278747, 295037, 327803, 327983, 7262, 9645, 9835, 10846, 17269, 17654, 17779, 34219, 67351,
68247, 82107, 134231, 139663, 163963, 264615, 264795, 2019, 6557, 9075, 18071, 20669, 24701, 35919, 67774,
67901, 69975, 133775, 140319, 263951, 264427, 2005, 3446, 3502, 5038, 5363, 9070, 33703, 66795, 135383,
135479, 263515, 270519, 1998, 4967, 6363, 8941, 10551, 12411, 20687, 66875, 68639, 131819, 131991, 263399,
278639, 3645, 9403, 17086, 20599, 34103, 262999, 266415, 1947, 4587, 6269, 8919, 2781, 2863, 5407,
17019, 41023, 3195, 69695, 131703, 263327, 4478, 8446, 33391, 1631, 33015, 131439, 1405, 2399, 262719,
65759, 131263$\}$

\medskip

$\cS_{19,4,9}=\{$523264, 520960, 518784, 511744, 513664, 514432, 514624, 517504, 517696, 518464, 520384, 506400, 513344, 474880, 488992,
489760, 490000, 505120, 505360, 506128, 511168, 520240, 441984, 454176, 456864, 457224, 488720, 502944, 503304, 505992,
518184, 196352, 228992, 245120, 245312, 253472, 257312, 257552, 259232, 259592, 260192, 260240, 260360, 260612, 425344,
425536, 441664, 449824, 450064, 453904, 455776, 455824, 455944, 456196, 456784, 456964, 474304, 482464, 482824, 486496,
486544, 486664, 486916, 488584, 489544, 489604, 498784, 498832, 498952, 499204, 502864, 503044, 504904, 504964, 505924,
508864, 511024, 513064, 514072, 514084, 517144, 517156, 518164, 520204, 375680, 387872, 389792, 391682, 228672, 253200,
259152, 259332, 453768, 482384, 482564, 488516, 500640, 502530, 504450, 513044, 516642, 358272, 359232, 374592, 382752,
383760, 386832, 389472, 389520, 389712, 390402, 390657, 391425, 195776, 253064, 257096, 257156, 449608, 449668, 453700,
472000, 474160, 480160, 482050, 484192, 484240, 485122, 486145, 487810, 488002, 488065, 496480, 496528, 497410, 498433,
500560, 501505, 504130, 504193, 504385, 508720, 510988, 512290, 512530, 512545, 516370, 516385, 516625, 520195, 323456,
326336, 372416, 380576, 383624, 386696, 387296, 387464, 387656, 388226, 252996, 437184, 441384, 445344, 449154, 451296,
451464, 452226, 452994, 453186, 480080, 481025, 487745, 494304, 494472, 495234, 500424, 501954, 508584, 510114, 510474,
512273, 516234, 122752, 126784, 128704, 129472, 319296, 325056, 351936, 355776, 368064, 380256, 380304, 380496, 382176,
382344, 382536, 383184, 383304, 384066, 384129, 386256, 386376, 388161, 193472, 195632, 224192, 228392, 236480, 239552,
244760, 244772, 247712, 248672, 248720, 250592, 250760, 252545, 254432, 254672, 254792, 254852, 255617, 256385, 256577,
416704, 419776, 424984, 424996, 432064, 441364, 444256, 444304, 445264, 446944, 447184, 447304, 447364, 447874, 448066,
448834, 451024, 451396, 451906, 462784, 471856, 474124, 476896, 477064, 477664, 477904, 478024, 478084, 478594, 478786,
478849, 479944, 481474, 483784, 484036, 484546, 485569, 493024, 493264, 493384, 493444, 494032, 494404, 494914, 494977,
495169, 496072, 496324, 496834, 497857, 500164, 500929, 508144, 508264, 508312, 508324, 508504, 508516, 508564, 508684,
509026, 509074, 509089, 509194, 509446, 509449, 510034, 510049, 510097, 510214, 510217, 510469, 510979, 512074, 512134,
512137, 516166, 516169, 516229, 311072, 326192, 343712, 358952, 372272, 374312, 374960, 375080, 375320, 375842, 385922,
219072, 228372, 247632, 250320, 250692, 251265, 251457, 252225, 255297, 408480, 424482, 437040, 438960, 439080, 439842,
440610, 440850, 445128, 469680, 469800, 470562, 471720, 473250, 473610, 476624, 476996, 478529, 479684, 480449, 500280,
501810, 503850, 508244, 509009, 509189, 512069, 98080, 114448, 128560, 129328, 294672, 324912, 343392, 343440, 343632,
351792, 355632, 357552, 357672, 357912, 358512, 358680, 359442, 359457, 367920, 373872, 374040, 375825, 381762, 381825,
385857, 184256, 193328, 195596, 210848, 211808, 211856, 224048, 225968, 226088, 227873, 236336, 239408, 242096, 242288,
242456, 242468, 243233, 244001, 244241, 247496, 248264, 248516, 407392, 407440, 408400, 416560, 419632, 422320, 422512,
422680, 422692, 423202, 423442, 424210, 431920, 436904, 438640, 439060, 439570, 443848, 444100, 444868, 462640, 464560,
464680, 465328, 465520, 465688, 465700, 466210, 466450, 466465, 468400, 468592, 468760, 468772, 469360, 469780, 470290,
470305, 470545, 471280, 471400, 471448, 471460, 471640, 471652, 471700, 471820, 472162, 472210, 472225, 472330, 472582,
472585, 473170, 473185, 473233, 473350, 473353, 473605, 474115, 479800, 481330, 483640, 483892, 484402, 485425, 487450,
487462, 487465, 495928, 496180, 496690, 497713, 500020, 500785, 503830, 503833, 503845, 508108, 508675, 65184, 114312,
126504, 129192, 294536, 310496, 310664, 310856, 319016, 322736, 322856, 323096, 324776, 325736, 325784, 326666, 355496,
367784, 371816, 371864, 373538, 379586, 130057, 176032, 178912, 179080, 191152, 191272, 193192, 210768, 218928, 223912,
225648, 226068, 226593, 226833, 227601, 236200, 239272, 240040, 240232, 240280, 240292, 242961, 247236, 405216, 405384,
408264, 416424, 419496, 420264, 420456, 420504, 420516, 421026, 421386, 424074, 429744, 429864, 431784, 432552, 432744,
432792, 432804, 433314, 433674, 435624, 435816, 435864, 435876, 436464, 436584, 436632, 436644, 436824, 436836, 436884,
437004, 437346, 437394, 437514, 437766, 438504, 438924, 439434, 440394, 440454, 444984, 448554, 450744, 451116, 451626,
452634, 452646, 462504, 464240, 464660, 466193, 469224, 469644, 470154, 471380, 472145, 472325, 479540, 480305, 487445,
493752, 494124, 494634, 499884, 501774, 64864, 64912, 65104, 97504, 97672, 97864, 113872, 113992, 122032, 122152,
122392, 126064, 126232, 128104, 128152, 129112, 130054, 294096, 294216, 318576, 318744, 324696, 326661, 351336, 351384,
355416, 357138, 357153, 367704, 371362, 373521, 378306, 378561, 379329, 147392, 155552, 159584, 159632, 161504, 161672,
162272, 162512, 162632, 162692, 171872, 171920, 175952, 177632, 177872, 177992, 178052, 178640, 179012, 184112, 186032,
186152, 186800, 186992, 187160, 187172, 189872, 190064, 190232, 190244, 190832, 191252, 192752, 192872, 192920, 192932,
193112, 193124, 193172, 193292, 195587, 204512, 204680, 208352, 208592, 208712, 208772, 210632, 211400, 211652, 216752,
216872, 218792, 219560, 219752, 219800, 219812, 222632, 222824, 222872, 222884, 223472, 223592, 223640, 223652, 223832,
223844, 223892, 224012, 224417, 224777, 225512, 225932, 227465, 232880, 233072, 233240, 233252, 234920, 235112, 235160,
235172, 235760, 235880, 235928, 235940, 236120, 236132, 236180, 236300, 236705, 237065, 238832, 238952, 239000, 239012,
239192, 239204, 239252, 239372, 239960, 239972, 240020, 240212, 240737, 240785, 240905, 241157, 241880, 241892, 242060,
242252, 242825, 243785, 243845, 247352, 248120, 248372, 250040, 250412, 251945, 254072, 254132, 254252, 254492, 255017,
256025, 256037, 400864, 401104, 401224, 401284, 404944, 405316, 406984, 407236, 408004, 413104, 413296, 413464, 413476,
415144, 415336, 415384, 415396, 415984, 416104, 416152, 416164, 416344, 416356, 416404, 416524, 416866, 416914, 417034,
417286, 419056, 419176, 419224, 419236, 419416, 419428, 419476, 419596, 420184, 420196, 420244, 420436, 420946, 421126,
422104, 422116, 422284, 422476, 422986, 423046, 424006, 429424, 429844, 431344, 431464, 431512, 431524, 431704, 431716,
431764, 431884, 432472, 432484, 432532, 432724, 433234, 433414, 435544, 435556, 435604, 435796, 436564, 438484, 438604,
439366, 443704, 443956, 444724, 446584, 446644, 446764, 447004, 447514, 447526, 448534, 450676, 450844, 451606, 462064,
462184, 462232, 462244, 462424, 462436, 462484, 462604, 464104, 464524, 465112, 465124, 465292, 465484, 465994, 466054,
466057, 468184, 468196, 468364, 468556, 469204, 469324, 470086, 470089, 470149, 471244, 471811, 476344, 476716, 477304,
477364, 477484, 477724, 478234, 478246, 478249, 479404, 481294, 483436, 483484, 484366, 485389, 492664, 492724, 492844,
493084, 493684, 493852, 494614, 494617, 494629, 495724, 495772, 496654, 497677, 499804, 500749, 507964, 508099, 293762,
308192, 309122, 310082, 318242, 320432, 321314, 322322, 324002, 324194, 324242, 324362, 338912, 339842, 342722, 350882,
353192, 353954, 354722, 354914, 354962, 354977, 355082, 357002, 363440, 364322, 365480, 366242, 367010, 367202, 367250,
367265, 367370, 369392, 369512, 369560, 370082, 370274, 370322, 370337, 370442, 371042, 371090, 371105, 371282, 371297,
371345, 371465, 372227, 372962, 373130, 373322, 373385, 374915, 379442, 381482, 385202, 385322, 385562, 385577, 387107,
97153, 112513, 113473, 121633, 124705, 125713, 127393, 127585, 127633, 127753, 155472, 161232, 161604, 175816, 183976,
185712, 186132, 190696, 191116, 192852, 204240, 204612, 210372, 216432, 216852, 218352, 218472, 218520, 218532, 218712,
218724, 218772, 218892, 219480, 219492, 219540, 219732, 220257, 220305, 220425, 220677, 222552, 222564, 222612, 222804,
223572, 224337, 224517, 225492, 225612, 226377, 226437, 227397, 234840, 234852, 234900, 235092, 235860, 236625, 236805,
238932, 242757, 247092, 249972, 250140, 250905, 250917, 251925, 254997, 408120, 415064, 415076, 415124, 415316, 416084,
418722, 419156, 429288, 429708, 431010, 431444, 435042, 435090, 436428, 444588, 461730, 462164, 464084, 464204, 465989,
467682, 467850, 469635, 476276, 476444, 478229, 479324, 480269, 492210, 492330, 494115, 499875, 500235, 95200, 111568,
119728, 123760, 127750, 291792, 292674, 316272, 317202, 322186, 323922, 324357, 333792, 334800, 335682, 335745, 337872,
339777, 341442, 341697, 342465, 346032, 346992, 347922, 347937, 348912, 349032, 349080, 349602, 349794, 349842, 349857,
349962, 350562, 350610, 350625, 350802, 350817, 350865, 350985, 351747, 352752, 353112, 353634, 353682, 353697, 353874,
353889, 353937, 354057, 354642, 354657, 354705, 354897, 355587, 356562, 356577, 356682, 356745, 356937, 357507, 358467,
362352, 364305, 365040, 365400, 365922, 365970, 365985, 366162, 366177, 366225, 366345, 366930, 366945, 366993, 367185,
367875, 370002, 370017, 370065, 370257, 371025, 372945, 373065, 373827, 378162, 378417, 379185, 381042, 381105, 381210,
381225, 381465, 381987, 382995, 385137, 385305, 386067, 96065, 120593, 125577, 127313, 147248, 155336, 159176, 159428,
171464, 171716, 175556, 183536, 183656, 183704, 183716, 183896, 183908, 183956, 184076, 185576, 185996, 186584, 186596,
186764, 186956, 189656, 189668, 189836, 190028, 190676, 190796, 192716, 193283, 210488, 211256, 211508, 216296, 216716,
218452, 222113, 223436, 232664, 232676, 232844, 233036, 234401, 235724, 238433, 238481, 238796, 246956, 247916, 247964,
406840, 407092, 407860, 412888, 412900, 413068, 413260, 414562, 414610, 415948, 418642, 419020, 429268, 429388, 430930,
431308, 443500, 443548, 444508, 459760, 460642, 460690, 460705, 461650, 461665, 461713, 462028, 462595, 463330, 463570,
463585, 463690, 463750, 463753, 464515, 465283, 465475, 467410, 467425, 467665, 467782, 467785, 467845, 468355, 468547,
469315, 471100, 471235, 475570, 475762, 475825, 475930, 475942, 475945, 476707, 477475, 477715, 479395, 479755, 483427,
483475, 483595, 483847, 491890, 491953, 492145, 492310, 492313, 492325, 492835, 493075, 493843, 495715, 495763, 495883,
496135, 499795, 499975, 507955, 60384, 109512, 117672, 123624, 125574, 289736, 290498, 298976, 302024, 302786, 305096,
306626, 313256, 314096, 314216, 314264, 314786, 314978, 315026, 315146, 316136, 317066, 317666, 317834, 318026, 319976,
320216, 320738, 320906, 321098, 321746, 321866, 322181, 323786, 342578, 346856, 347786, 349522, 349537, 349585, 349777,
350545, 353617, 354506, 362216, 362984, 363224, 363746, 363914, 364106, 364169, 365905, 366794, 369866, 370889, 378026,
378986, 379034, 379049, 93889, 106177, 110017, 118177, 118369, 118417, 118537, 120457, 121057, 121225, 121417, 124129,
124297, 124489, 125137, 125257, 127177, 147112, 155076, 175672, 178360, 178732, 183636, 185556, 185676, 191107, 210228,
216276, 216396, 217953, 218001, 218316, 222033, 234321, 246876, 404664, 405036, 407724, 412386, 412554, 418506, 426984,
427746, 427914, 428514, 428754, 428874, 428934, 430794, 434634, 434886, 436284, 442794, 442986, 443034, 443046, 460625,
461514, 463313, 463685, 464195, 475505, 475925, 476435, 479315, 479495, 491754, 492174, 493707, 55264, 56272, 59344,
85984, 89032, 92104, 102352, 104392, 116464, 116584, 116632, 117232, 117592, 118534, 119272, 119512, 120454, 121222,
121414, 123352, 124294, 124486, 125254, 127174, 282576, 284616, 286146, 309802, 312816, 313176, 314706, 315141, 315864,
316626, 316746, 317061, 317829, 318021, 320901, 321093, 321861, 323781, 341298, 341553, 342321, 345576, 345816, 346584,
347346, 347361, 347466, 347529, 347721, 349386, 350409, 353481, 361944, 363729, 363849, 365769, 377946, 377961, 378009,
378969, 89537, 113193, 118097, 120017, 120137, 146672, 146792, 146840, 146852, 147032, 147044, 147092, 147212, 155192,
159032, 159284, 160952, 161324, 161912, 161972, 162092, 162332, 171320, 171572, 175412, 177272, 177332, 177452, 177692,
178292, 178460, 181232, 183500, 184067, 185987, 186755, 186947, 189827, 190019, 190787, 192572, 192707, 203960, 204332,
207992, 208052, 208172, 208412, 210092, 211052, 211100, 213992, 215777, 215945, 221897, 223292, 230360, 230372, 231137,
231305, 231905, 232145, 232265, 232325, 234185, 235580, 238025, 238277, 238652, 246185, 246377, 246425, 246437, 400504,
400564, 400684, 400924, 404596, 404764, 406636, 406684, 407644, 410584, 410596, 411106, 411346, 411466, 411526, 412114,
412486, 414154, 414406, 415804, 418246, 418876, 426964, 427474, 427846, 430534, 431164, 442714, 442726, 442774, 442966,
459724, 460234, 460486, 460489, 461254, 461257, 461509, 461884, 462019, 471091, 475354, 475366, 475369, 475534, 475726,
475789, 476299, 477259, 477319, 491734, 491737, 491749, 491854, 491917, 492109, 492619, 492679, 493639, 507919, 48048,
80808, 109368, 111288, 113190, 120134, 274352, 276392, 277232, 277352, 277400, 277922, 278114, 278162, 278282, 289592,
290354, 291512, 292394, 293042, 293162, 293402, 301880, 302642, 304952, 306482, 307640, 307832, 308402, 308522, 308762,
309362, 309530, 309797, 314570, 316741, 323642, 334520, 335402, 337592, 338360, 338552, 339122, 339242, 339482, 339497,
341162, 342122, 342170, 342185, 354362, 361442, 363395, 366650, 369347, 369722, 370745, 377507, 81313, 81505, 81553,
81673, 93745, 95785, 96433, 96553, 96793, 106033, 109873, 111793, 111913, 112153, 112753, 112921, 117961, 127033,
146772, 154932, 160884, 161052, 175276, 178723, 185667, 203892, 204060, 210012, 213972, 214497, 214737, 214857, 214917,
215505, 215877, 217545, 217797, 218172, 221637, 230865, 231237, 233925, 246105, 246117, 246165, 246357, 400050, 400170,
402360, 403122, 403242, 403890, 404082, 404250, 404262, 405930, 406122, 406170, 406182, 418362, 430650, 434490, 434742,
460229, 461370, 467130, 467502, 469035, 475349, 475469, 476231, 30640, 31600, 46960, 73648, 77680, 79600, 79720,
79768, 80368, 80728, 81670, 88888, 91960, 94648, 94840, 95782, 96550, 96790, 104248, 110968, 111910, 112150,
112918, 117958, 127030, 270192, 275952, 276312, 277842, 278277, 284472, 286002, 291192, 291954, 292122, 292389, 293157,
293397, 306346, 308517, 308757, 309525, 314565, 323637, 333240, 333432, 334200, 334962, 335025, 335130, 335145, 335385,
337272, 339057, 339225, 341082, 341097, 341145, 342105, 345042, 345057, 345987, 346947, 348867, 349242, 350265, 352707,
353337, 361425, 362307, 364995, 365625, 377187, 377235, 377427, 48899, 63011, 63779, 64019, 81233, 89393, 95345,
95513, 109737, 144368, 146636, 154796, 158828, 158876, 171116, 171164, 175196, 177443, 177683, 178451, 181196, 183356,
183491, 192563, 201656, 203441, 203561, 205688, 205748, 206513, 206633, 207281, 207473, 207641, 207653, 209321, 209513,
209561, 209573, 221753, 234041, 237881, 238133, 398200, 398260, 398770, 398962, 399130, 399142, 399730, 400150, 402292,
402802, 403222, 405850, 405862, 405910, 406102, 414010, 414262, 418102, 430390, 459580, 459715, 460090, 460342, 460345,
461110, 461113, 461365, 461875, 462970, 463030, 463033, 463150, 463390, 463405, 463915, 464923, 464935, 467062, 467065,
467125, 467230, 467245, 467485, 467995, 468007, 469015, 471055, 28584, 31464, 40872, 44784, 44904, 44952, 46824,
47592, 47832, 48773, 55992, 59064, 59832, 60024, 60965, 63653, 64013, 77544, 95510, 102072, 108792, 109734,
110094, 112782, 124974, 147075, 159267, 161955, 162315, 270056, 273896, 274136, 277706, 282296, 285866, 289016, 289898,
289946, 292117, 298424, 298616, 301304, 302186, 302234, 304376, 306266, 306341, 306701, 309389, 312266, 314426, 321581,
336818, 338723, 340643, 369203, 32390, 48518, 48710, 56870, 60710, 60950, 62630, 62739, 62990, 63590, 63638,
63758, 81097, 89257, 93289, 93337, 105577, 105625, 109657, 115657, 117817, 142312, 154716, 166840, 172792, 172972,
175267, 175627, 178315, 190507, 201588, 202161, 202353, 202521, 202533, 203121, 203541, 206193, 206613, 209241, 209253,
209301, 209493, 217401, 217653, 221493, 233781, 396024, 396204, 396714, 396906, 396954, 396966, 399594, 400014, 402156,
402666, 403086, 403674, 403686, 403854, 404046, 411834, 412206, 417966, 426684, 427194, 427566, 428154, 428214, 428334,
428574, 430254, 434286, 434334, 460085, 460974, 462965, 463133, 463895, 24304, 24424, 24472, 28144, 28504, 30184,
30424, 31192, 32131, 32323, 40432, 40792, 46552, 48453, 54712, 54904, 55672, 56613, 56853, 58744, 60693,
62565, 62613, 62733, 63573, 73192, 73432, 77272, 81094, 85432, 85624, 88312, 89254, 89614, 91384, 93286,
93334, 93454, 95374, 96334, 101752, 103672, 105574, 105622, 105742, 109654, 111694, 115654, 117814, 119854, 120862,
123934, 146755, 154915, 155155, 158995, 160867, 160915, 161035, 161287, 161875, 162055, 269784, 277701, 281976, 283896,
285786, 285861, 286221, 289893, 289941, 290061, 291981, 292941, 302181, 302229, 302349, 304042, 306261, 308301, 312261,
314421, 316461, 317469, 320541, 332658, 332721, 333603, 334611, 336753, 337683, 340323, 340371, 340563, 348723, 352563,
364851, 32070, 48323, 56483, 56598, 56843, 60515, 60563, 60683, 62550, 62603, 63563, 89177, 93703, 96391,
105991, 107433, 109831, 111751, 112711, 120871, 123943, 124951, 126991, 135152, 137192, 138200, 138212, 141272, 141284,
142292, 144332, 146492, 149432, 150392, 150452, 152312, 152492, 156152, 156404, 156524, 156572, 165752, 165812, 166772,
168440, 168692, 168812, 168860, 171107, 171155, 171275, 172532, 172892, 175187, 177227, 181052, 181187, 183347, 185387,
186395, 189467, 198392, 198572, 199160, 199412, 199532, 199580, 200105, 200297, 200345, 200357, 201452, 202985, 203405,
205292, 205532, 206057, 206477, 207065, 207077, 207245, 207437, 213692, 215225, 215597, 221357, 229820, 230012, 230585,
230957, 231545, 231605, 231725, 231965, 233645, 237677, 237725, 394744, 394996, 395116, 395164, 395764, 396124, 396634,
396646, 396694, 396886, 397804, 398044, 398554, 398566, 398734, 398926, 399574, 399694, 401884, 402646, 402766, 410044,
410236, 410746, 410806, 410926, 411166, 411766, 411934, 413806, 413854, 417886, 426364, 427126, 427294, 430174, 459004,
459571, 459886, 459934, 459949, 460894, 460909, 460957, 461839, 30597, 31557, 60557, 89174, 107430, 158859, 174855,
176775, 188967, 269282, 272354, 273362, 275402, 277562, 281522, 283562, 285781, 287474, 287594, 287642, 296882, 297842,
299762, 299882, 299930, 303602, 303962, 304037, 312122, 329642, 330482, 330602, 330650, 330665, 331427, 332522, 334475,
336362, 336602, 336617, 337547, 338147, 338315, 338507, 344762, 346667, 352427, 360890, 361082, 361145, 362027, 362675,
362795, 363035, 364715, 368747, 368795, 56403, 56462, 60494, 72673, 75745, 76753, 78793, 80953, 84913, 86953,
89351, 90865, 90985, 91033, 95303, 100273, 101233, 103153, 103273, 103321, 106993, 107353, 115513, 119831, 137172,
149364, 152052, 152412, 166636, 172963, 198132, 198492, 200025, 200037, 200085, 200277, 201180, 201945, 201957, 202125,
202317, 202965, 203085, 206037, 206157, 213372, 214137, 214197, 214317, 214557, 215157, 215325, 217197, 217245, 221277,
230517, 230685, 233565, 289543, 291463, 301831, 304903, 307591, 307783, 315943, 319783, 320023, 459435, 459869, 56397,
78790, 80950, 86950, 90982, 91030, 103270, 103318, 107350, 115510, 154699, 154759, 158791, 169735, 176455, 184615,
184855, 188695, 268242, 275397, 277557, 280434, 283122, 283482, 283557, 287589, 287637, 297706, 299877, 299925, 303957,
312117, 329202, 329457, 329562, 329577, 329625, 330225, 330585, 331107, 331155, 331347, 332250, 332265, 332505, 333027,
333195, 333387, 334035, 334155, 336345, 337107, 337227, 344442, 344505, 344697, 345267, 345387, 345627, 346227, 346395,
348267, 348315, 352347, 360825, 361587, 361755, 364635, 46019, 48179, 54179, 55943, 58211, 58259, 59015, 59783,
59975, 71633, 83825, 86513, 86873, 101097, 117775, 135116, 144188, 149228, 149996, 150236, 165356, 165596, 166364,
168803, 168851, 172883, 180476, 181043, 284423, 291143, 315671, 458995, 459115, 459163, 459175, 459355, 459367, 459415,
459535, 24261, 28101, 31125, 43973, 48173, 52133, 58085, 58253, 86870, 101094, 101262, 107214, 115374, 142275,
146475, 150435, 156387, 156555, 167559, 174279, 182439, 182799, 188559, 266186, 275258, 280298, 281066, 281306, 283477,
296426, 296666, 297434, 297701, 297869, 303821, 311546, 311981, 27590, 30293, 31783, 39878, 42950, 48158, 51110,
52070, 52118, 53990, 54099, 54158, 54663, 54855, 55623, 57830, 58070, 58190, 58695, 69577, 78649, 83689,
84457, 84697, 99817, 100057, 100825, 114937, 142012, 145991, 148956, 154151, 158231, 166627, 166795, 172747, 180907,
282247, 288967, 298375, 298567, 301255, 304327, 312487, 312847, 313447, 313495, 313615, 315535, 319567, 459095, 23491,
26563, 31278, 31771, 38853, 51045, 51093, 52053, 53733, 53973, 54093, 57813, 69574, 78646, 83686, 83854,
84454, 84694, 84814, 86734, 90574, 99814, 100054, 100174, 100822, 102862, 114934, 115054, 115102, 115294, 137155,
144775, 146455, 149347, 149395, 150355, 152035, 152275, 152395, 152455, 156115, 156487, 156967, 167239, 169159, 181351,
181399, 181519, 182359, 184399, 266181, 273082, 275253, 280026, 280293, 280461, 281061, 281301, 281421, 283341, 287181,
296421, 296661, 296781, 297429, 299469, 311541, 311661, 311709, 311901, 22470, 28237, 36803, 45875, 50915, 51030,
51083, 51683, 51923, 52043, 53718, 53963, 57803, 76473, 80911, 83417, 84871, 90823, 100231, 101191, 103111,
106951, 115111, 115303, 115351, 115471, 134972, 136892, 137660, 137852, 140732, 140924, 141692, 143612, 165347, 165587,
165707, 166355, 168395, 180467, 180587, 180635, 180827, 281927, 283847, 312407, 458959, 15267, 23949, 29998, 43829,
45749, 45869, 51917, 76470, 76590, 78510, 83414, 107070, 142131, 144051, 144171, 150219, 152855, 174135, 176175,
266042, 267962, 268730, 268922, 271802, 271994, 272762, 273077, 273197, 273959, 274682, 275117, 276647, 277007, 280021,
303677, 14246, 15206, 15254, 27446, 29366, 29479, 39734, 42806, 43699, 43819, 45494, 45686, 45739, 45854,
50643, 50894, 51662, 69433, 71353, 72121, 72313, 75193, 75385, 76153, 78073, 83783, 86471, 115031, 136572,
172603, 288823, 290863, 301111, 304183, 307231, 12197, 14179, 14227, 15077, 15187, 15245, 23347, 26419, 27822,
29107, 29299, 29467, 38709, 43693, 45429, 50637, 69430, 71350, 71470, 72118, 72310, 72478, 75190, 75382,
75550, 76150, 78070, 78190, 78238, 78430, 86590, 90430, 102718, 137011, 141995, 143731, 144151, 148939, 149191,
149959, 159775, 169015, 266037, 267642, 267957, 268077, 268725, 268917, 269085, 269607, 269847, 271797, 271989, 272157,
272757, 273687, 274677, 274797, 274845, 275037, 275559, 275607, 275727, 276567, 283197, 287037, 299325, 14166, 15901,
22326, 23221, 23341, 26293, 26413, 27253, 27303, 27421, 29046, 29101, 29293, 29341, 30781, 36659, 38579,
38699, 39347, 39539, 39598, 39707, 42419, 42611, 42670, 42779, 43379, 43438, 43630, 43678, 45299, 45419,
45467, 45659, 53819, 57659, 71033, 72487, 75559, 76567, 78247, 78439, 78487, 78607, 90679, 102967, 106807,
114895, 134396, 138767, 142607, 168251, 283703, 458815, 8037, 8085, 12003, 12117, 12171, 13797, 14157, 15051,
23211, 23662, 26283, 27051, 27243, 27291, 36533, 36653, 38573, 39341, 39533, 39581, 42413, 42605, 42653,
43253, 43373, 43421, 43613, 45293, 51773, 57533, 69294, 71030, 76014, 100542, 134835, 134955, 136875, 137643,
137835, 137883, 140715, 140907, 140955, 140967, 141555, 141675, 141723, 141735, 141915, 143595, 150075, 155835, 166959,
265466, 265901, 267637, 272621, 273551, 297149, 7910, 8078, 11750, 11990, 12110, 14023, 14791, 20150, 20263,
21941, 22183, 22301, 22901, 22951, 23198, 25973, 26215, 26270, 26870, 26983, 27038, 27230, 28903, 29021,
36278, 36470, 36523, 36638, 38259, 38318, 38510, 38558, 39158, 39278, 39326, 39518, 42230, 42350, 42398,
42590, 43243, 43358, 45278, 50750, 51518, 53438, 57470, 68857, 71447, 76431, 78167, 86327, 100911, 106671,
166075, 281647, 298015, 7651, 7891, 8011, 11731, 13771, 19891, 20083, 20251, 21931, 22123, 22171, 22771,
22891, 22939, 23127, 25518, 25843, 25963, 26007, 26203, 26971, 28891, 36213, 38133, 38253, 38301, 38493,
39261, 42333, 50493, 53373, 68854, 68974, 69022, 69214, 70894, 71902, 74974, 83134, 84094, 99454, 134515,
134935, 136435, 136555, 136603, 136795, 136807, 136855, 137563, 137575, 137623, 140635, 140887, 141655, 143575, 148795,
149047, 149815, 150559, 151675, 151735, 155767, 265461, 265581, 265629, 265821, 267501, 268509, 269391, 271581, 279741,
280701, 296061, 7638, 19830, 20125, 21750, 21863, 21918, 22110, 22765, 22878, 25837, 25950, 26845, 36083,
36203, 36251, 36443, 38123, 38238, 39131, 42203, 50363, 51323, 69031, 69223, 69271, 69391, 71311, 71911,
72079, 72271, 74983, 75151, 75343, 75991, 76111, 78031, 83503, 84151, 84271, 84511, 86191, 90223, 90271,
99511, 99631, 99871, 100471, 100639, 102511, 102559, 106591, 114751, 132092, 164987, 14903, 21358, 21851, 36077,
134379, 164535, 164655, 263162, 265127, 271079, 271247, 13871, 14639, 19687, 19805, 21725, 36062, 66553, 68951,
70871, 70991, 83063, 83231, 86111, 137423, 140495, 148655, 149615, 279223, 279343, 295351, 295543, 295711, 11835,
13623, 14523, 19182, 19675, 66550, 134359, 148255, 164215, 263157, 264039, 264087, 265047, 266727, 266967, 267087,
270807, 7742, 11581, 13501, 14462, 33779, 68815, 77887, 133007, 133967, 278903, 7455, 11447, 13435, 17902,
33773, 132075, 164079, 264911, 7343, 11375, 17383, 18877, 25663, 33758, 133751, 137279, 278767, 295135, 7287,
17371, 18045, 132055, 132535, 147679, 263631, 9181, 66511, 68671, 139583, 9143, 264767, 270527, 5023, 9071,
3003, 4983, 8955, 263487, 266367, 1981, 2941, 4847, 8702, 66367, 1887, 4603, 1783, 2527, 65791$\}$

\bigskip
\bigskip

$\cS_{10,7,4,3}=\{$38, 21, 11$\}$

\medskip

$\cS_{11,7,4,3}=\{$44, 74, 25, 134, 69, 35$\}$

\medskip

$\cS_{12,7,4,3}=\{$56, 84, 140, 146, 74, 273, 38, 521, 1029, 2051$\}$

\medskip

$\cS_{13,7,4,3}=\{$112, 168, 280, 292, 148, 546, 76, 1042, 1057, 529, 2058, 4102, 4105, 2053, 67$\}$

\medskip

$\cS_{14,11,4,3}=\{$38, 21, 11$\}$

\medskip

$\cS_{14,7,4,3}=\{$224, 336, 560, 584, 296, 1092, 152, 2084, 2114, 1058, 4116, 4161, 8204, 8210, 8225,
1041, 4106, 134, 2057, 261, 515$\}$

\medskip

$\cS_{15,11,4,3}=\{$44, 74, 25, 134, 69, 35$\}$

\medskip

$\cS_{15,7,4,3}=\{$448, 672, 1120, 1168, 592, 2184, 304, 4168, 4228, 2116, 8232, 8322, 16408, 16420, 16450,
16513, 2082, 8212, 8257, 268, 4114, 4129, 522, 2065, 1030, 1033, 517, 259$\}$

\medskip

$\cS_{16,13,4,3}=\{$38, 21, 11$\}$

\medskip

$\cS_{16,7,4,3}=\{$896, 1344, 2240, 2336, 1184, 4368, 608, 8336, 8456, 4232, 16464, 16644, 32816, 32840, 32900,
33026, 4164, 16424, 16514, 536, 8228, 8258, 1044, 4130, 2060, 2066, 1034, 518$\}$

\medskip

$\cS_{17,13,4,3}=\{$44, 74, 25, 134, 69, 35$\}$

\medskip

$\cS_{17,7,4,3}=\{$1792, 2688, 4480, 4672, 2368, 8736, 1216, 16672, 16912, 8464, 32928, 33288, 65632, 65680, 65800,
66052, 8328, 32848, 33028, 1072, 16456, 16516, 2088, 8260, 4120, 4132, 2068, 1036, 7$\}$

\medskip

$\cS_{18,13,4,3}=\{$56, 84, 140, 146, 74, 273, 38, 521, 1029, 2051$\}$

\medskip

$\cS_{18,7,4,3}=\{$3584, 5376, 8960, 9344, 4736, 17472, 2432, 33344, 33824, 16928, 65856, 66576, 131264, 131360, 131600,
132104, 16656, 65696, 66056, 2144, 32912, 33032, 4176, 16520, 8240, 8264, 4136, 2072, 38, 21,
11$\}$

\medskip

$\cS_{19,13,4,3}=\{$112, 168, 280, 292, 148, 546, 76, 1042, 1057, 529, 2058, 4102, 4105, 2053, 67$\}$

\medskip

$\cS_{19,7,4,3}=\{$7168, 10752, 17920, 18688, 9472, 34944, 4864, 66688, 67648, 33856, 131712, 133152, 262528, 262720, 263200,
264208, 33312, 131392, 132112, 4288, 65824, 66064, 8352, 33040, 16480, 16528, 8272, 4144, 44, 74,
25, 134, 69, 35$\}$

\medskip

$\cS_{12,8,4,4}=\{$3084, 780, 1546, 2314, 2566, 2569, 204, 1286, 1289, 1541, 3075, 2309, 170, 771, 105,
60, 90, 102, 150, 153, 165, 195, 85, 51, 15$\}$

\medskip

$\cS_{13,8,4,4}=\{$6168, 1560, 3092, 4628, 5132, 5138, 408, 2572, 2578, 3082, 6150, 4618, 212, 785, 1542,
120, 308, 332, 338, 1169, 1289, 172, 178, 202, 390, 649, 2129, 2309, 298, 4145,
4169, 4229, 4355, 581, 2089, 2179, 102, 1061, 1091, 547, 15$\}$

\medskip

$\cS_{14,8,4,4}=\{$12336, 3120, 6184, 9256, 10264, 10276, 816, 5144, 5156, 6164, 12300, 9236, 240, 424, 1570,
3084, 616, 664, 676, 2338, 2578, 2593, 344, 356, 404, 780, 1298, 1313, 1553, 4258,
4618, 12291, 596, 2321, 8290, 8338, 8353, 8458, 8710, 8713, 204, 1162, 3075, 4178, 4193,
4241, 4358, 4361, 4613, 2122, 2182, 2185, 8273, 8453, 771, 1094, 1097, 1157, 2117, 60,
195, 51, 15$\}$

\medskip

$\cS_{15,8,4,4}=\{$24672, 6240, 12368, 18512, 20528, 20552, 1632, 10288, 10312, 12328, 24600, 18472, 480, 848, 3140,
6168, 1232, 1328, 1352, 4676, 5156, 5186, 688, 712, 808, 1560, 2596, 2626, 3106, 8516,
9236, 9281, 24582, 1192, 4642, 16580, 16676, 16706, 16916, 16961, 17420, 17426, 17441, 408, 2324,
2369, 3089, 6150, 8356, 8386, 8482, 8716, 8722, 8737, 9226, 12293, 4244, 4289, 4364, 4370,
4385, 4625, 5129, 16546, 16906, 18437, 20483, 1542, 2188, 2194, 2209, 2314, 2569, 8465, 10243,
120, 4234, 16529, 16649, 390, 773, 8329, 1157, 1283, 643, 101, 86, 51, 46, 75,
29$\}$

\medskip

$\cS_{16,12,4,4}=\{$49164, 12300, 24586, 36874, 40966, 40969, 3084, 20486, 20489, 24581, 49155, 36869, 780, 1546, 12291,
2314, 2566, 2569, 204, 1286, 1289, 1541, 3075, 2309, 170, 771, 105, 60, 90, 102,
150, 153, 165, 195, 85, 51, 15$\}$

\medskip

$\cS_{17,12,4,4}=\{$98328, 24600, 49172, 73748, 81932, 81938, 6168, 40972, 40978, 49162, 98310, 73738, 1560, 3092, 12305,
24582, 4628, 5132, 5138, 18449, 20489, 408, 2572, 2578, 3082, 6150, 10249, 33809, 36869, 4618,
66065, 66569, 67589, 69635, 212, 1542, 9221, 33289, 34819, 120, 308, 332, 338, 16901, 17411,
172, 178, 202, 390, 8707, 298, 101, 15$\}$

\medskip

$\cS_{18,12,4,4}=\{$196656, 49200, 98344, 147496, 163864, 163876, 12336, 81944, 81956, 98324, 196620, 147476, 3120, 6184, 24610,
49164, 9256, 10264, 10276, 36898, 40978, 40993, 816, 5144, 5156, 6164, 12300, 20498, 20513, 24593,
67618, 73738, 196611, 9236, 36881, 132130, 133138, 133153, 135178, 139270, 139273, 240, 424, 3084, 18442,
49155, 66578, 66593, 67601, 69638, 69641, 73733, 616, 664, 676, 33802, 34822, 34825, 132113, 135173,
344, 356, 404, 780, 12291, 17414, 17417, 18437, 596, 33797, 204, 3075, 771, 60, 195,
51, 15$\}$

\medskip

$\cS_{19,12,4,4}=\{$393312, 98400, 196688, 294992, 327728, 327752, 24672, 163888, 163912, 196648, 393240, 294952, 6240, 12368, 49220,
98328, 18512, 20528, 20552, 73796, 81956, 81986, 1632, 10288, 10312, 12328, 24600, 40996, 41026, 49186,
135236, 147476, 147521, 393222, 18472, 73762, 264260, 266276, 266306, 270356, 270401, 278540, 278546, 278561, 480,
848, 6168, 36884, 36929, 49169, 98310, 133156, 133186, 135202, 139276, 139282, 139297, 147466, 196613, 1232,
1328, 1352, 67604, 67649, 69644, 69650, 69665, 73745, 81929, 264226, 270346, 294917, 327683, 688, 712,
808, 1560, 24582, 34828, 34834, 34849, 36874, 40969, 135185, 163843, 1192, 67594, 264209, 266249, 408,
6150, 12293, 133129, 18437, 20483, 1542, 10243, 120, 390, 773, 1157, 1283, 643, 101, 86,
51, 46, 75, 29$\}$

\end{document}